\documentclass[a4paper,twocolumn,11pt,accepted=2025-04-09,titlepage]{quantumarticle}
\pdfoutput=1
\usepackage{amsmath}
\usepackage{amssymb}
\usepackage{hyphenat}

\usepackage{amsthm}
\usepackage{natbib}
\setcitestyle{numbers,open={[},close={]},comma}
\usepackage{color}
\usepackage{bbm}
\usepackage{bm}
\usepackage{upgreek}
\usepackage{graphicx}
\usepackage{mathrsfs}
\usepackage{enumitem}
\usepackage{footnotebackref}
\usepackage{hyperref}
\usepackage{pdftexcmds}

\DeclareFontFamily{U}{cbgreek}{}
\DeclareFontShape{U}{cbgreek}{m}{n}{
        <-6>    grmn0500
        <6-7>   grmn0600
        <7-8>   grmn0700
        <8-9>   grmn0800
        <9-10>  grmn0900
        <10-12> grmn1000
        <12-17> grmn1200
        <17->   grmn1728
      }{}
\DeclareFontShape{U}{cbgreek}{bx}{n}{
        <-6>    grxn0500
        <6-7>   grxn0600
        <7-8>   grxn0700
        <8-9>   grxn0800
        <9-10>  grxn0900
        <10-12> grxn1000
        <12-17> grxn1200
        <17->   grxn1728
      }{}

\makeatletter
\newcommand{\normalorbold}{%
  \ifnum\pdf@strcmp{\math@version}{bold}=\z@ bx\else m\fi
}
\makeatother

\theoremstyle{plain}
\newtheorem{ithm}{Theorem}
\newtheorem{theorem}{Theorem}
\newtheorem{athm}{Theorem}[section]

\newtheorem{iconjm}{Conjecture}
\newtheorem{conjm}{Conjecture}
\newtheorem{obs}[athm]{Observation}
\newtheorem{prop}[athm]{Proposition}
\newtheorem{alem}[athm]{Lemma}
\newtheorem{coro}{Corollary}[athm]

\theoremstyle{definition}
\newtheorem{definition}{Definition}
\newtheorem{defs}[athm]{Definition}
\newtheorem{remm}[athm]{Remark}

\newcommand*{\eins}{\ensuremath{\mathbbm 1}}
\def\gbm#1{{\let\alpha\upalpha \let\phi\upphi \let\lambda\uplambda \let\mu\upmu \let\rho\uprho \let\sigma\upsigma \let\tau\uptau \let\theta\uptheta \let\eta\upeta \let\xi\upxi \bm{#1}}}

\newcommand*{\bbC}{\mathbb{C}}
\newcommand*{\bbR}{\mathbb{R}}

\newcommand*{\bbZ}{\mathbb{Z}}

\newcommand*{\cA}{\mathcal{A}}

\newcommand*{\cC}{\mathcal{C}}
\newcommand*{\cD}{\mathcal{D}}
\newcommand*{\cE}{\mathcal{E}}
\newcommand*{\cF}{\mathcal{F}}
\newcommand*{\cG}{\mathcal{G}}
\newcommand*{\cH}{\mathcal{H}}
\newcommand*{\cI}{\mathcal{I}}

\newcommand*{\cL}{\mathcal{L}}
\newcommand*{\cM}{\mathcal{M}}

\newcommand*{\cP}{\mathcal{P}}

\newcommand*{\cS}{\mathcal{S}}
\newcommand*{\cT}{\mathcal{T}}
\newcommand*{\cU}{\mathcal{U}}
\newcommand*{\cV}{\mathcal{V}}
\newcommand*{\cW}{\mathcal{W}}
\newcommand*{\cX}{\mathcal{X}}

\newcommand*{\sI}{\mathscr{I}}

\newcommand*{\sT}{\mathscr{T}}

\newcommand*{\sV}{\mathscr{V}}
\newcommand*{\sW}{\mathscr{W}}

\newcommand*{\rA}{\mathrm{A}}
\newcommand*{\rB}{\mathrm{B}}
\newcommand*{\rC}{\mathrm{C}}

\newcommand*{\rJ}{\mathrm{J}}
\newcommand*{\rM}{\mathrm{M}}
\newcommand*{\rR}{\mathrm{R}}
\newcommand*{\rS}{\mathrm{S}}

\newcommand*{\ket}[1]{\left|#1\right\rangle}
\newcommand*{\bra}[1]{\left\langle #1\right|}
\newcommand*{\bket}[2]{\bra{#1\vphantom{#1 #2}}\left.#2\vphantom{#1 #2}\right\rangle}
\newcommand*{\kbra}[2]{\ket{#1\vphantom{#1 #2}}\bra{#2\vphantom{#1 #2}}}
\newcommand*{\bkt}[1]{\bket{#1}{#1}}
\newcommand*{\proj}[1]{\kbra{#1}{#1}}
\newcommand*{\tr}{\mathrm{Tr}}
\newcommand{\nrm}[1]{\left\|#1\right\|}
\newcommand{\abs}[1]{\left|#1\right|}
\newcommand{\vecg}[1]{\gbm{#1}}
\newcommand{\vect}[1]{\mathbf{#1}}
\newcommand{\n}{\textendash}
\newcommand{\m}{\textemdash}

\hypersetup{
    pdftoolbar=true,        
    pdfmenubar=true,        
    pdffitwindow=false,     
    pdfstartview={FitH},    
    colorlinks=true
}

\begin{document}
\title{The coherent measurement cost of coherence distillation}
\date{10 April 2025}
\author{Varun Narasimhachar}
\email{varun.achar@gmail.com}
\orcid{0000-0001-8155-6120}
\affiliation{Institute of High Performance Computing (IHPC),
Agency for Science, Technology and Research (A*STAR), 1 Fusionopolis Way, Republic of Singapore 138632}

\begin{abstract}
    Quantum coherence\m an indispensable resource for quantum technologies\m is known to be distillable from a noisy form using operations that cannot create it. However, distillation exacts a hidden coherent \emph{measurement} cost, which has not previously been examined. We devise the \emph{target effect} construction to characterize this cost through detailed conditions on the coherence-measuring structure necessary in any process realizing exact (maximal or non-maximal) or approximate distillation. As a corollary, we lower-bound the requisite measurement coherence, as quantified by operationally-relevant measures. We then consider the asymptotic limit of distilling from many copies of a given noisy coherent state, where we offer rigorous arguments to support the conjecture that the (necessary and sufficient) coherent measurement cost scales extensively in the number of copies. We also show that this cost is no smaller than the coherence of measurements saturating the scaling law in the generalized quantum Stein's lemma. Our results and conjectures apply to any task whereof coherence distillation is an incidental outcome (e.g.,\ incoherent randomness extraction). But if pure coherence is the only desired outcome, our conjectures would have the cautionary implication that the measurement cost is often higher than the distilled yield, in which case coherence should rather be prepared afresh than distilled from a noisy input.
\end{abstract}

\maketitle


\section{Introduction}
Coherence is a cornerstone of quantum mechanics, playing a central role in the wavelike interference effects that epitomize quantum phenomena. It is also a valuable resource, powering transformative quantum technologies such as quantum computing, quantum communication, and quantum metrology \cite{SAP17}.

A central concept in quantum information science is that of \emph{resource distillation}: the conversion of a resource (like coherence) from an impure form into a standard, pure form. The distillation protocol must not itself consume or generate the resource under question\m it must just convert the resource to the desired form. Under the formalism of quantum processes, this condition is not unambiguous. Depending on how it is interpreted, it begets formal conditions of varying strength. In the case of coherence, the least-constrained distillation protocols are \emph{coherence\n non-creating channels} \cite{aaberg2004subspace,aberg2006quantifying}, established in the literature as the \emph{maximal incoherent operations} (MIO). When acting on incoherent input states, MIO produce outputs that are also incoherent.

But this condition is so weak that it may yet admit operationally-unreasonable possibilities, such as a process that amplifies certain aspects of coherence in an already-coherent input. An operational subclass of MIO, called the \emph{incoherent operations} (IO) \cite{baumgratz2014quantifying}, mitigate this incongruity by placing the additional constraint that the process must be implementable through elementary sub-processes (formally, Kraus operators) that are \emph{each} coherence\n non-creating. Winter and Yang \cite{WY16} found that the IO subclass is already powerful enough to distill coherence maximally, i.e.\ as efficiently as MIO.

Moving further in the operationalist direction, Yadin \textit{et al.}\ \cite{YMG+16} defined a subclass of IO called the \emph{strictly incoherent operations} (SIO). These processes are implementable through sub-processes that individually \emph{neither create nor detect} (i.e.,\ measure) coherence. Lami \textit{et al.}\ \cite{lami_2019-1} found that this constraint is too strong to admit full distillation; subsequently, Lami \cite{lami_2019} exactly quantified how efficiently SIO can distill.

It is evident from these developments that the coherence-measuring capability of the components implementing IO is essential to enabling maximal distillation. But there has been no investigation yet on exactly \emph{how much} coherence-measuring power is required. This is the question that motivates our work.

We approach the problem by constructing a measurement that we call the \emph{target effect} of a process. A process' target effect captures the probability with which it succeeds in converting an arbitrary input to a desired standard form of the resource\m in other words, the process' efficacy for the purpose of distillation. The target effect also contains significant information about the structure of coherent measurements required to implement the process. Our foundational lemmas put stringent conditions on the structure of the target effect for the associated process to succeed in various instances of distillation\m exact (maximal and non-maximal) and approximate. Based on these lemmas, our main results provide lower bounds on the requisite coherent measurement cost (formalized in the relevant technical section below) of the respective distillation instances. Our results apply not only to IO, but to the most general operational class MIO.

We then consider the so-called \emph{asymptotic limit} of distillation, where the input is a large number of copies of a given resource state. Based on our results on approximate distillation, we conjecture that the necessary and sufficient coherent measurement cost of maximal asymptotic distillation is what we call the input's \emph{irretrievable coherence}\m a quantity related to the irreversibility of resource conversion under IO. We make substantial progress towards proving this conjecture and discuss the technical difficulties that hinder the proof. In the process, we establish a relationship between our problem and an important quantum-statistical result called the generalized quantum Stein's lemma \cite{hayashi2024generalized,lami2024solution}: the requisite measurement coherence of asymptotically maximal distillation is bounded below by the coherence of hypothesis-testing measurement effects that saturate the error exponent scaling bound set by the lemma. Finally, we discuss the implications of our conjectures: (1) the asymptotic coherent measurement cost may outweigh the very distilled yield in a significant fraction of cases\m thus rendering distillation wasteful; (2) possibilities of a tradeoff between the coherent measurement budget and the distilled yield.

In the rest of this section, we will provide a brief background and motivation for our research problem, followed by a summary of our main results. We strive to keep the discussion high-level; but it will, of necessity, get gradually more technical by the end of the section. We will often use some technical terms and notation without interrupting the flow of presentation with their formal definitions, which we defer until the next section.

\subsection{Coherence relative to bases}
Quantum coherence refers to the presence of superposition effects in quantum mechanics. But of course, this begs the question ``a superposition of what sort of entities?'' Indeed, coherence can be given different formal definitions based on what we consider the ``unsuperposed'' objects; examples include eigenstates of conserved quantities \cite{MS16}, orthogonal subspaces induced by measurements \cite{aberg2006quantifying,BKB21}, and even mutually-nonorthogonal elements, such as the classical states of bosonic modes \cite{TKEP17}. For our purposes it suffices to consider the unsuperposed objects to be mutually-orthogonal subspaces of the Hilbert space of quantum states, which we shall for brevity call elementary subspaces.

Within this notion, there is a broad distinction between two subtypes, according to what happens when two systems $\rA$ and $\rB$ are brought together:
\begin{enumerate}
    \item The elementary subspaces $\cV^{\rA\rB}$ of the composite $\rA\rB$ are all and only the pairwise tensor products of those of $\rA$ and $\rB$: $\cV_{j,k}^{\rA\rB}=\cV_{j}^\rA\otimes\cV_{k}^\rB$. Under this composition rule, two distinct elementary subspaces $\cV_{j_1}^\rA$, $\cV_{j_2}^\rA$ of the same subsystem never contribute to the same elementary subspace of the composite.
    \item The elementary subspaces of $\rA\rB$ depend on those of $\rA$ and $\rB$ in a manner different from the simple tensor-product composition. An example of this is when the elementary subspaces are eigenspaces of operators (e.g.,\ local Hamiltonians): $\rA\rB$ may have degenerate eigenspaces that combine distinct $\cV_{j_1}^\rA,\cV_{j_2}^\rA$ (in the Hamiltonian example, when the energies of these eigenspaces are related to those of a pair of $\rB$ eigenspaces through $E_{j_1}^\rA+E_{k_1}^\rB=E_{j_2}^\rA+E_{k_2}^\rB$).
\end{enumerate}
We will restrict our attention to the first case, where the composition of two or more systems is governed by a simple tensor product. Furthermore, we will only consider the case of one-dimensional elementary subspaces; it is evident that this dimensionality is preserved under tensor-product composition. In other words, the coherence of interest to us is relative to a certain fixed \emph{orthogonal basis} on each system, and the corresponding tensor-product bases of compositions thereof. This special case has been studied extensively under the resource-theoretic paradigm, which we will review in a later section.

\subsection{Operational aspects of coherence}\label{cohasp}
The \emph{static} form of coherence, present in quantum states, is usually most familiar and easily understood: for instance, a quantum bit (qubit) in one of the states $\ket\pm:=\left(\ket0\pm\ket1\right)/\sqrt2$ contains coherence relative to the so-called computational basis $\left\{\ket0,\ket1\right\}$. But there is also a \emph{dynamical} aspect to coherence, manifested in the action of processes that can transform an incoherent state to a coherent one. In addition, a \emph{mensural} aspect of coherence is embodied by the ability of a measurement device to detect coherent superpositions. An elementary prototype of dynamical coherence is the unitary process effected by the qubit Hadamard gate, given by its action $H\ket0=\ket+$ and $H\ket1=\left(\ket0-\ket1\right)/\sqrt2$ on the computational basis states. In turn, the resulting so-called Hadamard basis states $\ket\pm$ can be considered elementary prototypes of static coherence, and a measurement in this basis one of measurement coherence.

Under the tensor-product composition model within which we are working, these three aspects of coherence are not equal in their operational power: indeed, dynamical coherence is strictly more powerful than the other two, as it can be used to elevate both incoherent states and incoherent measurements to their respective coherent counterparts. Meanwhile, static coherence cannot be used to simulate either dynamical or measurement coherence, and measurement coherence is likewise restricted. Note that this restriction may not present itself under other composition rules\m notably, when the elementary subspaces are eigenspaces of conserved quantities, the composition of two or more systems can admit static coherence to be turned dynamical \cite{aaberg2014catalytic}.

Some important clarifications are in order regarding our use of the term ``coherent measurement''. Firstly, by ``measurement'', we do not mean necessarily an action \emph{intended} as a measurement, or one accompanied by a readout or wavefunction collapse. Rather, we are referring to a property on the mathematical level of description, which on the operational level corresponds to elementary features of the interaction between the system and the apparatus carrying out a process thereon; these elementary interactions necessarily cause certain pieces of information about the system's state to impinge on the apparatus'\m\emph{effectively} a measurement of the system by the apparatus. Furthermore, when we say a process requires coherent measurement over $k$ levels of the system, we mean that at least some of the information inevitably imprinted on the apparatus can distinguish $k$-fold superposed states of the system from other states. In other words, regardless of whether an agent sets up the apparatus for the express purpose of such measurement, carrying out the process entails the apparatus \emph{effectively} measuring (in the above sense) relative to $k$-fold superpositions over the system's computational states. This is a distinctly ``measurement-flavoured'' coherent capability that the apparatus \emph{must} possess for carrying out the process, in contrast to static and dynamical coherence. Consider the Hadamard basis measurement mentioned above (with or without readout): one way of implementing it is to simply project the input onto the Hadamard basis states (corresponding to the Kraus operators $\proj\pm$); in this case, the measuring process has dynamical coherence, since it can (probabilistically) map incoherent inputs to coherent outputs. But another way to implement the same measurement is to use incoherent ``flags'' for the two outcomes (e.g.,\ via the Kraus operators $\kbra{0}{+}$ and $\kbra{1}{-}$), in which case the process lacks dynamical coherence. Measurement coherence is that operational capability which the process possesses in both cases alike.

Just as the static coherence consumed by a process is considered a cost, so should the requisite measurement coherence for a task (such as coherence distillation), even though we do not usually think of measurement as a ``consumable''. A spiritually kindred concept is that of \emph{query complexity}: the number of times a certain operational element needs to be activated in an algorithm. A further justification for considering it a cost is that both static and measurement coherence can be derived from dynamical coherence\m indeed, they always are in practice, though their operational distinction from dynamical coherence is an interesting feature of the quantum formalism (this point will be discussed at length in Section \ref{secram}).

Before moving on to the next section, we note that our notion of measurement coherence is closely related to that formalized by Kim and Lee \cite{kim2022relation,kim2024maneuvering}. But their notion is defined on POVMs\m a formalization that does not incorporate the dynamical details of what happens to the system post-measurement. Also, POVMs are used in contexts with operational \emph{intent} to measure (in contrast with ours, where operational \emph{capability} takes precedence over intent). There are also differences in the finer technical details of Kim and Lee's treatment besides these. Another notion in the literature distinct from ours is \emph{coherence relative to measurements}, also called POVM coherence \cite{aberg2006quantifying,BKB21}: while we are concerned with the coherence (relative to the fixed incoherent basis) \emph{within} the effective measurement action in the Kraus operators, this notion of POVM coherence focuses on coherence \emph{relative to} the block structure induced by the POVM effects themselves. This is, in a sense, complementary to our notion and that of Kim and Lee.

\subsection{Resource theories of coherence}
Various formalizations of the concept of coherence have been explored under the broad umbrella of \emph{resource theories}; for a general exposition on coherence resource theories, we direct the reader to \cite{SAP17}. Here we will provide a brief introduction to certain specific resource theories of coherence, adequate for our purposes.

A resource theory formalizes the study of a quantum resource by identifying the operational capabilities required to create or proliferate it. Such capabilities are axiomatically forbidden, leaving only certain constrained actions that can be performed, called the ``free operations''. The theory then endeavours to chart out what can and cannot be done using only the free operations\m spiritually akin to determining, say, the plane figures that can be constructed using only a compass and a straightedge\footnote{We are indebted to Gilad Gour for this evocative analogy.}. Typically, the free operations in quantum resource theories are a class of quantum (sub)channels, including the preparation of so-called ``free states''\footnote{We use the term ``states'' to refer not only to pure states (associated with wavefunctions spanning a Hilbert space), but also to mixed states (formally identified with density operators acting on the Hilbert space).}. All states that are not free are called resource states.

In the resource theories we will consider, the resource is coherence relative to a fixed orthogonal basis of the Hilbert space of a given quantum system
. This basis is variously termed computational, canonical, classical, etc.;\ we will simply call it the \emph{incoherent basis}. Furthermore, we will work in the paradigm where the incoherent basis of a composite system is just the tensor product of its subsystems' incoherent bases. This notion of coherence falls under what has been called \emph{speakable coherence} in the literature \cite{MS16}. It is operationally relevant for, e.g.,\ gate-based quantum computing, where every elementary system has a computational basis and where tensor products of computational-basis states are easy to prepare.

The free states in these coherence resource theories are the incoherent states, i.e.\ the states whose density matrices are diagonal in the incoherent basis. We will refer to all other states\m the resource states\m as \emph{coherent}. Most of the extensively studied resource theories are primarily concerned with static resources, and this is true also of coherence resource theories. Nevertheless, dynamical and mensural considerations do play a part in axiomatizing the class of free operations.

Adhering to the basic tenets of the resource-theoretic paradigm, the free operations in all of these resource theories are constrained to be incapable of creating coherence. But as it turns out, there are diverse ways to choose families of free operations obeying this constraint, spawning a veritable zoo of distinct coherence resource theories. Amongst these, we will focus on the resource theory whose free operations are the so-called \emph{incoherent operations}\footnote{We will adhere to the term ``incoherent operations'' established in the literature, notwithstanding its regrettable inspecificity.} (IO) \cite{baumgratz2014quantifying}. Informally, an IO is a quantum process that can be broken down into sub-processes that \emph{may detect} coherence (i.e.,\ measure relative to coherent states, in the sense explained in Section \ref{cohasp}) but \emph{must not create} coherence when acting on incoherent input states. In other words, each of the sub-processes that constitute the process is devoid of dynamical coherence but may contain measurement coherence.

\subsection{Distillation of coherence-resource}
The main motivation for this work comes from the resource-theoretic concept of \emph{distillation}: the task of converting an arbitrary resource state to a standard form. Resource distillation is often an essential part of applications \cite{chitambar2019quantum}; for example, coherence distillation is closely related to the task of randomness extraction using incoherent measurements \cite{HFW21}. Beyond this direct value, the study of resource distillation also offers valuable insight into the structure of a resource theory.

In all resource theories of coherence, the standard form of coherence-resource is a pure state containing a uniform superposition of some number of incoherent basis elements, e.g.\ $\ket{\Psi_M}:=M^{-1/2}\sum_{m\in[M]}\ket m$. This choice is justified by several factors. Firstly, these states are usually optimal for applications (e.g.,\ phase estimation) requiring coherence. Secondly, the free operations can produce any required state from a single copy of one of these. A third justification comes from the \emph{asymptotic} or \emph{independent and identically-distributed} (i.i.d.)\ limit of the resource theory, where the free operations act on large numbers of independent copies of identical states. In this limit, the standard coherent states of different $M$ values admit reversible\footnote{Note that this reversibility is to leading order in the number of copies; in this work we will not consider higher-order effects \cite{korzekwa2019avoiding}.} ``currency exchange'' at a rate proportional to $\log_2M$, which is the equivalent number of standard coherent bits (or \emph{cobits}) $\ket{\Psi_2}$ (more familiar as the Hadamard state $\ket+$ discussed above). The cobit thus functions as a convenient unit for quantifying coherence. For the same reason, the Hadamard gate makes a good standard unit for dynamical coherence, and the Hadamard measurement one for measurement coherence.

The asymptotic limit of the IO resource theory affords an added feature: copies of \emph{any} coherent state\m pure or mixed\m can be converted (albeit \emph{not} reversibly\m more on this later) by IO to cobits at a rate that is maximal in a resource-theoretic sense \cite{BG15}. In other words, coherence is asymptotically \emph{universally distillable} by IO. But at the heart of this universal distillability lies the central question that motivates our work.

\subsection{What powers coherence distillation?}
Recall that IO can be implemented using components that do not create coherence, but may nevertheless detect it. \emph{Strictly incoherent operations} (SIO) are the sub-class of IO that use only components that \emph{cannot even detect} coherence \cite{YMG+16}. This restriction ends up breaking the asymptotic universal distillability seen under IO. Indeed, SIO exhibit a particularly severe form of non-distillable, or ``bound'', coherence: any number\m however large\m of copies of certain coherent states cannot be converted, even approximately, to even a \emph{single} cobit \cite{lami_2019-1}.

In summary, the \emph{unbounded} measurement coherence of IO enables universal distillation, while the strictly coherence\n non-detecting SIO are too constrained to distill universally. But what lies between these two extremes? Our paper is an attempt to understand this intervening operational landscape, by answering questions such as:
\begin{enumerate}
    \item How much measurement coherence (quantified in a way that will be discussed later) is necessary to recover the maximal distillability afforded by IO?
    \item What are the corresponding costs for non-maximal and approximate distillation?
    \item How does this coherent measurement cost behave in the asymptotic limit?
    \item What are the conditions for the lower bound on the cost to be attainable? Are these conditions met in the asymptotic limit?
\end{enumerate}
We approach these questions using a construction that we call the \emph{target effect}: a measurement associated with a given quantum process, containing information about the efficacy of the process at mapping arbitrary inputs to a desired target output, as well as about the coherence-detecting power of the process. Among other things, we show that the coherent measurement cost of distillation is bounded by a particular property of the target effect, which we will now discuss.

\subsection{Irretrievable coherence}
In a resource theory, a real-valued function of states is called a \emph{resource measure} if it satisfies the following two conditions: (1) It is a non-increasing monotone under the free operations; (2) It is \emph{faithful}, i.e.\ takes nonzero values on all and only the resource (non-free) states. The answers to our central questions turn out to involve some important measures of coherence.

Given a state $\rho$, its \emph{relative entropy of coherence} is defined as
\begin{equation}
    C_r(\rho)=\min\left\{S\left(\rho\|\sigma\right):\;\Delta[\sigma]=\sigma\right\},
\end{equation}
where $\sigma$ takes values of density operators, $S\left(\cdot\|\cdot\right)$ is the Umegaki quantum relative entropy \cite{umegaki1959conditional}, and $\Delta(\cdot)$ denotes the diagonal part (in the incoherent basis representation) of the argument. Thus, the minimization is over all diagonal states $\sigma$\m in other words, the free states. Conveniently, the minimization evaluates to $C_r(\rho)=S\left(\rho\|\Delta[\rho]\right)=S\left(\Delta[\rho]\right)-S(\rho)$, where $S(\cdot)$ is the von Neumann entropy. Meanwhile, $\rho$'s \emph{coherence of formation} is given by the so-called \emph{convex-roof extension} of the restriction of $C_r$ to pure states:
\begin{equation}\label{cof}
    C_f(\rho)=\min_{p_x\ge0;\;\sum_xp_x\proj{\phi_x}=\rho}\sum_xp_xC_r\left(\proj{\phi_x}\right),
\end{equation}
where the minimization is over all convex decompositions of $\rho$ into pure states. Notice that $C_r(\psi)=C_f(\psi)=S\left(\Delta[\psi]\right)=H\left(\vect p\right)$ (where $H$ is the Shannon entropy) for a pure state $\psi\equiv\proj{\psi}$\footnote{We will use this shorthand for rank-1 projectors, where it is possible without ambiguity.} with incoherent-basis distribution $\abs{\bket i\psi}^2=p_i$. In particular, $C_r\left(\Psi_M\right)=C_f\left(\Psi_M\right)=\log_2M$ for the standard resources.

These measures have operational significance in the IO resource theory. Firstly, $C_r(\rho)$ is the regularized asymptotic \emph{distillable coherence} under IO, defined as the maximum asymptotic rate at which cobits can be distilled from copies of $\rho$ by IO. That is, $C_r(\rho)$ is the largest $r\in\bbR$ such that the transformation $\rho^{\otimes n}\mapsto\Psi_2^{\otimes rn}$ can be achieved by IO to an arbitrarily good approximation as $n\to\infty$. Likewise, $C_f(\rho)$ is the regularized asymptotic \emph{coherence cost} under IO: the minimum asymptotic rate at which cobits must be \emph{consumed} to \emph{prepare} copies of $\rho$ by IO, in an operational task called \emph{resource dilution} or \emph{formation}\m the opposite of distillation. Mathematically, $C_f(\rho)$ is the smallest $r\in\bbR$ such that the transformation $\Psi_2^{\otimes rn}\mapsto\rho^{\otimes n}$ can be achieved arbitrarily well as $n\to\infty$.

For almost all states $\rho$ (in a measure-theoretic sense), $C_f$ is strictly larger than $C_r$ \cite{WY16}. Hence, the coherence distillable by IO from a given input is generically smaller than that required to prepare the same input. Thus, IO presents an instance of \emph{irreversibility} in resource theories, a topic of current interest \cite{LR23,LRS23}. IO's is a particularly strong form of irreversibility, since it persists even in the asymptotic limit and is, moreover, already present in the lowest order (i.e.,\ between the regularized distillation and formation rates). Nevertheless, as we alluded to above, the IO-\emph{distillable} coherence is in fact maximal under general resource-theoretic constraints \cite{BG15}; therefore, the culprit behind the irreversibility is the inflated \emph{cost} of resource formation under IO. Incidentally (as the reader may have anticipated from our earlier statements), the SIO resource theory is even more irreversible\m this additional disparity owing solely to SIO's inferior \emph{distillable} coherence compared to IO's, the two theories' coherence costs being equal!

In our main results, these two coherence measures feature in the form of their difference $\ell(\rho):=C_f(\rho)-C_r(\rho)$. Because of the operational meaning of this quantity vis-\`a-vis the asymptotic irreversibility of the IO resource theory, we christen it the \emph{irretrievable coherence}. It has, in fact, been encountered (though not named) in the literature in a different operational context: it quantifies the difference between the quantum and the classical values of the so-called \emph{intrinsic randomness} of a state \cite{YZCM15,YZGM19}. It is worth noting that, while the irretrievable coherence is determined by two coherence measures and is itself a \emph{signature} of coherence\m it is nonzero only for coherent states\m it is \emph{not} a coherence measure in the resource-theoretic sense. It fails to be faithful, as can be seen from the case of pure states. But more importantly, it is not a monotone under IO or, indeed, any reasonable class of free operations.

\subsection{Clues in the literature}\label{secclue}
Recall that distillation in the asymptotic limit is the task of converting many copies of a given input to an output close to a standard resource. Formally, for every $n\in\bbZ_+$, the input is $\varrho_n\equiv\rho^{\otimes n}$ and is to be mapped approximately to $\Psi_2^{\otimes m_n}$ for some $m_n$. ``Asymptotic'' refers to the limit $n\to\infty$, and the asymptotic rate of distillation is the value  $r=\lim_{n\to\infty}\left(m_n/n\right)$. As we alluded to earlier, the highest achievable rate for a given $\rho$ is $r=C_r(\rho)$.

Winter and Yang \cite{WY16} constructed an IO protocol achieving this maximal distillation rate. A high-level examination of the protocol already hints at connections between asymptotic irreversibility and the object of our interest, viz.\ the coherent measurement cost of distillation. Crucially, the protocol consists (apart from some asymptotically-inconsequential measurements) of just a unitary transformation of the input followed by a partial trace. Considering the purity required of the output (distillate), the effect of the protocol before the final partial trace can be summarized approximately as
\begin{equation}\label{hintst}
    \varrho_n^\rA\stackrel\cU\longmapsto\tau^\rS\otimes{\Psi_M}^\rM.
\end{equation}
Here the superscript $\rA$ labels the input system, $\rM$ the output system, and $\rS$ the part that will be traced out. The question of how much coherent measurement the IO needs translates to how coherently this unitary channel $\cU$ must act. Since the unitary does not involve any additional systems, the systems' dimensionalities (which we denote by italicizing the corresponding labels) satisfy $A=SM$. Let us now make some heuristic estimates for these numbers, appealing to (an extremely crude form of) asymptotic typicality \cite{wilde2013quantum}; for brevity, we will omit qualifiers like ``approximate'' and ``typical part'' in the following statements, but stress that these qualifications are implicit.

Consider the input $\varrho_n\equiv\rho^{\otimes n}$: its rank (by unitarity, also the rank of $\tau$) is\footnote{Throughout this paper, we will use the notation $\exp_2(\cdot)\equiv2^{(\cdot)}.$} $S_0:=\exp_2\left[n\,S(\rho)\right]$, due to the asymptotic equipartition property (AEP). Applying AEP on the diagonal part $\Delta\left(\varrho_n\right)$, which is in fact $\left[\Delta(\rho)\right]^{\otimes n}$, we see that the relevant dimensionality of the input (covering all the incoherent basis labels that occur with nonzero amplitudes) is $A=\exp_2\left[n\,S\left(\Delta[\rho]\right)\right]$. Finally, the size of the maximal distillate is $M=\exp_2\left[n\,C_r(\rho)\right]=\exp_2\left[n\left(S\left[\Delta(\rho)\right]-S[\rho]\right)\right]$. Notice that $M=A/S_0$, and therefore, $S_0=S$. A basic consequence of this is that coherent measurements of rank $S$ are sufficient for maximal distillation; this therefore serves as a reference value against which to evaluate nontrivial bounds on the coherent measurement cost, which is our main concern.

Noting (again from equipartition) that the input's spectrum must be flat, we conclude that $\tau$ must be maximally mixed. In particular, this means that the subsystem $\rS$ is discarded in an incoherent state, uncorrelated with the distillate $\rM$. Now let us view the protocol in reverse: we take $\Psi_M$, append to it an auxiliary system $\rS$ in the \emph{incoherent} state $\tau$, and apply the unitary channel $\cU^\dagger$ to map the composite $\rS\rM$ to $\varrho_n^{\rA\equiv\rS\rM}$. How much coherence must $\cU^\dagger$ generate for this? If it had had to act on a fully incoherent input, it would have needed to create all of the coherence in $\varrho_n$ from scratch; considering that it has the $\Psi_M$ to begin with, it still needs to account for the deficit\m a hint at the difference between the formation cost and the distillable coherence.

To be sure, how much coherence an operation needs to generate to prepare a required resource is not an unambiguous concept: it depends on the class of operations considered. The Winter\n Yang IO protocol's reversal, which we considered, is not itself an IO; nor is it in any of the other classes of incoherent operations defined in the literature. Besides, this maximal distillation protocol is but one possibility; in general, a protocol may use auxiliary systems, instead of acting unitarily on just the input. In this connection, a further difficulty is that there is not yet an operational understanding of IO in terms of their unitary implementations (or \emph{dilations}): IO are only understood in terms of the abstract mathematical objects called Kraus operators. In contrast, SIO (for example) are understood through both their Kraus operators and their dilations.

But another hint in the same direction\m again from the near\n maximally-mixed fate of $\rS$ and the near-unitarity of the protocol\m is that IO distillation seems to entail fully ``resolving'' the space on which the input is supported, so as to retain in the distillate not only all of the input's coherence but also all of its \emph{purity}. As such, the requisite coherent measurement of distillation translates to that of fully resolving the input's support. As a simple illustration of what we mean, consider the 2-dimensional subspace of a ququart (4-dimensional system) spanned by the vectors
\begin{align}
    \ket{v_0}\propto\ket++\ket2;\;\ket{v_1}\propto\ket-+\ket3,
\end{align}
where $\ket\pm\propto\ket0\pm\ket1$ as usual. An IO can distill one cobit from this subspace, e.g.\ using the Kraus operators $K_0=\kbra0++\kbra12$ and $K_1=\kbra0-+\kbra13$. But to do so it must necessarily act coherently on the $\ket\pm$ part.

Let us now try to estimate the coherent measurement rank required to resolve (in this sense) the support of $\varrho_n$. Asymptotically, the projector onto this support is close to $\varrho_n$ itself, as the latter's spectrum flattens out. Let $\varrho_n=\sum_jq_j\phi_j$ be some convex decomposition into pure components. Now consider the typical letter strings (i.e.,\ incoherent basis indices) that occur in $\varrho_n$ (or, equivalently, along its diagonal). As noted above, these are $A=\exp_2\left[n\,S\left(\Delta[\rho]\right)\right]$ in number; and by asymptotic typicality, each of them carries a weight of $A^{-1}$ in the overall input. On the other hand, the weight that any of these strings accrues by virtue of its occurrence in any single $\phi_j$ is $\lesssim\left(S\exp_2\left[nC_f(\rho)\right]\right)^{-1}$: the $S$ factor is the dimensionality of the support, while $\exp_2\left[nC_f(\rho)\right]$ lower-bounds the number of (approximately equally-superposed) typical strings in each pure component, as can be seen from the definition \eqref{cof}. Thus, in order to account for all of a string's weight in $\varrho_n$, it must occur in no less than
\begin{equation}
\frac{S2^{nC_f(\rho)}}A=\exp_2\left(n\left[C_f(\rho)-C_r(\rho)\right]\right)=2^{n\ell(\rho)}
\end{equation}
distinct $\phi_j$'s. This bound applies alike to all of the typical strings. Therefore, a measurement that resolves the entire support of $\varrho_n$ can be expected to involve coherence over blocks no smaller than this size.

Though these hints are based on loose intuition and crude estimates, they proved helpful in our project, directing us towards more rigorous investigations and methods. In particular, they inspired us to consider a non-asymptotic idealization of the above ``crudely-typicalized'' case of maximal distillation, yielding a result (Theorem \ref{itfinmax} below) that somewhat validates the hints and informs our conjectures \ref{ipan} and \ref{ipas} on the asymptotic case.

\subsection{Summary of contributions}\label{secsum}
While the technical details in the following summary are included for the benefit of the expert reader, the essence of our contributions can be appreciated in light of the foregoing background discussion. In particular, one may note that our main results provide lower bounds on the coherent measurement cost of various instances of coherence distillation. Also notable is that the irretrievable coherence and its variants feature among these bounds, as our intuitive arguments above had suggested.

Our first contribution is to establish a formal connection between the coherent measurement cost of distillation and a construction that we call the \emph{target effect} (Section \ref{sectw}). For any given quantum process (channel) $\cE$ and designated target state $\ket\alpha$, we define an associated target effect $T_\cE^\alpha$, constructed to capture the probability with which $\cE$ maps an arbitrary input $\rho$ to the target $\alpha$. We then show that the requisite coherent measurement action in any implementation of $\cE$ is quantified by the \emph{static} coherence of $T_\cE^\alpha$ considered as a state. We then use this result as the foundation to study the coherent measurement cost in several forms of the distillation task. Although the construction is motivated by our interest in the gap between SIO and IO, its applicability is general enough that all our results apply to the most general class of free operations in coherence resource theories, the \emph{coherence\n non-creating channels} (established in the literature as ``maximal incoherent operations'', or MIO).

As mentioned above, we first consider a certain single-shot (i.e.,\ finite-sized and non-asymptotic) variant containing idealized versions of the AEP-related features encountered in maximal asymptotic distillation (Section \ref{secmax}):
\begin{ithm}\label{itfinmax}
    Any coherence\n non-creating channel that deterministically maps a rank-$S$ state $\rho^\rA$ to a standard coherent resource $\Psi_M$ with $M=A/S$ must involve coherent measurement over at least
    \begin{equation}
        M^{-1}r_C\left(\tau_\rho\right)\ge M^{-1}\exp_2C_f\left(\tau_\rho\right)\ge\exp_2\ell\left(\tau_\rho\right)
    \end{equation}
    elements of $\rA$'s incoherent basis, where $\tau_\rho:=\eins_\rho/S$.
\end{ithm}
Here, $r_C$ denotes the \emph{coherence rank} of a state, defined for pure states as $r_C(\psi)=\mathrm{rank}\Delta(\psi)$ and for mixed states as
\begin{equation}\label{crank1}
    r_C(\rho)=\min_{p_x\ge0;\;\sum_xp_x\phi_x=\rho}\max_xr_C\left(\phi_x\right).
\end{equation}
We prove Theorem \ref{itfinmax} by showing that the target effect associated with any channel achieving such a transformation must be proportional to $\tau_\rho$. When we present our results in detail, we will see that the condition $M=A/S$ is always associated with the distilled resource' being maximal. In general, if an IO can distill $\Psi_M$ from a rank-$S$ state in $A$ dimensions, then $M\le A/S$. Our next result (Section \ref{secnonm}) applies to exact \emph{non-maximal} distillation, again in the single-shot regime.
\begin{ithm}\label{itfin}
    Any coherence\n non-creating channel that deterministically maps a rank-$S$ state $\rho^\rA$ to a standard coherent resource $\Psi_M$ with $M=pA/S$ must involve coherent measurement over at least $M^{-1}r_{C;p}\left(\tau_\rho\right)\ge M^{-1}\exp_2C_{f;p}\left(\tau_\rho\right)\ge\exp_2\ell_{;p}\left(\tau_\rho\right)$ elements of $\rA$'s incoherent basis, where $\tau_\rho:=\eins_\rho/S$ and the bounds are defined as follows:
    \begin{align}
        r_{C;p}(\tau)&:=\min_{\gamma\in\sT^\rA_p(\tau)}r_C(\gamma);\\
        C_{f;p}(\tau)&:=\min_{\gamma\in\sT^\rA_p(\tau)}C_f(\gamma);\\
        \ell_{;p}(\tau)&:=C_{f;p}(\tau)-C_r(\tau)
    \end{align}
    with
    \begin{equation}
        \sT^\rA_p(\tau):=\left\{\gamma=p\tau+(1-p)\tau_\perp:\Delta(\gamma)=\frac{\eins^\rA}A\right\},
    \end{equation}
    $\tau_\perp$ taking values as density operators on the subspace complementary to $\tau$'s support.
\end{ithm}
We arrive at this result by methods similar to those of Theorem \ref{itfinmax}: showing that the target effect of any viable channel is associated with a state $\sigma$ containing a fraction $p$ of $\tau_\rho$ mixed with some state orthogonal thereto.

Moving on, we have an \emph{approximate} version of the maximal case\m a lower bound on the coherent measurement cost of mapping a given input to an output that is close enough (i.e.,\ has high enough fidelity) to a near-maximal standard resource (Section \ref{secappr}).
\begin{ithm}\label{itfina}
    Any coherence\n non-creating channel that deterministically maps a rank-$S$ state $\rho^\rA$, satisfying $r_{\min}\eins^\rA\le\rho^\rA\le r_{\max}\eins^\rA$, to an output $\sigma^\rM$ such that $F\left(\sigma,\Psi_M\right)\ge1-\epsilon$ for $M=A/\tilde S$ must involve coherent measurement over at least
    \begin{equation}
        \frac{\exp_2\left[C_f\left(\tau_\rho\right)-\delta\log_2A-(1+\delta)h\left(\frac\delta{1+\delta}\right)\right]}M
    \end{equation}
    elements of $\rA$'s incoherent basis, where $\tau_\rho:=\eins_\rho/S$ and $\delta:=\sqrt{2\left(1-\frac{S_{\rho,\epsilon}}{\sqrt{S\tilde S}}\right)}$
    with
    \begin{equation}
        S_{\rho,\epsilon}:=\max\left\{\frac{1-\epsilon}{r_{\max}},\,S\left(1-\frac{\epsilon}{r_{\min}}\right)\right\}.
    \end{equation}
\end{ithm}
For small $\epsilon$ and near-maximal distillate (i.e.,\ $\tilde S\approx S$), $\delta$ becomes small and we recover Theorem \ref{itfinmax} as a limiting case. Although we do not delve into approximate \emph{non}-maximal distillation, our methods can be suitably extended to address this case. Theorem \ref{itfina} follows by applying the asymptotic continuity of the coherence of formation \cite{WY16} to approximate versions of Theorem \ref{itfinmax}'s conditions.

Next, in Section \ref{secas}, we look at distillation in the asymptotic limit, where the input is of the form $\varrho_n\equiv\rho^{\otimes n}$. Based on our result on approximate distillation, and on the fact that $\varrho_n$ for large $n$ is approximately maximally-mixed on a large part of its support (due to a statistical effect called asymptotic typicality\m see Appendix \ref{apptyp}), we make the following conjecture.
\begin{iconjm}\label{ipan}
    Suppose a sequence $\cE_n$ of MIO channels is maximally-distilling on copies of an input $\rho$. That is, $\cE_n$ acting on $\varrho_n\equiv\rho^{\otimes n}$ achieves $F\left[\cE_n\left(\varrho_n\right),\Psi_{M_n}\right]\ge1-o(1)$ for $\log_2M_n=n\left[C_r(\rho)-o(1)\right]$. Then any Kraus operator decomposition of $\cE_n$ involves coherent measurement over at least $L_n$ elements of the input's incoherent basis, where
    \begin{equation}
        \log_2L_n\ge n\left[\ell(\rho)-o(1)\right].
    \end{equation}
    This rank bound applies both to the single most coherent measurement element and to the \emph{average} (of the rank's logarithm) over all involved measurements under the distribution induced by the input.
\end{iconjm}
We make some progress towards proving this conjecture. The proof sketch proceeds very similarly to the approximate single-shot case, with the approximation threshold dictated by $n$-dependent parameters associated with asymptotic equipartition. Essentially, we show that with increasing $n$ the task gets closer to the idealized maximal instance of Theorem \ref{itfinmax}\m a formalization of the observations we made in Section \ref{secclue}. But unfortunately, some subtleties of asymptotic typicality pose obstacles in completing our proof. Nevertheless, in the course of our attempts, we show that the target effects in asymptotically maximal distillation saturate certain scaling laws enforced by the generalized quantum Stein's lemma \cite{hayashi2024generalized,lami2024solution}\m a lemma with far-reaching consequences in quantum statistics and resource theories.

So far, we showed or conjectured the \emph{necessity} of a certain coherent measurement cost for distillation. In the asymptotic limit, we conjecture that the cost scaling in Conjecture \ref{ipan} is also \emph{sufficient}.
\begin{iconjm}\label{ipas}
    For any $\rho$, there exists a sequence $\cE_n$ of maximally-distilling \emph{IO} channels with all but an asymptotically-vanishing fraction of measurements \emph{individually} attaining the bound of Conjecture \ref{ipan}, as well as attaining it on average.
\end{iconjm}
This is motivated by certain special properties of the IO distillation protocol constructed by Winter and Yang \cite{WY16}. We also make partial progress towards proving this conjecture, including a general recipe for constructing maximally-distilling IO channels (Appendix \ref{psc2}, Observation \ref{consmax}). But overall, it turns out to be rather more involved than the other direction, requiring putting together several pieces:
\begin{enumerate}
    \item   A construction for a decomposition of $\varrho_n$ that
            \begin{itemize}
                \item   asymptotically approaches the defining bound \eqref{cof} of the coherence of formation and
                \item   possesses some symmetries (thanks to $\varrho_n$'s asymptotic typicality properties), whereby the coherence of each component in the decomposition approaches the overall average value (i.e.,\ $\varrho_n$'s coherence of formation).
            \end{itemize} 
    \item   A sequence of maximally-distilling (candidate) IO subchannels $\cF_n$ based on
            \begin{itemize}
                \item   filtering the above decomposition to further ``typicalize'' or ``flatten'' the coherence in each pure component,
                \item   truncating the remaining components to get rid of parts more coherent than a threshold that asymptotically scales favourably, and
                \item   adapting from Winter and Yang's IO distillation protocol \cite{WY16} a certain ``pooling'' of the classical labels to construct potential maximally-distilling IO channels by connecting their target effects to approximate convex decompositions of $\varrho_n$.
            \end{itemize}
            By virtue of the above truncation, the resulting $\cF_n$ use measurement coherence bounded by the claimed scaling.
    \item   Showing that, despite the above filtering and truncation, the IO subchannel sequence $\cF_n$ asymptotically converges towards trace preservation, so that the maximal distillate it produces is asymptotically deterministic.
    \item   Finally, showing that the (asymptotically negligible but nonzero) trace deficit remaining can be fulfilled by completing each $\cF_n$ to a full channel $\cE_n=\cF_n+\cG_n$ using a subchannel $\cG_n$ that is also IO.
\end{enumerate}
We encounter hurdles related to those of Conjecture \ref{ipan}, but also other, unrelated ones. Due to the added difficulties, we are inclined to place less confidence in Conjecture \ref{ipas}.

In Section \ref{secram} we build on these conjectures to speculate on the requisite coherent measurement cost of \emph{non-}maximal distillation in the asymptotic limit, and on a possible tradeoff between the amount of coherent measurement used and the distilled yield.

Apart from the results summarized above, we also make some progress in understanding cases where the target state is a nonuniform superposition (Section \ref{secnonm}), a semidefinite programming (SDP)\n based approach to estimating the coherent measurement cost (Appendix \ref{appsdp}), the behaviour of certain SIO monotones under constrained IO (Appendix \ref{secmuk}), and connections between coherence distillation and certain linear-algebraic structures that we call \emph{decoupling schemes} (Appendix \ref{secdec}).

With this summary, we are now ready to present our work in full detail.

\section{Technical preliminaries}
Throughout this paper, we will use upright Roman symbols (e.g.,\ $\rA$) for system labels; these labels will be rendered as superscripts where we deem them necessary, and omitted altogether where clear from context. For a system $\rA$, $\cH^\rA$ will denote its associated Hilbert space and $A$ the dimensionality thereof. We will use $a$ for the classical symbols labelling $\rA$'s incoherent basis vectors $\ket a$, and $\cA$ for the collection thereof (i.e.,\ $\rA$'s ``classical alphabet''). For a generic space $\cV\subseteq\cH^\rA$, we will denote the space of its linear automorphisms by $\cL\left(\cV\right)$ and the projector onto $\cV$ by $\eins_\cV$, but use the abbreviation $\eins^\rA\equiv\eins_{\cH^\rA}$ in the case of the entire Hilbert space associated with some system, $\eins_\cI$ for the projector onto $\mathrm{span}\left\{\ket a:\;a\in\cI\subseteq\cA\right\}$, and $\eins_\rho\equiv\eins_{\mathrm{supp}\rho}$ in the case of the support of a positive-semidefinite operator. All vector spaces mentioned will be implicitly assumed to be finite-dimensional. We will use the term ``basis'' to mean specifically ``orthonormal basis''. We will use the established notation $S(\cdot)$ for the von Neumann entropy, $H(\cdot)$ for the Shannon entropy, and $S(\cdot\|\cdot)$ for the Umegaki quantum relative entropy. Though we will also name some variables $S$, this latter use of the symbol will be clear from context, with no scope for confusion. We will abusively express von Neumann (Shannon) entropies with their arguments \emph{density operators} (\emph{distributions}) instead of the systems (random variables) distributed thereby; e.g.,\ $H(\vect p)\equiv H(X)_{\left(p_x\right)_x}$. For a Hermitian operator $T$, we will use the notation $\nrm T_p$ (with $p\in\bbR$) for its $p$\n Schatten norm, defined as the $\ell^p$ norm of the vector of its singular values. We will usually put the arguments of a map / function in round parentheses, e.g.\ $f(x)$, but sometimes use square brackets in the context of nested parenthesization, e.g.\ $\left(x,f[x]\right)$, generally preferring to alternate between square and round in a nested sequence.

In the foregoing discussion, we used the term ``(sub)channel'' a few times. A channel is a quantum operation, or state transformation, induced locally on a quantum system by a unitary interaction with an auxiliary system with which it is initially uncorrelated. Mathematically, a channel mapping states of some system $\rA$ to states of $\rM$\footnote{We will use $\rM$ for the output system since we focus on distillation, a context where denoting the output $\rM$ seems common practice.} can be identified with a completely-positive (CP) trace-preserving (TP) map $\cE:\,\cL\left(\cH^\rA\right)\to\cL\left(\cH^\rM\right)$; we will often identify the input and output systems with the shorthand $\cE^{\rA\to\rM}$, and further abbreviate as $\cE^{\rA}\equiv\cE^{\rA\to\rA}$ when the systems are identical. A subchannel is a CP trace-\emph{nonincreasing} map, and corresponds with the action of a quantum operation conditioned on one of several possible outcomes, each occurring with some (initial state\n dependent) probability; a channel is a special case with only one outcome, occurring deterministically.

Any subchannel can be specified operationally through a so-called \emph{dilation}: a (non-unique) specification of an auxiliary system, a unitary interaction, and an auxiliary measurement, that collectively implement it. Alternately, it can be described somewhat more abstractly through a set (also non-unique) of operators called \emph{Kraus operators}. In the remainder, we will assume the reader has a basic background of these concepts;\ for a detailed introduction, see \cite{wilde2013quantum}. For our purposes, we will need the following definitions for classes of quantum operations. We will primarily have IO in consideration, but most of our results will apply to the largest class of free operations, MIO.
\begin{definition}[Coherence\n non-creating channel]
    $\cE^{\rA\to\rM}$ is a \emph{coherence\n non-creating channel} if it maps all incoherent inputs to incoherent outputs; that is, for all $a\in\cA$,
    \begin{equation}
        \cE^{\rA\to\rM}\left(\proj{a}\right)=\Delta^\rM\circ\cE^{\rA\to\rM}\left(\proj{a}\right).
    \end{equation}
    Since the class of such operations has been established in the literature as the \emph{maximal (class of) incoherent operations} (MIO), we will also use the abbreviation MIO.
\end{definition}
\begin{remm}
    The operations in the MIO class have occasionally been called ``maximally incoherent operations''\m a maximally irresponsible misnomer, considering that MIO are arguably \emph{minimally} incoherent! To avoid such awkwardness, we take the liberty to use the descriptor ``coherence\n non-creating''. Note also that in this paper we use the abbreviation MIO to refer exclusively to \emph{channels} (i.e.,\ trace-preserving operations).
\end{remm}
\begin{definition}[Incoherent operation]
    A (sub) channel $\cE^{\rA\to\rM}$ is an \emph{incoherent operation} (IO) if it admits a Kraus operator decomposition $\cE(\cdot)=\sum\limits_{c\in\cC}K_c(\cdot) K_c^\dagger$ such that $\forall~c\in\cC$ and $a\in\cA$,
    \begin{equation}\label{condIO}
        K_c\ket{a}\propto\ket{m\equiv g_c(a)},
    \end{equation}
    where $g_c:\,\cA\to\cM$ is a function. 
\end{definition}
\begin{remm}\label{iodil}
    For a generic channel $\cE^{\rA\to\rM}(\cdot)\equiv\sum_cK_c(\cdot)K_c^\dagger$, each of its Kraus operators can be expressed (in the incoherent basis) in terms of its rows $\ket{w_{c,m}}^\rA$ as
    \begin{equation}\label{gench}
        K_c=\sum_{m\in\cM}\ket m^\rM\bra{w_{c,m}}^\rA.
    \end{equation}
    In particular, IO Kraus operators have rows of the form
    \begin{equation}\label{genio}
        \ket{w_{c,m}}=\sum_{a\in g_c^{-1}(m)}\xi_{ca}\ket a.
    \end{equation}
    Notice that, by virtue of the IO condition \eqref{condIO}, the $\ket{w_{c,m}}$ for any fixed $c$ and distinct $m$ values involve disjoint subsets of the incoherent basis:
    \begin{equation}
        \Delta\left(\proj{w_{c,m}}\right)\,\perp\,\Delta\left(\proj{w_{c,m'}}\right).
    \end{equation}
    This also implies, of course, that $\bket{w_{c,m}}{w_{c,m'}}\propto\delta_{mm'}$.
\end{remm}
\begin{definition}[Strictly incoherent operation]
    An IO $\cE^{\rA\to\rM}$ is a \emph{strictly incoherent operation} (SIO) if it admits a Kraus operator decomposition $\cE(\cdot)=\sum_{c\in\cC}K_c(\cdot) K_c^\dagger$ that, in addition to satisfying \eqref{condIO}, has every $g_c$ invertible on its image. In other words, $\forall~c\in\cC$ and $m\in g_c\left(\cA\right)$,
    \begin{equation}
        K_c^\dagger\ket{m}\propto\ket{a\equiv g_c^{-1}(m)},
    \end{equation}
    where $g_c^{-1}(m)$ is unique and well-defined.
\end{definition}
\begin{remm}
    Recall that we used $\Delta(\cdot)$ earlier for the diagonal part of the argument in the incoherent basis. In fact, this mapping is a channel\m the dephasing channel\m realizable by SIO: $\Delta(\cdot)=\sum_a\proj a(\cdot)\proj a$ (where $a$ runs over the incoherent labels). Under our assumption that the incoherent basis of a composite system is the tensor product of the constituent systems' incoherent bases, the dephasing channel inherits this convenient multiplicativity: $\Delta^{\rA\rB}=\Delta^\rA\otimes\Delta^\rB$.
\end{remm}
Finally, the formal definition of the quantity that figures in our main results:
\begin{definition}[Irretrievable coherence]
    We define the \emph{irretrievable coherence} of a state $\rho$ as
\begin{equation}
    \ell\left(\rho\right):=C_f\left(\rho\right)-C_r\left(\rho\right),
\end{equation}
where $C_r(\rho):=\min\limits_{\sigma:\;\Delta[\sigma]=\sigma}S\left(\rho\|\sigma\right)=S\left[\rho\|\Delta(\rho)\right]=S\left[\Delta(\rho)\right]-S(\rho)$ is the \emph{relative entropy of coherence} and
\begin{equation}\label{cof2}
    C_f(\rho)=\min_{p_x\ge0;\;\sum_xp_x\phi_x=\rho}\sum_xp_xC_r\left(\phi_x\right)
\end{equation}
its convex-roof extension, the \emph{coherence of formation}.
\end{definition}
In addition to these coherence measures, we will also use the \emph{coherence rank} $r_C$, defined for pure states as $r_C(\psi)=\mathrm{rank}\Delta(\psi)$ and for mixed states as
\begin{equation}\label{crank}
    r_C(\rho)=\min_{p_x\ge0;\;\sum_xp_x\phi_x=\rho}\max_xr_C\left(\phi_x\right).
\end{equation}
Note that $\log_2r_C(\psi)\ge C_r(\psi)$ and therefore $\log_2r_C(\rho)\ge C_f(\rho)$ in general.

\section{The coherent measurement cost and target effects}\label{seccmte}
As we develop various technical tools and results below, we shall bear in mind our ultimate aim\m namely, to estimate the minimum measurement coherence required for performing certain tasks. To this end, let us first formalize this concept.

\subsection{Operationalizing the cost}\label{seccmc}
For a given channel in the form \eqref{gench} in terms of a Kraus operator decomposition, consider the vectors $\ket{w_{c,m}}^\rA$. These capture the coherently-measured input components mapped to different incoherent output labels $m$ (cf.\ the qutrit example in Section \ref{secclue}). Via a suitable Stinespring dilation, they can be operationally interpreted as some part of the requisite coherent measurement when the channel is implemented through (1) a scheme for preparing incoherent basis states $\ket m\otimes\ket c$ \emph{coherently conditioned} on a destructive measurement in some orthonormal basis extending $\left\{\ket{w_{c,m}}\right\}_{c,m}$ (i.e.,\ some $\left\{\ket{\tilde w_{c,m}^{\tilde\rA}}=\ket{w_{c,m}}\oplus\ket{v_{c,m}}\right\}_{c,m}$ with a suitable extension $\tilde\rA$ of $\rA$), followed by (2) an ancillary measurement of the index $c$, effectively decohering the output system relative to this index. Our deliberate choice to expand the dilation isometry in the $\ket m\otimes\ket c$ basis on its output side is in keeping with the resource-theoretic principle of disallowing the preparation of coherent states; of course, the isometry is ultimately just some coherent dynamical process, and its interpretation as ``coherently-conditioned incoherent-state preparation'' is meant only to capture the resource-theoretic spirit in operational terms.

Moving on, note that we can use different measures of coherence to quantify this coherent measurement resource. To keep our theorems clear, we will use the following:
\begin{definition}[Rank-based coherent measurement cost]\label{rcmc}
    We will say that a channel's implementation associated with Kraus operators $K_c$ of the form standardized in Remark \ref{iodil} \emph{involves coherent measurement over at least $r$ elements of the input's incoherent basis}, where
    \begin{equation}
        r:=\max_{c,m}r_C\left(\phi_{c,m}\right).
    \end{equation}
\end{definition}
This definition employs the coherence rank as the underlying coherence measure and, furthermore, applies to specific operator-sum decompositions of channels (rather than to the channels themselves). However, we note the following: (1) the technical lemmas underlying our theorems will provide \emph{measure-agnostic} constraints on the structure of the $\ket{w_{c,m}}^\rA$; (2) moreover, the constraints will hold for \emph{any} channel accomplishing the relevant distillation task and \emph{any} Kraus operator decomposition thereof.
For these reasons, we can meaningfully interpret the constraints as determining the coherent measurement cost of the given task, despite their being derived from Kraus operators instead of from operational descriptions such as Stinespring dilations.
    
While additional coherent measurement action may be required on auxiliary systems, our method accounts for all such action \emph{in the space of the input} $\rA$. On this note, we also assume that $\rA$ is just big enough to contain the incoherent basis elements occurring with nonzero amplitudes in the input under consideration. This brings about no loss in generality, since any channel action outside this subspace is irrelevant for matters concerning the coherent measurement cost of distilling from the given input (except in possibly inflating the cost superfluously).

\subsection{The target effect construction}\label{sectw}
The following construction will prove useful in characterizing the coherence of the vectors $\ket{w_{c,m}}$ and, thereby, the coherent measurement cost of the channel they implement.
\begin{definition}[Target effect]\label{defTE}
    Given a channel $\cE^{\rA\to\rM}$ and a ``target'' state $\alpha^\rM$, define the \emph{target effect}
    \begin{equation}
        T_\cE^\alpha:=\cE^\dagger\left(\alpha\right),
    \end{equation}
    where $\cE^\dagger$ is the \emph{Hilbert\n Schmidt adjoint map} of $\cE$. Consequently, for any operator $X\in\cL\left(\cH^\rA\right)$,
\begin{equation}\label{trkt}
    \tr\left(X T_\cE^\alpha\right)=\tr\left[\alpha\cE\left(X\right)\right].
\end{equation}
In particular, by choosing $X$ to be a density operator, we see that $T_\cE^\alpha$ is an effect\footnote{A positive-semidefinite Hermitian operator that induces probabilities on density operators is called a \emph{positive operator\n valued measure (POVM) effect}, or simply an effect.} whose expectation value under any given state quantifies the probability with which $\cE$ maps the state to the target\m hence its name. As an effect, it satisfies $0\le T_\cE^\alpha\le\eins^\rA$.
\end{definition}
\begin{remm}
    Hereafter, we will consider only \emph{pure} target states, which is the relevant case for distillation.
\end{remm}
\begin{obs}\label{TgivesF}
    Given a channel $\cE^{\rA\to\rM}$, a decomposition $\cE(\cdot)=\sum_cK_c(\cdot)K_c^\dagger$ with Kraus operators $K_c=\sum_m\ket m\bra{w_{c,m}}$, and a pure target $\ket\alpha=\sum_m\alpha_m\ket m$, define
    \begin{equation}
        \ket{w_c^\alpha}:=\sum_m\alpha_m\ket{w_{c,m}}
    \end{equation}
    for each $c$. Then, the associated target effect satisfies
    \begin{equation}
        T_\cE^\alpha=\sum_c\proj{w_c^\alpha},
    \end{equation}
    \emph{irrespective} of the Kraus operator decomposition chosen.
\end{obs}
\begin{proof}
    Using the definition of $K_c$ in terms of the $\ket{w_{c,m}}$,
    \begin{align}
        &T_\cE^\alpha=\cE^\dagger\left(\alpha\right)=\sum_cK_c^\dagger\proj\alpha K_c\nonumber\\
        &=\sum_{c,m,m'}\kbra{w_{c,m}}{m}\alpha_{m}\kbra{m}{m'}\alpha_{m'}^*\kbra{m'}{w_{c,m'}}\nonumber\\
        &=\sum_{c,m,m'}\alpha_{m}\kbra{w_{c,m}}{w_{c,m'}}\alpha_{m'}^*=\sum_c\proj{w_c^\alpha}.
    \end{align}
        
\end{proof}
In the IO case, as noted in Remark \ref{iodil}, the various $\ket{w_{c,m}}$ for a given $c$ involve disjoint incoherent subbases. As such, their coherence is fully reflected in that of the $\ket{w_c^\alpha}$, enabling us to use the latter to put lower bounds on various measures of the $\ket{w_{c,m}}$'s coherence. While this is no longer true for more general channels, lower bounds on coherence measures of the $\ket{w_c^\alpha}$ place \emph{some} constraints on the collective structure of the the $\ket{w_{c,m}}$. In particular, the \emph{coherence rank} of $\ket{w_{c}^\alpha}$ cannot be larger than the sum of those of the $\ket{w_{c,m}}$'s; this will allow us to put a lower bound on the average row coherence rank even under the more general MIO class of operations.

The $\ket{w_{c}^\alpha}$, in turn, constitute a convex decomposition of $T_\cE^\alpha$, whereby coherence quantifiers based on minimizing over all convex decompositions (e.g.,\ $C_f$) applied on the latter yield lower bounds on the values that corresponding pure-state coherence quantifiers (e.g.,\ $C_r$) take on the $\ket{w_c^\alpha}$.
\begin{obs}\label{tbound}
    For an MIO $\cE^{\rA\to\rM}$ and target $\ket\alpha$, the associated target effect $T_\cE^\alpha$ satisfies
        \begin{align}
            \bra aT_\cE^\alpha\ket a&=\tr\left[\cE\left(\proj a\right)\;\Delta\left(\alpha\right)\right]\nonumber\\
            &=\bra a\cE^\dagger\circ\Delta\left(\alpha\right)\ket a
        \end{align}
        $\forall a\in\cA$.
\end{obs}
\begin{proof}
    Since $\cE$ is an MIO, $\cE\left(\proj a\right)=\Delta\circ\cE\left(\proj a\right)$ for any incoherent basis element $\ket a^\rA$. Therefore,
    \begin{align}
        \bra aT_\cE^\alpha\ket a&=\tr\left[\proj a\;\cE^\dagger\left(\alpha\right)\right]\nonumber\\
        &=\tr\left[\cE\left(\proj a\right)\;\alpha\right]\nonumber\\
        &=\tr\left[\Delta\circ\cE\left(\proj a\right)\;\alpha\right]\nonumber\\
        &=\tr\left[\cE\left(\proj a\right)\;\Delta\left(\alpha\right)\right]\nonumber\\
        &=\tr\left[\proj a\;\cE^\dagger\circ\Delta\left(\alpha\right)\right].
    \end{align}
\end{proof}
\begin{remm}\label{remtrace}
    By summing over $a\in\cA$, we get $\tr\,T_\cE^\alpha=\tr\,\cE^\dagger\left[\Delta\left(\alpha\right)\right]$. But note that this trace condition holds for any channel $\cE$ that maps the identity to an incoherent output: $\cE\left(\eins\right)=\Delta\circ\cE\left(\eins\right)$. Our forthcoming results can all be adapted to apply to all channels with this property (e.g.,\ all unital channels) by loosening every application of Observation \ref{tbound} to one of this trace condition.
\end{remm}
\begin{remm}
    In all our subsequent use of this construction, the target $\ket\alpha$ will be clear from the context, and therefore we will use the simplified notation $\ket{w_c}$ and $T_\cE$.
\end{remm}

\section{Single-shot distillation}
We are now ready to apply the tool of target effects to the study of distillation. We first consider the task of distilling in the single-shot regime, i.e.\ from finite-sized inputs away from the asymptotic limit. Let us start with exact distillation: an MIO $\cE$, acting on an input $\rho^\rA$, is required to map it exactly to a desired target $\alpha^\rM=\proj\alpha$. Since the desired output is pure and the transformation is required to be exact, $\cE$ must in fact uniformly map the entire space $\cL(\cV)$, with $\cV\equiv\mathrm{supp}\rho\subseteq\cH^\rA$, to $\alpha$ (up to a scalar factor)\footnote{We will hereafter abbreviate this condition, abusively, as $\cE$ ``mapping $\cV$ to $\alpha$''.}; as such, the only relevant property of $\rho$ is the structure of its support $\cV$.
\begin{alem}\label{TWisV}
    If a channel $\cE^{\rA\to\rM}$ maps a subspace $\cV\subset\cH^\rA$ exactly to a target $\alpha$, then the target effect $T_\cE$ has the structure
    \begin{equation}\label{tblock}
        T_\cE=\eins_\cV+T_{\perp},
    \end{equation}
    where $T_{\perp}$ is supported entirely on $\cV_\perp\subset\cH^\rA$, the subspace complementary to $\cV$.
\end{alem}
\begin{proof}
    Since $\cE$ is required to map all of $\cV$ to the pure state $\alpha$, Definition \ref{defTE} entails $\bra vT_\cE\ket v=\bkt v$ for all $\ket v\in\cV$, whereby $\eins_\cV T_\cE\eins_\cV=\eins_\cV$. Furthermore, $0\le T_\cE\le\eins^\rA$ (also noted in Definition \ref{defTE}) then implies $\eins_\cV T_\cE\eins_{\cV_\perp}=0$.
\end{proof}
All our main results will follow by applying the above simple result (or its variants) to progressively more complex cases of distillation.

\subsection{Maximal distillation}\label{secmax}
Let us start with the ``idealized maximal distillation'' case motivated in Section \ref{secclue}.
\begin{alem}\label{lwcu}
    If an MIO channel $\cE^{\rA\to\rM}$ maps an $S$-dimensional subspace $\cV\subset\cH^\rA$ exactly to the $M$-fold standard coherent state $\Psi_M$ with $S=A/M$ (the divisibility being assumed), then
    \begin{equation}\label{iwcwcu}
        T_\cE=\eins_\cV.
    \end{equation}
\end{alem}
\begin{proof}
    By virtue of Lemma \ref{TWisV}, it remains only to show that $T_\perp=0$. Applying Observation \ref{tbound} to $\ket\alpha\equiv\ket{\Psi_M}$,
    \begin{equation}
        \tr\,T_\cE=\tr\,\cE^\dagger\left[\Delta\left(\Psi_M\right)\right]=\tr\,\cE^\dagger\left(\frac{\eins^\rM}M\right).
    \end{equation}
    Since $\rA$ contains no $\cV$-extraneous incoherent labels (see Remark \ref{iodil}), and since all of $\cV$ must be mapped to $\Psi_M$, the image of $\cE$ is supported on the $M$-dimensional space spanned by the incoherent components of $\ket{\Psi_M}$. The TP property of $\cE$ implies that $\cE^\dagger$ is unital on this image: $\cE^\dagger\left(\eins^\rM\right)=\eins^\rA$. Therefore,
    \begin{equation}
        \tr\,T_\cE=\tr\,\frac{\eins^\rA}M=\frac AM=S.
    \end{equation}
    Thus, $\tr\,T_\perp=0$. Since $T_\cE\ge0$, this implies $T_\perp=0$.
\end{proof}
By applying coherence measures on Lemma \ref{lwcu}, we have the following precursor to our first main result.
\begin{prop}\label{pfinu}
    Let $\cE^{\rA\to\rM}$, with Kraus operators $K_c:=\sum_{m\in\cM}\ket m^\rM\bra{w_{c,m}}^\rA$, be an MIO channel that distills $\Psi_M$ from a subspace $\cV\subseteq\cH^\rA$ of dimensionality $S=A/M$. Denoting $t_{c,m}:=\bkt{w_{c,m}}$ and $t_c:=M^{-1}\sum_mt_{c,m}$, define the normalized states $\ket{\phi_{c,m}}:=t_{c,m}^{-1/2}\ket{w_{c,m}}$, $\ket{\phi_{c}}:=\left(Mt_c\right)^{-1/2}\sum_m\ket{w_{c,m}}=\sum_m\sqrt{\frac{t_{c,m}}{Mt_c}}\ket{\phi_{c,m}}$, and $\tau_\cV:=\eins_\cV/S$; note that $\left(t_{c,m}/\left[Mt_c\right]\right)_{m}$ and $\left(t_c/S\right)_c$ are normalized distributions. Then, the following must hold:
    \begin{enumerate}
        \item\label{pfp1} $\sum\limits_{c}\frac{t_c}SC_r\left(\phi_{c}\right)\ge C_f\left(\tau_\cV\right)$;
        \item\label{pfp2} $\max\limits_{c}r_C\left(\phi_{c}\right)\ge r_C\left(\tau_\cV\right)$;
        \item\label{pfp3} $\sum\limits_mr_C\left(\phi_{c,m}\right)\ge r_C\left(\phi_{c}\right)$ for all $c$;
        \item\label{pfp4} $\log_2M\le C_r\left(\tau_\cV\right)$.
    \end{enumerate}
    Consequently,
    \begin{widetext}
        \begin{align}\label{cftauu}
            \max\limits_{c,m}r_C\left(\phi_{c,m}\right)\ge\frac{\max\limits_c\sum\limits_mr_C\left(\phi_{c,m}\right)}M&\ge\frac{\max\limits_cr_C\left(\phi_{c}\right)}M\nonumber\\
            &\ge\left\{\begin{array}{c}
            \frac{r_C\left(\tau_\cV\right)}M\\
            \left\{\begin{array}{c}
            \frac{2^{\sum\limits_c\frac{t_c}S\log_2r_C\left(\phi_{c}\right)}}M\\
            \frac{2^{\max\limits_cC_r\left(\phi_{c}\right)}}M
            \end{array}\right\}\ge\frac{2^{\sum\limits_c\frac{t_c}SC_r\left(\phi_{c}\right)}}M
            \end{array}\right\}\ge\frac{2^{C_f\left(\tau_\cV\right)}}M\ge2^{\ell\left(\tau_\cV\right)}.
        \end{align}
    \end{widetext}
    Furthermore, if some $K_c$ is also IO, then
    \begin{equation}\label{crphi}
        \log_2M+\sum_{m\in\cM}\frac{t_{c,m}}{Mt_c}C_r\left(\phi_{c,m}\right)\ge C_r\left(\phi_{c}\right),
    \end{equation}
    which adds more detail to the inequality chains in \eqref{cftauu}.
\end{prop}
\begin{proof}
    Applying the appropriate normalization factors to \eqref{iwcwcu}, $\tau_\cV=\sum_c\left(t_c/S\right)\phi_c$, whence points \ref{pfp1} and \ref{pfp2} follow. Point \ref{pfp3} follows from the sub-additivity of the coherence rank under vector addition. For point \ref{pfp4}, since $\cE\left(\tau_\rho\right)=\Psi_M$ and the relative entropy of coherence is a monotone under MIO, $C_r\left(\tau_\rho\right)\ge C_r\left(\Psi_M\right)=\log_2M$.
    
    Now suppose the $K_c$ are IO Kraus operators, so that for any given $c$, the distinct $\ket{\phi_{c,m}}$ don't overlap in their classical symbols. Then, $\Delta\left(\phi_c\right)\cong\bigoplus_m\frac{t_{c,m}}{Mt_c}\Delta\left(\phi_{c,m}\right)$ and therefore
    \begin{align}\label{crpcu}
        C_r\left(\phi_{c}\right)&=S\left[\Delta\left(\phi_c\right)\right]\nonumber\\
        &=H\left[\left(\frac{t_{c,m}}{Mt_c}\right)_m\right]+\sum_{m\in\cM}\frac{t_{c,m}}{Mt_c}S\left[\Delta\left(\phi_{c,m}\right)\right]\nonumber\\
        &\le\log_2M+\sum_{m\in\cM}\frac{t_{c,m}}{Mt_c}C_r\left(\phi_{c,m}\right).
    \end{align}
\end{proof}
\begin{remm}\label{tcmtc}
    In fact, in the IO case, $t_{c,m}=t_c$ for all $(c,m)$ and, furthermore,
    \begin{equation}\label{crpctu}
        C_r\left(\phi_{c}\right)=\log_2M+\sum_{m\in\cM}M^{-1}C_r\left(\phi_{c,m}\right).
    \end{equation}
    To see this, note that since $T_\cE-\proj{w_c}\ge0$, $\ket{w_c}\in\cV$ $\forall c$. Exploiting the action of the $K_c$'s on $\cV$,
    \begin{align}
        K_c\ket{w_c}&\propto\ket{\Psi_M}\nonumber\\
        \Rightarrow\left(\sum_{m_1}\kbra{m_1}{w_{c,m_1}}\right)\left(\sum_{m_2}\ket{w_{c,m_2}}\right)&\propto\sum_m\ket m\nonumber\\
        \Rightarrow\sum_m\ket m\bkt{w_{c,m}}&\propto\sum_m\ket m,
    \end{align}
    the latter owing again to $\bket{w_{c,m_1}}{w_{c,m_2}}\propto\delta_{m_1m_2}$. Consequently, $\bkt{w_{c,m}}$ must be independent of $m$, which by normalization yields $t_{c,m}=t_c$ for all $(c,m)$. Hence, we can refine the observation preceding \eqref{crpcu} to $\Delta\left(\phi_c\right)\cong\bigoplus_mM^{-1}\Delta\left(\phi_{c,m}\right)$, in turn refining \eqref{crpcu} to
    \begin{align}
        C_r\left(\phi_{c}\right)&=H\left[\left(M^{-1}\right)_m\right]+\sum_{m\in\cM}\frac{S\left[\Delta\left(\phi_{c,m}\right)\right]}M,
    \end{align}
    thence obtaining \eqref{crpctu}. Nevertheless, the essence of Proposition \ref{pfinu} is in bounding various measures of the collective coherence of the $\phi_{c,m}$ in any Kraus operator representation of the channel in question. We will later prove a version of this result for more general cases, where the output is not a maximal $\Psi_M$. In that context, we will see that a variant of the average coherence inequality \eqref{cftauu} still holds, even though the properties discussed in this remark do not.
\end{remm}

Recalling Definition \ref{rcmc} for the requisite coherent measurement rank, Proposition \ref{pfinu} immediately yields our first main result as a corollary, which we state without proof.
\begin{theorem}\label{tfinmaxu}
    Any coherence\n non-creating channel that deterministically maps a rank-$S$ state $\rho^\rA$ to a standard coherent resource $\Psi_M$ with $M=A/S$ must involve coherent measurement over at least
    \begin{equation}
        M^{-1}r_C\left(\tau_\rho\right)\ge M^{-1}\exp_2C_f\left(\tau_\rho\right)\ge\exp_2\ell\left(\tau_\rho\right)
    \end{equation}
    elements of $\rA$'s incoherent basis, where $\tau_\rho:=\eins_\rho/S$.
\end{theorem}
\begin{remm}
    Proposition \ref{pfinu} is stronger than Theorem \ref{tfinmaxu}: it puts bounds not only on the maximal coherent measurement rank, but also on various collective or average coherence properties of the $\ket{\phi_{c,m}}$. While the distribution governing this averaging is operationally appropriate only for an input satisfying $\tau_\rho=\rho$ (i.e.,\ one with a flat spectrum) and not otherwise, it is expected to approach the actual distribution of coherent measurement resources in the asymptotic limit (Section \ref{secas}). As such, our bounds affect not only some outlying measurements occurring with small probabilities, but even the collective statistics of all involved measurements, especially in the asymptotic limit.
    
    It is also notable that we can get variants of Proposition \ref{pfinu} by using different coherence quantifiers, although generic quantifiers may not admit a clean splitting of the coherence of each $\phi_c$ into separate terms for $\Psi_M$ and the $\phi_{c,m}$. The essential structure is captured by Lemma \ref{lwcu}, whereof each coherence quantifier illuminates certain specific facets.
    
    Nevertheless, we have explicated and highlighted the weaker form in Theorem \ref{tfinmaxu} for two reasons. Firstly, applying Proposition \ref{pfinu} to an instance requires a detailed specification of all the Kraus operators, which may be arbitrarily numerous. On the other hand, Theorem \ref{tfinmaxu} only requires specifying the largest coherent measurement rank used in implementing a channel. Secondly, we will see in Section \ref{secas} that in the asymptotic limit, the above-noted operational interpretation of the average coherence is expected to be not only necessary but also sufficient for maximal distillation.
\end{remm}
We shall now extend the above results to non-maximal distillation tasks.

\subsection{Non-maximal distillation}\label{secnonm}
Consider the distillation of a generic $\ket\alpha=\sum_m\alpha_m\ket m$ (without loss of generality, we can assume $\alpha_m\in\bbR$ and $\alpha_m>0$). First, we have a counterpart to Lemma \ref{lwcu}.
\begin{alem}\label{lwc}
    If $\cE^{\rA\to\rM}$ is an MIO that distills $\alpha^\rM$ from a subspace $\cV\subseteq\cH^\rA$, then
    \begin{equation}
        \alpha_{\min}^2\eins^\rA\le\Delta\left(T_\cE\right)\le\alpha_{\max}^2\eins^\rA,
    \end{equation}
    where $\alpha_{\min/\max}:=\mathop{\min/\max}\limits_m\alpha_m$.
\end{alem}
\begin{proof}
    First, note that $\Delta\left(\alpha\right)=\sum_m\alpha_m^2\proj m$. Thus, $\alpha_{\min}^2\eins^\rM\le\Delta\left(\proj\alpha\right)\le\alpha_{\min}^2\eins^\rM$
    , whence the result follows from Observation \ref{tbound} and the TP property of $\cE$.
\end{proof}
This allows us to derive a non-maximal variant of Proposition \ref{pfinu}.
\begin{prop}\label{pfin}
    Suppose an MIO $\cE^{\rA\to\rM}$, with Kraus operators $K_c:=\sum\limits_{m\in\cM}\ket m^\rM\bra{w_{c,m}}^\rA$, distills the state $\alpha$ from a subspace $\cV\subseteq\cH^\rA$. Define $S:=\tr\,T_\cE$ for the associated target effect, and $\tau_\cV:=\eins_\cV/\dim\cV$. Denoting $t_{c,m}:=\bkt{w_{c,m}}$ and $t_c:=\sum\limits_m\alpha_m^2t_{c,m}$, define the normalized states $\ket{\phi_{c,m}}:=t_{c,m}^{-1/2}\ket{w_{c,m}}$ and $\ket{\phi_{c}}:=t_c^{-1/2}\sum_m\alpha_m\ket{w_{c,m}}=\sum_m\alpha_m\sqrt{\frac{t_{c,m}}{t_c}}\ket{\phi_{c,m}}$; note that $\left(\alpha_m^2t_{c,m}/t_c\right)_{m}$ and $\left(t_c/S\right)_c$ are now normalized distributions by Lemma \ref{lwc}. Define the set $\sT:=\left\{T\in\cL\left(\cH^\rA\right):\;\mathrm{Conditions~\eqref{propst}}\right\}$, wherein the conditions are as follows:
    \begin{align}\label{propst}
        T^\dagger&=T;\nonumber\\
        \eins_\cV\le T&\le\eins^\rA;\nonumber\\
        \left(T-\eins_\cV\right)\eins_\cV&=0;\nonumber\\
        \alpha_{\min}^2\eins^\rA\le\Delta(T)&\le\alpha_{\max}^2\eins^\rA.
    \end{align}
    Through this set, define
    \begin{align}
        r_0&:=\min_{T\in\sT}r_C\left(T/\tr T\right);\\
        C_0&:=\min_{T\in\sT}C_f\left(T/\tr T\right).
    \end{align}
    Then,
    \begin{widetext}
        \begin{equation}\label{cftau}
            \max\limits_{c,m}r_C\left(\phi_{c,m}\right)\ge\frac{\max\limits_c\sum\limits_mr_C\left(\phi_{c,m}\right)}M\ge\frac{\max\limits_cr_C\left(\phi_{c}\right)}M\ge\left\{\begin{array}{c}
            r_0/M\\
            \left\{\begin{array}{c}
            \frac{2^{\sum\limits_c\frac{t_c}{S}\log_2r_C\left(\phi_{c}\right)}}M\\
            \frac{2^{\max\limits_cC_r\left(\phi_{c}\right)}}M
            \end{array}\right\}\ge\frac{2^{\sum\limits_c\frac{t_c}{S}C_r\left(\phi_{c}\right)}}M
            \end{array}\right\}\ge\frac{2^{C_0}}M.
        \end{equation}
    \end{widetext}
    Furthermore, if some $K_c$ is also IO, then
    \begin{equation}\label{crphii}
        \log_2M+\sum_{m\in\cM}\frac{\alpha_m^2t_{c,m}}{t_c}C_r\left(\phi_{c,m}\right)\ge C_r\left(\phi_{c}\right).
    \end{equation}
\end{prop}
We omit the proof, since it follows exactly like that of Proposition \ref{pfinu}. We leave further investigations concerning general $\ket\alpha$'s for future work. In the remainder of this paper, we will restrict ourselves to standard outputs $\ket\alpha=\ket{\Psi_M}$, for which $\alpha_{\min}=\alpha_{\max}=M^{-1/2}$. The following definitions will be useful for these cases.
\begin{definition}\label{defcp}
    Given a density operator $\tau^\rA$ and $p\in[0,1]$, let
    \begin{equation}
        \sT^\rA_p(\tau):=\left\{p\tau+(1-p)\sigma:\tau\perp\sigma,\Delta(\gamma)=\frac{\eins^\rA}A\right\},
    \end{equation}
    where $\sigma$ takes values as density operators on $\rA$ acting on the subspace complementary to $\tau$'s support. Note that, unlike the $\sT$ of Proposition \ref{pfin}, this set contains only normalized density operators. Using it, define
    \begin{align}
        r_{C;p}(\tau)&:=\min_{\gamma\in\sT^\rA_p(\tau)}r_C(\gamma);\\
        C_{f;p}(\tau)&:=\min_{\gamma\in\sT^\rA_p(\tau)}C_f(\gamma);\\
        \ell_{;p}(\tau)&:=C_{f;p}(\tau)-C_r(\tau).
    \end{align}
\end{definition}
\begin{coro}
    Suppose an MIO $\cE^{\rA\to\rM}$, with Kraus operators $K_c:=\sum\limits_{m\in\cM}\ket m^\rM\bra{w_{c,m}}^\rA$, distills $\Psi_M$ from a subspace $\cV\subseteq\cH^\rA$. For the associated target effect $T_\cE$, let $S:=\tr T_\cE$ and define $\tau_\cV:=\eins_\cV/\dim\cV$; note that Lemma \ref{lwc} implies $S=A/M$, while Lemma \ref{TWisV} entails $S\ge\dim\cV$. Denoting $t_{c,m}:=\bkt{w_{c,m}}$ and $t_c:=M^{-1}\sum_mt_{c,m}$, define the normalized states $\ket{\phi_{c,m}}:=t_{c,m}^{-1/2}\ket{w_{c,m}}$ and $\ket{\phi_{c}}:=\left(Mt_c\right)^{-1/2}\sum_m\ket{w_{c,m}}=\sum_m\sqrt{\frac{t_{c,m}}{Mt_c}}\ket{\phi_{c,m}}$. Then, the inequality chains in \eqref{cftau} hold with $r_0=r_{C;p}\left(\tau_\cV\right)$ and $C_0=C_{f;p}\left(\tau_\cV\right)$, where $p:=\frac{\dim\cV}S$. Furthermore, if some $K_c$ is also IO, then \eqref{crphi} holds for it.
\end{coro}
Our next main result follows from this in much the same elementary way as Theorem \ref{tfinmaxu} from Proposition \ref{pfinu}.
\begin{theorem}\label{tfinu}
    Any coherence\n non-creating channel that deterministically maps a rank-$S$ state $\rho^\rA$ to a standard coherent resource $\Psi_M$ with $M=pA/S$ must involve coherent measurement over at least $M^{-1}r_{C;p}\left(\tau_\rho\right)\ge M^{-1}\exp_2C_{f;p}\left(\tau_\rho\right)\ge\exp_2\ell_{;p}\left(\tau_\rho\right)$ elements of $\rA$'s incoherent basis, where $\tau_\rho:=\eins_\rho/S$.
\end{theorem}
\begin{remm}
    Our results on maximal distillation are in fact corollaries of the non-maximal versions. But we deemed the former to be of sufficient importance, and derivable through sufficiently simpler means, to warrant the order of presentation that we have chosen.
\end{remm}

\subsection{Approximate distillation}\label{secappr}
Let us now consider the task of approximate distillation, where we only require an output that is close enough to a standard resource. Formally, given an input $\rho^\rA$ and an error tolerance $\epsilon\in[0,1]$, we shall require that the action of a channel $\cE$ satisfy
\begin{equation}\label{cappr}
    F\left[\cE(\rho),\Psi_M\right]\ge1-\epsilon,
\end{equation}
where
\begin{equation}
    F(\sigma,\tau):=\left(\tr\sqrt{\sqrt{\sigma}\tau\sqrt{\sigma}}\right)^2
\end{equation}
is the \emph{Uhlmann\n Jozsa fidelity}. Since $\Psi_M$ is pure, the condition \eqref{cappr} simplifies as $F\left[\cE(\rho),\Psi_M\right]=\bra{\Psi_M}\cE(\rho)\ket{\Psi_M}=\tr\left(\rho T_\cE\right)$. Thus,
\begin{equation}\label{ferp}
    \tr\left(\rho T_\cE\right)\ge1-\epsilon.
\end{equation}
Notably, unlike in exact distillation where only the space $\mathrm{supp}\rho$ mattered, here the detailed structure of $\rho$ must be taken into account. Say $\rho=\sum_sr_s\psi_s$ is an eigendecomposition. Denoting $r_{\max/\min}:=\mathop{\max/\min}_sr_s$, we have $\eins_\rho=\sum_s\psi_s\ge\sum_s\left(r_s/r_{\max}\right)\psi_s=\rho/r_{\max}$. Therefore, \eqref{ferp} implies
\begin{equation}\label{pe11}
    \tr\left(\eins_\rho T_\cE\right)\ge\frac{1-\epsilon}{r_{\max}}.
\end{equation}
Meanwhile, another consequence of \eqref{ferp} is as follows: since each $s$ term is weighted by an $r_s$ factor, $\bra{\psi_s}T_\cE\ket{\psi_s}$ can afford to be further away from the ideal value 1 for those $s$ whose $r_s$ are smaller. A lower bound on the smallest possible single $\bra{\psi_s}T_\cE\ket{\psi_s}$ is $F_{\min}$, defined through $r_{\min}F_{\min}+\left(1-r_{\min}\right)\cdot1=1-\epsilon$. This is solved by $F_{\min}=1-\epsilon/r_{\min}$, and so
\begin{equation}\label{pe12}
    \tr\left(\eins_\rho T_\cE\right)\ge S\cdot\left(1-\frac{\epsilon}{r_{\min}}\right),
\end{equation}
where $S:=\mathrm{rank}\rho$ (note the departure from the notation of Section \ref{secnonm}). Combining \eqref{pe11} and \eqref{pe12}, we have
\begin{equation}\label{pe1}
    \tr\left(\eins_\rho T_\cE\right)\ge S_{\rho,\epsilon},
\end{equation}
where
\begin{equation}
    S_{\rho,\epsilon}:=\max\left\{\frac{1-\epsilon}{r_{\max}},\,S\left(1-\frac{\epsilon}{r_{\min}}\right)\right\}.
\end{equation}
When $\rho$ is nearly maximally-mixed on its support, i.e.\ $r_{\max}\approx r_{\min}\approx1/S$ (as in the asymptotic case that we will soon take up), the first bound is tighter: $\frac{1-\epsilon}{r_{\max}}\approx S\left(1-\epsilon\right)$, whereas $S\left(1-\frac{\epsilon}{r_{\min}}\right)\approx S\left(1-S\epsilon\right)$. The second bound is more useful when $\rho$ is far from maximally-mixed (i.e.,\ $r_{\max}\gg1/S$) and $\epsilon\ll r_{\min}$: then, $\frac{1-\epsilon}{r_{\max}}\ll S$, while $S\left(1-\frac{\epsilon}{r_{\min}}\right)\approx S$.

These constraints, together with those of Observation \ref{tbound}, can be used to find lower bounds on the coherent measurement cost, as we did in the exact case. There we had $T_\cE^\rho\equiv\eins_\rho T_\cE\eins_\rho=\eins_\rho$, leading to the exact block structure \eqref{tblock}. In the approximate case, for small enough $\epsilon$ we should be able to bound the amplitude of the cross-block parts. We leave this line of inquiry and the pursuit of ``good'' bounds for future work, here contenting ourselves with a crude bound that suffices for our analysis of maximal asymptotic distillation (Section \ref{secas}). For this bound, we will show that the normalized density operator $\tau_\cE:=T_\cE/\tilde S$ (where $\tilde S:=\tr\,T_\cE=A/M$) is ``not too different from'' $\tau_\rho:=\eins_\rho/S$.

Let $\tau_\cE^\rho:=\eins_\rho\tau_\cE\eins_\rho$, and define its normalized version $\tau_\cE^{|\rho}:=\tau_\cE^\rho/\tr\,\tau_\cE^\rho$. Note that $\tr\,\tau_\cE^\rho\ge S_{\rho,\epsilon}/\tilde S$. Then,
\begin{align}\label{fer1}
    F\left(\tau_\cE,\tau_\rho\right)&=\left(\tr\sqrt{\sqrt{\tau_\rho}\tau_\cE\sqrt{\tau_\rho}}\right)^2\nonumber\\
    &=\frac1S\left(\tr\sqrt{\eins_\rho\tau_\cE\eins_\rho}\right)^2\nonumber\\
    &=\frac1S\left(\tr\sqrt{\tau_\cE^\rho}\right)^2\ge\frac{S_{\rho,\epsilon}}{S\tilde S}\nrm{\sqrt{\tau_\cE^{|\rho}}}_1^2.
\end{align}
Since $T_\cE\le\eins$, also $T_\cE^\rho\le\eins$. Thus, $\tau_\cE^{|\rho}=T_\cE^\rho/\tr\,T_\cE^\rho\le\eins/S_{\rho,\epsilon}$. Meanwhile, $\nrm{\sqrt{\tau_\cE^{|\rho}}}_2=\sqrt{\tr\,\tau_\cE^{|\rho}}=1$. Therefore,
\begin{align}
    \nrm{\sqrt{\tau_\cE^{|\rho}}}_1&\nonumber\\
    \ge&\min_{\mathbf x\in\bbR^S}\left\{\nrm{\mathbf x}_1:\nrm{\mathbf x}_2=1,\nrm{\mathbf x}_\infty\le\frac1{\sqrt{S_{\rho,\epsilon}}}\right\}\nonumber\\
    \ge&\sqrt{S_{\rho,\epsilon}}.
\end{align}
Combining this with the bound in \eqref{fer1} and expressing the result in terms of the Bures distance $B(\sigma,\tau):=\sqrt{2\left[1-\sqrt{F(\sigma,\tau)}\right]}$,
\begin{equation}\label{fer}
    B\left(\tau_\cE,\tau_\rho\right)\le\sqrt{2\left(1-\frac{S_{\rho,\epsilon}}{\sqrt{S\tilde S}}\right)}=:\delta.
\end{equation}
This $\delta$ depends on $\epsilon$, $A$, $M$, and $\rho$; we suppress its dependencies to avoid clutter. Notice that $\delta$ approaches $0$ when $SM\approx A$ and $\epsilon\ll1/S$, which is the case we will encounter in maximal asymptotic distillation. We will skip an approximate analog to Proposition \ref{pfinu} and proceed directly to an analog to Theorem \ref{tfinmaxu}, which follows by applying the asymptotic continuity of the coherence of formation \cite{WY16} (Lemma \ref{cfcon} in Appendix \ref{appcont}) on \eqref{fer} and repeating the arguments of Theorem \ref{tfinmaxu}.
\begin{theorem}\label{tfina}
    Any coherence\n non-creating channel that deterministically maps a rank-$S$ state $\rho^\rA$, satisfying $r_{\min}\eins^\rA\le\rho^\rA\le r_{\max}\eins^\rA$, to an output $\sigma^\rM$ such that $F\left(\sigma,\Psi_M\right)\ge1-\epsilon$ for $M=A/\tilde S$ must involve coherent measurement over at least
    \begin{equation}
        \frac{\exp_2\left[C_f\left(\tau_\rho\right)-\delta\log_2A-(1+\delta)h\left(\frac\delta{1+\delta}\right)\right]}M
    \end{equation}
    elements of $\rA$'s incoherent basis, where $\tau_\rho:=\eins_\rho/S$ and $\delta:=\sqrt{2\left(1-\frac{S_{\rho,\epsilon}}{\sqrt{S\tilde S}}\right)}$
    with
    \begin{equation}
        S_{\rho,\epsilon}:=\max\left\{\frac{1-\epsilon}{r_{\max}},\,S\left(1-\frac{\epsilon}{r_{\min}}\right)\right\}.
    \end{equation}
\end{theorem}
This bound is good for the case where $M$ is near-maximal, i.e.\ $\tilde S\approx S$. We leave a more careful analysis of the non-maximal approximate case for future work. Moving on, we shall derive a slight variant that performs well for nearly\n maximally-mixed $\rho$; it will simplify our task in the asymptotic case.

Our approach above was to show that $\tau_\cE$ is close to $\tau_\rho$. When $\rho$ is close to maximally-mixed on its support, we can use $\rho$ itself instead of $\tau_\rho$. Noting that $\rho\ge r_{\min}\eins_\rho$, we can modify \eqref{fer1} to
\begin{align}\label{fer2}
    F\left(\tau_\cE,\rho\right)&\ge r_{\min}\left(\tr\sqrt{\eins_\rho\tau_\cE\eins_\rho}\right)^2\nonumber\\
    &\ge\frac{r_{\min}S_{\rho,\epsilon}}{\tilde S}\nrm{\sqrt{\tau_\cE^{|\rho}}}_1^2\ge\frac{r_{\min}S_{\rho,\epsilon}^2}{\tilde S}.
\end{align}
Repeating the rest of the steps as above, we have:
\begin{alem}\label{lmm}
    An MIO channel mapping a state $\rho^\rA$, satisfying $r_{\min}\eins^\rA\le\rho^\rA\le r_{\max}\eins^\rA$, to an output $\sigma^\rM$ such that $F\left(\sigma,\Psi_M\right)\ge1-\epsilon$ for $M=A/\tilde S$ must involve coherent measurement over at least
    \begin{equation}
        \frac{\exp_2\left[C_f\left(\rho\right)-\delta\log_2A-(1+\delta)h\left(\frac\delta{1+\delta}\right)\right]}M
    \end{equation}
    elements of $\rA$'s incoherent basis, where
    \begin{equation}
        \delta:=\sqrt{2\left(1-S_{\rho,\epsilon}\sqrt{\frac{r_{\min}}{\tilde S}}\right)}
    \end{equation}
    with $S_{\rho,\epsilon}:=\left(1-\epsilon\right)/r_{\max}$.
\end{alem}
\begin{remm}\label{remtrace2}
    Our analysis of approximate distillation incorporated only the overall trace condition on $T_\cE$; we have not found a way to work its parent constraint $\Delta\left(T_\cE\right)=\eins^\rA/M$ (resulting from Observation \ref{tbound}) into the reckoning. Considering the level of detail this would add, we expect its inclusion to significantly improve the relevant results (including those pertaining to asymptotic distillation, studied in the next section). On the other hand, our reliance solely on the trace condition expands the scope of these results' applicability beyond MIO (see Remark \ref{remtrace}).
\end{remm}

\section{Asymptotic distillation}\label{secas}
The asymptotic limit (see Section \ref{secclue} for a brief background) is an important window into a resource theory. The behaviour of resource-theoretic quantities in this limit is aptly compared with the classic laws of thermodynamics: asymptotic equipartition leads to certain near-universal features across diverse resource theories, such as extensivity of resource distillation yields and formation costs. We will now present some evidence suggesting the extensivity of the coherent measurement cost (quantified as the requisite number of elementary coherent gates such as the qubit Hadamard gate) of asymptotically maximal coherence distillation.

Chitambar \cite{chitambar2018dephasing} showed that the resource theory of coherence is asymptotically reversible under the class of free operations called dephasing-covariant operations (DIO), with the asymptotic rate of interconversion given by the relative entropy of coherence $C_r$. A consequence of this fact is that the rate of distillation of copies of $\Psi_2$ from those of a state $\rho$ under even the largest class of free operations\m the coherence\n non-creating operations MIO\m is bounded above by $C_r(\rho)$. Both DIO and IO, although strict subclasses of MIO, achieve this distillation rate \cite{WY16}.

For MIO achieving this maximal rate, we have the following conjectures.
\begin{conjm}\label{pan}
    Suppose a sequence $\cE_n$ of MIO channels is maximally-distilling on copies of an input $\rho$. That is, $\cE_n$ acting on $\varrho_n\equiv\rho^{\otimes n}$ achieves $F\left[\cE_n\left(\varrho_n\right),\Psi_{M_n}\right]\ge1-o(1)$ for $\log_2M_n=n\left[C_r(\rho)-o(1)\right]$. Then any Kraus operator decomposition of $\cE_n$ involves coherent measurement over at least $L_n$ elements of the input's incoherent basis, where
    \begin{equation}
        \log_2L_n\ge n\left[\ell(\rho)-o(1)\right].
    \end{equation}
    This rank bound applies both to the single most coherent measurement element and to the \emph{average} (of the rank's logarithm) over all involved measurements under the distribution induced by the input.
\end{conjm}
\begin{conjm}\label{pas}
    For any $\rho$, there exists a sequence $\cE_n$ of maximally-distilling \emph{IO} channels with all but an asymptotically-vanishing fraction of measurements \emph{individually} attaining the bound of Conjecture \ref{pan}, as well as attaining it on average.
\end{conjm}

\subsection{Towards bounding the asymptotic cost}\label{secgenn}
We will now attempt to prove Conjecture \ref{pan}, essentially through a formalization of our crude typicality-based statements in Section \ref{secclue}. The crudeness, when examined deeper, turns out unfortunately to conceal some finer features of asymptotic typicality that thwart our efforts at completing the proof. Nevertheless, we hope that the following account of our proof attempt elicits a more complete treatment from the community, thereby either proving or refuting our conjecture.
\begin{remm}\label{boxplus}
    In the following, we will make use of the triangle inequality property of the \emph{Fubini\n Study metric} (also called the \emph{fidelity angle}) $\theta(\rho,\sigma):=\arccos\sqrt{F(\rho,\sigma)}$; namely,
    \begin{equation}
        \theta(\rho,\sigma)+\theta(\sigma,\tau)\ge\theta(\tau,\rho).
    \end{equation}
    In our calculations, the approximation parameters will be associated with $1-F(\cdot,\cdot)$, whereby they are the \emph{squared sines} of the associated fidelity angles. For convenience in applying the angle triangle inequality in their terms, we will use the shorthand
    \begin{align}
        \epsilon\boxplus\delta:=&\sin^2\left(\arcsin\sqrt\epsilon+\arcsin\sqrt\delta\right)\nonumber\\
        =&\left(\sqrt{\epsilon(1-\delta)}+\sqrt{\delta(1-\epsilon)}\right)^2.
    \end{align}
    Note that $\epsilon\boxplus\delta\le\left(\sqrt{\epsilon}+\sqrt{\delta}\right)^2\le4\max\{\epsilon,\delta\}$ in general; but if $\delta\to0$ while $\epsilon$ is held fixed, $\epsilon\boxplus\delta\to\epsilon$.
\end{remm}
\begin{proof}[Proof sketch for Conjecture \ref{pan}]\renewcommand{\qedsymbol}{\framebox{\tiny?}}
    We shall build on Lemma \ref{lmm}, with\footnote{We use quotes to indicate a variable from Lemma \ref{lmm} whose analog here has a different notation, and whose Lemma \ref{lmm} notation may here mean something different.}${}^,$\footnote{Any $\epsilon_n$'s we introduce shall be understood to be asymptotically-vanishing as $n\to\infty$ and the relevant AEP-related $\delta$'s $\to0$.}
    \begin{equation}
        \text{``}M\text{''}\equiv M_n=\exp_2\left(n\left[C_r(\rho)-\epsilon_n^{(0)}\right]\right).
    \end{equation}
    To estimate the other relevant parameters, we will use the quantum asymptotic equipartition property (AEP) reviewed in Appendix \ref{apptyp}.

    For some $\delta_\rS>0$, let $\varrho_n^{\delta_\rS}$ denote the unnormalized projection of $\varrho_n$ onto its $\delta_\rS$\n weakly-typical subspace, and $\varrho_n^{|\delta_\rS}$ the normalized version thereof; the analog to ``$\rho$'' will be $\varrho_n^{|\delta_\rS}$. Applying Lemma \ref{lqaep}, $\tr\varrho_n^{\delta_\rS}\ge1-\delta_\rS$ and $2^{-n\left[S(\rho)+\delta_\rS\right]}\eins\le\varrho_n^{\delta_\rS}\le2^{-n\left[S(\rho)-\delta_\rS\right]}\eins$, so that
    \begin{equation}\label{rhosbounds}
        2^{-n\left[S(\rho)+\delta_\rS\right]}\eins\le\varrho_n^{|\delta_\rS}\le\frac{2^{-n\left[S(\rho)-\delta_\rS\right]}\eins}{1-\delta_\rS}.
    \end{equation}
    The lower bound above functions as ``$r_{\min}$''; we will presently define ``$r_{\max}$'', slightly differently from how we did in Lemma \ref{lmm}. Meanwhile, let us apply typicality on the classical sequences $\vect a\equiv\left(a_1\dots a_n\right)$ formed by the incoherent basis labels occurring in $\varrho_n$, which are distributed according to $\Delta\left(\varrho_n\right)\equiv\left[\Delta(\rho)\right]^{\otimes n}$. For any $\delta_\rA>0$ we can identify the $\delta_\rA$\n weakly-typical subalphabet $\cA_n^{\delta_\rA}$. The corresponding system $\rA_n^{\delta_\rA}$ will be the analog to Lemma \ref{lmm}'s ``$\rA$''; by the AEP, its dimensionality
    \begin{equation}\label{andelta}
        A_n^{\delta_\rA}\le\exp_2\left[n\left(S\left[\Delta(\rho)\right]+\delta_\rA\right)\right],
    \end{equation}
    yielding the bound
    \begin{equation}
        \text{``}\tilde S\text{''}=A_n^{\delta_\rA}/M_n\le\exp_2\left(n\left[S(\rho)+\epsilon_n^{(0)}+\delta_\rA\right]\right).
    \end{equation}
    We will now work towards bounding ``$S_{\rho,\epsilon}$''.
    
    Let $\varrho_n^{\delta_\rA}$ denote the unnormalized projection of $\varrho_n$ on this subalphabet, and $\varrho_n^{\setminus\delta_\rA}$ that on the complement thereof, so that $\tr\left(\varrho_n^{\delta_\rA}+\varrho_n^{\setminus\delta_\rA}\right)=1$. We shall extend this superscript notation to projections of any operator. Furthermore, we will use superscripts combining ``$\delta_\rS$'', ``$\delta_\rA$'', and ``$\setminus$'' to denote the results of successive projections from the inside out. For example, ``$\setminus\delta_\rS,\delta_\rA$'' will denote a projection on the complement of the $\delta_\rS$-typical subspace followed by one on the $\delta_\rA$-typical subalphabet. As before, we will denote the corresponding normalized density operators by preceding the superscripts with ``$|$''. With this notational arrangement, first note that $\varrho_n^{\delta_\rS}+\varrho_n^{\setminus\delta_\rS}=\varrho_n$, since these projections are defined via $\varrho_n$'s eigenspaces. Note also that the $\delta_\rA$-projections commute with the dephasing channel $\Delta$, so that $\tr\varrho_n^{\setminus\delta_\rA}=\tr\Delta\left(\varrho_n^{\setminus\delta_\rA}\right)=\tr\left[\Delta\left(\varrho_n\right)\right]^{\setminus\delta_\rA}\le\delta_\rA$, the last inequality following from the properties of the $\delta_\rA$-typical subalphabet. Using these facts,
    \begin{align}\label{tdsda}
        \tr\left(\varrho_n^{\delta_\rS,\setminus\delta_\rA}+\varrho_n^{\setminus\delta_\rS,\setminus\delta_\rA}\right)=\tr\varrho_n^{\setminus\delta_\rA}&\le\delta_\rA\nonumber\\
        \Rightarrow\tr\varrho_n^{\delta_\rS,\setminus\delta_\rA}&\le\delta_\rA\nonumber\\
        \Rightarrow\tr\varrho_n^{\delta_\rS,\delta_\rA}=\tr\left(\varrho_n^{\delta_\rS}-\varrho_n^{\delta_\rS,\setminus\delta_\rA}\right)&\ge1-\delta_\rS-\delta_\rA.
    \end{align}
    As a step towards determining ``$S_{\rho,\epsilon}$'', we shall now show that $\varrho_n$ is close to the normalized version $\varrho_n^{|\delta_\rS,\delta_\rA}$ of the operator in the last line above.
    \begin{align}\label{rnrndsa}
        F\left(\varrho_n,\varrho_n^{|\delta_\rS,\delta_\rA}\right)&=\left(\tr\varrho_n^{\delta_\rS,\delta_\rA}\right)^{-1}F\left(\varrho_n,\varrho_n^{\delta_\rS,\delta_\rA}\right)\nonumber\\
        &=\left(\tr\varrho_n^{\delta_\rS,\delta_\rA}\right)^{-1}F\left(\varrho_n^{\delta_\rA},\varrho_n^{\delta_\rS,\delta_\rA}\right)\nonumber\\
        &\ge\left(\tr\varrho_n^{\delta_\rS,\delta_\rA}\right)^{-1}\left(\frac{\tr\varrho_n^{\delta_\rS,\delta_\rA}}{\tr\varrho_n^{\delta_\rA}}\right)^2\nonumber\\
        &\ge1-\delta_\rS-\delta_\rA,
    \end{align}
    where we again used $\varrho_n=\varrho_n^{\delta_\rS}+\varrho_n^{\setminus\delta_\rS}$, this time to apply point \ref{fidcon} of Appendix \ref{appfid} on $\varrho_n^{\delta_\rA}=\varrho_n^{\delta_\rS,\delta_\rA}+\varrho_n^{\setminus\delta_\rS,\delta_\rA}$. In a minor variation of the method of Lemma \ref{lmm}, we will make $\nrm{\varrho_n^{|\delta_\rS,\delta_\rA}}_\infty$ the analog of ``$r_{\max}$''\m which is possible since we identify Lemma \ref{lmm}'s ``$\rA$'' with $\rA_n^{\delta_\rA}$. To bound this quantity, we first note that
    \begin{equation}
        \varrho_n^{\delta_\rS,\delta_\rA}\le\varrho_n^{\delta_\rS}\le\exp_2\left(-n\left[S(\rho)-\delta_\rS\right]\right)
    \end{equation}
    by the quantum AEP. Since $\tr\varrho_n^{\delta_\rS,\delta_\rA}\ge1-\delta_\rS-\delta_\rA$ as noted in \eqref{tdsda}, this implies
    \begin{equation}\label{rmax}
        \varrho_n^{|\delta_\rS,\delta_\rA}\le\frac{\exp_2\left(-n\left[S(\rho)-\delta_\rS\right]\right)}{1-\delta_\rS-\delta_\rA}.
    \end{equation}
    By assumption, the sequence $\cE_n$ of channels achieves $F\left[\cE_n\left(\varrho_n\right),\Psi_{M_n}\right]\ge1-\epsilon_n$. By \eqref{rnrndsa}, $F\left(\varrho_n,\varrho_n^{|\delta_\rS,\delta_\rA}\right)\ge1-\delta_\rS-\delta_\rA$, so that (by contractivity and the fidelity-angle triangle inequality) $F\left[\cE_n\left(\varrho_n^{|\delta_\rS,\delta_\rA}\right),\Psi_{M_n}\right]\ge1-\epsilon_n^{(1)}$ with $\epsilon_n^{(1)}:=\epsilon_n\boxplus\left(\delta_\rS+\delta_\rA\right)$. The associated target effects $T_n\equiv T_{\cE_n}$ must therefore satisfy $\tr\left(T_n\varrho_n^{|\delta_\rS,\delta_\rA}\right)\ge1-\epsilon_n^{(1)}$. Using an argument similar to Lemma \ref{lmm}'s based on \eqref{rmax},
    \begin{align}\label{srhoe}
        \tr\;T_n^{\delta_\rA,\delta_\rS}&\ge\frac{1-\epsilon_n^{(1)}}{\nrm{\varrho_n^{|\delta_\rS,\delta_\rA}}_\infty}\nonumber\\
        &\ge\left[1-\epsilon_n^{(1)}\right]\left(1-\delta_\rS-\delta_\rA\right)2^{n\left[S(\rho)-\delta_\rS\right]}.
    \end{align}
    This bound is analogous to ``$S_{\rho,\epsilon}$''. Note the order of projections in $T_n^{\delta_\rA,\delta_\rS}$: we are essentially using $T_n^{\delta_\rA}$ as ``$T_\cE$'' and $\varrho_n^{\delta_\rS}$ as ``$\rho$'', whereby $T_n^{\delta_\rA,\delta_\rS}$ plays the role of ``$T_\cE^\rho$''. Putting the pieces together as in \eqref{fer2},
    \begin{widetext}
    \begin{align}\label{frrv}
        F\left(\tau_n^{|\delta_\rA},\varrho_n^{|\delta_\rS}\right)\ge\left[1-\epsilon_n^{(1)}\right]^2\left(1-\delta_\rS-\delta_\rA\right)^2\exp_2\left(-n\left[\epsilon_n^{(0)}+3\delta_\rS+\delta_\rA\right]\right)=:1-\epsilon_n^{(2)}.
    \end{align}
    \end{widetext}
    We have now arrived at our primary obstacle: for this $\epsilon_n^{(2)}$ to be a vanishing sequence, we would need the exponents to all vanish. The $n\epsilon_n^{(0)}$ is obviously menacing, as the definition of maximal distillation makes no stipulation whatsoever on how fast the $\epsilon_n^{(0)}$ must decay. As for the $n\delta$'s: naively, we might expect that eventually taking the limits $\delta_\rS,\delta_\rA\to0$ would at least get rid of these exponents. However, these parameters cannot be taken to zero \emph{independently of $n$}: for any given $\delta_\rS$, the bounds in \eqref{rhosbounds} are guaranteed only ``for large enough $n$''; likewise for $\delta_\rA$ and \eqref{andelta}. Indeed, the requisite $n$ to validate these bounds scales as $\delta^{-2}$, whereby $\exp_2\left(-n\delta\right)\sim\exp_2\left(-\sqrt n\right)$. This puts a definitive end to any prospects of a bound like \eqref{frrv} succeeding.
    
    Let us now pretend we did not encounter the above problem, and continue our proof sketch as if $\epsilon_n^{(2)}\in o(1)$. Using a projection-related property of the fidelity (point \ref{fidproj} of Appendix \ref{appfid}), $F\left(\varrho_n,\varrho_n^{|\delta_\rS}\right)=\tr\varrho_n^{\delta_\rS}\ge1-\delta_\rS$. Combining this with \eqref{frrv} through the angle triangle inequality,
    \begin{equation}\label{tndrn}
        F\left(\tau_n^{|\delta_\rA},\varrho_n\right)\ge1-\epsilon_n^{(3)},
    \end{equation}
    where $\epsilon_n^{(3)}:=\epsilon_n^{(2)}\boxplus\delta_\rS$. Finally, ``$\delta$'', for which we can conveniently use the same symbol, is given by
    \begin{equation}
        \delta:=\sqrt{2\left(1-\sqrt{1-\epsilon_n^{(3)}}\right)}.
    \end{equation}
    Thus, the coherent measurement rank for the action of $\cE_n$ on $\rA_n^{\delta_\rA}$ is no less than
    \begin{align}
        L_n&:=M_n^{-1}\exp_2\left[C_f\left(\tau_n^{|\delta_\rA}\right)\right]\nonumber\\
        &\ge\frac{2^{C_f\left(\varrho_n\right)-\delta\log_2A_n-(1+\delta)h\left(\frac\delta{1+\delta}\right)}}{M_n},
    \end{align}
    where $A_n\equiv A^n$, with $A:=\mathrm{rank}\Delta(\rho)$, is the dimensionality of the entire Hilbert space where $\varrho_n$ acts. By the additivity of the coherence of formation under tensor products, $C_f\left(\varrho_n\right)=nC_f\left(\rho\right)$, while $\log_2M_n=n\left[C_r\left(\rho\right)-\epsilon_n^{(0)}\right]=n\left[C_r\left(\rho\right)-o(1)\right]$. We now take the limit as both $\delta_\rS$ and $\delta_\rA$ approach zero, whereupon $\epsilon_n^{(1)}\to\epsilon_n$, $\epsilon_n^{(2)}\le2\epsilon_n+n\epsilon_n^{(0)}$, and $\epsilon_n^{(3)}\le2\epsilon_n+n\epsilon_n^{(0)}$ as well. Thus, if we resolve to ignore the $n\epsilon_n^{(0)}$, we get $\delta\le\sqrt{2\left(1-\sqrt{1-2\epsilon_n}\right)}\le\sqrt{2\epsilon_n}\in o(1)$, and so
    \begin{align}
        L_n&\ge\exp_2\left(n\left[C_f\left(\rho\right)-C_r(\rho)-o(1)\right]-o[1]\right)\nonumber\\
        &\ge\exp_2\left(n\left[\ell\left(\rho\right)-o(1)\right]\right).
    \end{align}
    Finally, in light of \eqref{tndrn}, we conclude that the bound applies asymptotically also to the \emph{average} logarithmic measurement coherence rank under the distribution induced by the input $\varrho_n$.
\end{proof}
We attempted to prove the conjecture via the asymptotic continuity of the coherence of formation (Lemma \ref{cfcon} in Appendix \ref{appcont}). As such, we set ourselves the tall order of showing that the normalized target effect $\tau_n$ is close in fidelity to $\varrho_n$. In retrospect, the expectation that factors as large as $2^{nS(\rho)}$ mutually cancel to leave something close to 1 was naively optimistic: the normalization involves exponential factors with the AEP-related $\delta$'s, not to mention the $\epsilon_n^{(0)}$ from $M_n$ that we already had to contend with. Indeed, even if $\tau_n$ were exactly proportional to $\eins_{\varrho_n^{\delta_\rS}}$\m which is the most complete characterization we got in the non-asymptotic case\m we would then be required to show that the latter is close to $\varrho_n^{\delta_\rS}$. While AEP does guarantee that $\varrho_n^{\delta_\rS}$ gets rather flat in its spectrum, it is not nearly flat enough to have high fidelity with the maximally-mixed state on its support.

A more direct approach might try to avoid having to precisely trade exponentially large factors. For example, there might be a modified ``asymptotic continuity'' property that holds for density operators that are not necessarily close in fidelity. Of course there can be no such property for fully general pairs of operators $\sigma^{\rA_n}$ and $\tau^{\rA_n}$. But in our context, the operators come from sequences of $T_n\le\eins$ and $\varrho_n=\rho^{\otimes n}$ that satisfy $\tr\left(T_n\varrho_n\right)\to1$; the $T_n$ also satisfy $\tr T_n^{\delta_\rA}\approx A_n^{\delta_\rA}/M_n$ (which we used) and the stronger condition (following from Observation \ref{tbound}) $\bra{\vect a}T_n\ket{\vect a}=M_n^{-1}$ for all $\vect a\in\cA_n$ (which could yet be incorporated, though we could not find a way to).

This latter property reveals an interesting connection between the target effect and the generalized quantum Stein's lemma, an important result in quantum statistics and the general overarching theory of resource theories. The lemma was initially believed to have been proved by Brand\~ao and Plenio \cite{brandao2010generalization}, but a gap in their proof was discovered by Berta \textit{et al.}\ \cite{berta2023gap}, who were able to prove the lemma by alternate means for a subclass of resource theories, including those of coherence. Remarkably, Hayashi and Yamasaki \cite{hayashi2024generalized} and Lami \cite{lami2024solution} independently proved the lemma recently in full generality. The connection to our problem, evident from the previous paragraph's discussion, is as follows:
\begin{obs}\label{obstein}
    Given an MIO sequence $\cE_n$ that asymptotically distills maximally from copies of an input $\rho$, the sequence of associated target effects $T_n$ saturates the type-II error exponent scaling bound set by the generalized quantum Stein's lemma. In fact, $\tr\left(\sigma T_n\right)=M_n^{-1}$ identically for all free states $\sigma$, with $M_n$ saturating the scaling bound. Consequently, the coherent measurement cost (quantified by any measure) of maximal asymptotic distillation is bounded below by the infimum of the values attained by the measure over the set of hypothesis-testing effect sequences saturating the lemma's bound.
\end{obs}
All this additional structure could admit a modified notion of asymptotic continuity. As a signature of the special structure, we can observe from our proof sketch that $F\left(\tau_n,\varrho_n\right)$ can only decay sub-exponentially, which is not a generic property among pairs of state sequences. Unfortunately, this property alone is not enough to ensure that $C_f\left(\tau_n\right)$ differs at most sub-extensively from $C_f\left(\tau_n\right)$: a simple counterexample is $\tau_n=\epsilon_n\varrho_n+\left(1-\epsilon_n\right)\sigma_n$ for any coherent $\rho$, incoherent $\sigma_n$, and sub-exponentially\n decaying $\epsilon_n$. All the same, considered together with the other properties discussed above, it lends credence to our conjecture.

Our problem might also benefit from the tools and methods of the smooth entropy calculus \cite{tomamichel2015quantum}. These tools are often used to derive AEPs for entropy-like quantities (see, e.g.,\ \cite{tomamichel2009fully}). An example relevant to our problem is the AEP
\begin{equation}\label{zhaoaep}
    \lim_{\epsilon\to0^+}\lim_{n\to\infty}\frac{\log_2r_C^\epsilon\left(\rho^{\otimes n}\right)}n=C_f(\rho),
\end{equation}
due to Zhao \textit{et al.} \cite{zhao2018one}; here, $r_C^\epsilon$ is a so-called \emph{smoothed} variant of $r_C$, with $\epsilon$ a real smoothing parameter. This property is reflected in one of the ``flattening'' steps we are able to accomplish (point \ref{flc}) in our construction in Section \ref{howto}. Indeed, $C_f(\rho)$ is a good candidate for a universal AEP limit for all reasonable coherence measures based on optimizations over convex decompositions. The bounds encountered in our problem are all measures of this kind, though evaluated not on the input state $\varrho_n\equiv\rho^{\otimes n}$ itself, but instead on the maximally mixed state $\tau_{\varrho_n}$ on its support (or variations thereof). Yet, considering that $\varrho_n$ flattens out spectrally as $n$ grows, it is not unreasonable to expect $C_f(\rho)$ to emerge as the limiting rate even for measures evaluated on $\tau_{\varrho_n}$.

A possible way to put such cases on an equal footing with those like \eqref{zhaoaep} is to use a normalized version of ${\varrho_n}^\alpha$, for $\alpha\in\bbR$, as the argument; in this landscape, $\tau_{\varrho_n}\propto{\varrho_n}^0$. We could suitably define smoothed counterparts of the measures; Definition \ref{defcp} is an example of a possible way to effect the smoothing. It is plausible then that there is an AEP of the form
\begin{equation}
    \lim_{\epsilon\to0^+}\lim_{n\to\infty}\frac1nC^\epsilon\left(\frac{{\varrho_n}^\alpha}{\tr{\varrho_n}^\alpha}\right)=C_f(\rho)
\end{equation}
for a whole class of measures $C$. If true, this could give us a way of relating our $r_C\left(\tau_n\right)$ with $C_f\left(\rho\right)$, thereby potentially proving Conjecture \ref{pan}.

Of course, it is a different matter that all such quantities\m thanks to their involving optimizations over all convex decompositions\m would likely be hard to compute. Nevertheless, there might still be a way to \emph{formally} establish an AEP and relate all of these (computationally intractable) measures evaluated on $\varrho_n$ to the (also computationally intractable, but better-understood) $C_f(\rho)$.

Finally, a caveat is also in order concerning the claim ``[t]his rank bound applies both to the single most coherent measurement element and to the \emph{average} (of the rank's logarithm) over all involved measurements under the distribution induced by the input'' in Conjecture \ref{pan}: the failure of \eqref{tndrn} deals a blow to the average clause, even were the rest of the conjecture valid. Nevertheless, by virtue of the logarithms, we remain hopeful that any future technique capable of proving the latter would also be up to the more exacting task of proving the whole conjecture.

\subsection{How to attain the bounds}\label{howto}
From putting a lower bound on the asymptotic coherent measurement cost, we now turn to attaining the bound (Conjecture \ref{pas}). One way of proving Conjecture \ref{pas} would be to explicitly construct bound-attaining distillation channels. We now formulate some general guiding principles towards such constructions. Let us take (for example) the result of Proposition \ref{pfin} and inspect the chain \eqref{cftau} of inequalities therein:
\begin{align}\label{cftau1}
    \log_2M&+\sum_{c,m}\frac{\alpha_m^2t_{c,m}}{\tilde S}C_r\left(\phi_{c,m}\right)\nonumber\\
    \ge\sum_{c}&\frac{t_c\left(H\left[\left(\frac{\alpha_m^2t_{c,m}}{t_c}\right)_m\right]+\sum\limits_{m}\frac{\alpha_m^2t_{c,m}}{t_c}C_r\left[\phi_{c,m}\right]\right)}S\nonumber\\
    =\sum_{c}&\frac{t_c}{S}C_r\left(\phi_{c}\right)\ge C_f\left(\tau_\cE\right)\ge C_0.
\end{align}
For the last inequality to be saturated, we need $\dim\cV=S$\m i.e.,\ our ``maximal distillation'' condition $M\dim\cV=A$. For the one before, the decomposition $\sum_{c}\left(t_c/\tilde S\right)\phi_{c}$ of $\tau_\cE$ must be one that attains the bound in the definition of $C_f\left(\tau_\cE\right)$\m a so-called optimal decomposition. Finally, the inequality in the middle line is saturated when $\left(\alpha_m^2t_{c,m}/t_c\right)_m$ is uniform for each $c$. This is the case for a maximal uniform distillate, i.e.\ $\ket\alpha=\ket{\Psi_M}$, as mentioned in Remark \ref{tcmtc}; we are not aware of weaker conditions where it still holds.

These conditions would suffice, in principle, for the loosest bound of Proposition \ref{pfin} (effectively Proposition \ref{pfinu}, considering the above) to be attained without necessarily saturating the inequality in Theorem \ref{tfinmaxu}. The latter would further require
\begin{enumerate}
    \item\label{how1} each of the $C_r\left(\phi_{c}\right)$ to be equal to $C_f\left(\tau_\rho\right)$;
    \item\label{how2} each of the $C_r\left(\phi_{c,m}\right)$ to be equal, in turn, to $C_f\left(\tau_\rho\right)-\log_2M$; and moreover,
    \item\label{how3} each $\phi_{c,m}$ to be a uniform superposition.
\end{enumerate}
Condition \ref{how1} occurs when $\tau_\rho$ is a so-called \emph{flat-roof point} for the convex-roof function in question \cite{uhlmann2010roofs}. There does not seem to be any work in the literature on flat-roof points for the function $C_f$. As for the further conditions \ref{how2} and \ref{how3}, suppose
\begin{itemize}
    \item \ref{how1} holds;
    \item the Kraus operators are IO;
    \item each $\left(\alpha_m^2t_{c,m}/t_c\right)_m$ is uniform; and furthermore,
    \item each $\phi_c$ is a uniform superposition.
\end{itemize}
This would automatically ensure both \ref{how2} and \ref{how3}; once again, we are not aware of situations where these conditions could be met without those in the previous sentence. Finally, attaining the rank bound $\ell\left(\tau_\rho\right)$ would necessitate $C_r\left(\tau_\rho\right)=\log_2M$\m i.e.,\ the distillate is also maximal in terms of $C_r$.

For each inequality in these chains, the condition for its saturation is essentially some sort of ``flattening'' of features. We are not aware of any noteworthy situations in which some of the bounds are tight while others are not; in particular, whether the ``average coherence'' bound of Proposition \ref{pfinu} can be attained in cases where the maximal rank bound of Theorem \ref{tfinmaxu} cannot is an intriguing open question. There is one situation, though, where all manner of flattening tendencies collude: maximal asymptotic distillation. We use the above observations to attempt a systematic construction (summarized towards the end of Section \ref{secsum}) of a distillation protocol that asymptotically attains the $\ell$-based rank bound of Conjecture \ref{pan}. The construction is summarized as follows:
\begin{enumerate}
    \item\label{stepWY} In our proof sketch for Conjecture \ref{pan}, we attempted to show that the target effect\n based density operators $\tau_n\equiv\tau_{\cE_n}$ associated with any sequence of asymptotically maximally-distilling MIO channels $\cE_n$ must approach $\varrho_n$ [see \eqref{tndrn}]. Any given set of Kraus operators of $\cE_n$ correspond to the pure states in a certain convex decomposition of $\tau_n$ (Observation \ref{TgivesF}), which would then be an approximate convex decomposition of $\varrho_n$. For the converse, we adapt Winter and Yang's maximal IO distillation protocol \cite{WY16} to derive conditions under which a given approximation sequence $\tau_n$ corresponds to some maximally-distilling channel sequence $\cE_n$; we also show that \emph{any} convex decomposition of such a $\tau_n$ yields \emph{IO} Kraus operators for $\cE_n$. We start our construction with $\varrho_n$ itself, progressively working towards a viable $\tau_n$.
    \item\label{stepcf} If we choose an optimal decomposition attaining $C_f\left(\varrho_n\right)$ from the convex roof of $\varrho_n$, we already obtain Kraus operators that attain the measurement coherence bound in the \emph{average} sense of Proposition \ref{pfinu}. It remains to flatten out further to attain the bound on a per-$\ket{\phi_{c,m}}$ basis.
    \item\label{flc} First, to flatten relative to $c$, we construct a decomposition wherein all but an asymptotically-vanishing weight is carried by pure states whose individual $C_r$ values are close to $C_f\left(\varrho_n\right)$: simply take a decomposition $\rho=\sum_jq_j\phi_j$ that is optimal for $\rho$, and decompose (most of) $\varrho_n$ into pure states of the form $\ket{\Phi_{c\equiv\vect j}}\equiv\bigotimes_{k=1}^n\ket{\phi_{j_k}}$ where $\vect j\equiv\left(j_k\right)_k$ is a strongly-typical sequence under $\vect q^{\otimes n}$ (see Appendix \ref{apptyp} for background on asymptotic typicality)\footnote{This works except when $\vect q$ is uniform. For this case, we can use this construction for arbitrarily small non-uniform perturbations of $\vect q$, resulting in corresponding perturbations of $\rho$. The properties of convex roofs ensure that the perturbed $\vect q$ decomposition is optimal for the perturbed $\rho$ \cite{regula2017convex}.}. These pure components then have $C_r\left(\Phi_{\vect j}\right)\approx C_f\left(\varrho_n\right)=nC_f(\rho)$.
    \item\label{typs} To ensure that the final target effect is bounded as $T_n\le\eins$ (as required for the associated map to be a valid subchannel), we project all remaining $\ket{\Phi_{\vect j}}$ onto a strongly-typical subspace of $\varrho_n$, resulting in $\ket{\Phi_{\vect j}^{\delta_\rS}}$.
    \item\label{flm1} To flatten relative to $m$, we first show that each $\ket{\Phi_{\vect j}^{\delta_\rS}}$ can be made arbitrarily close to a near-uniform superposition $\ket{\Phi_{\vect j}^{\delta_\rS,\delta_\rA}}$.
    \item\label{flm2} We then draw again on Winter\n Yang's construction to show that we can discard an asymptotically-vanishing fraction of the remaining vectors to leave only ones almost entirely contained (with an asymptotically-vanishing error) in a subspace $\cV$ with the following property: any $\ket w\in\cV$ can be decomposed as $\ket w=M_n^{-1/2}\sum_m\ket{w_m}$, with $\bkt{w_m}=\bkt w$ $\forall m$, such that $\ket{w_m}\in\mathrm{span}\left\{\ket{\vect a}:\;\vect a\in \cI_m\right\}$, where $\left\{\cI_m\right\}_m$ is a \emph{fixed} (i.e.,\ $\ket w$-independent) disjoint partitioning of $\cA_n\equiv \cA^n$. The latter property enables the vectors to be used in constructing IO Kraus operators, while the former ensures that they are weighted equally over $m$. Thus, we now have $\ket{\Phi_{\vect j}^{\delta_\rS,\delta_\rA}}\approx M_n^{-1/2}\sum_m\ket{\Phi_{\vect j,m}^{\delta_\rS,\delta_\rA}}$ with $\bkt{\Phi_{\vect j,m}^{\delta_\rS,\delta_\rA}}$ close to a uniform distribution over $m$.
    \item\label{flm3} We then show that we can delete a vanishing fraction of $\ket{\Phi_{\vect j,m}^{\delta_\rS,\delta_\rA}}$ from each remaining $\ket{\Phi_{\vect j}^{\delta_\rS,\delta_\rA}}$ such that the remaining have logarithmic coherence rank tightly concentrated around $C_r\left(\Phi_{\vect j}^{\delta_\rS,\delta_\rA}\right)-\log_2M_n\approx n[C_f(\rho)-C_r(\rho)]=n\ell(\rho)$. We thus obtain IO Kraus operator fragments $\ket{w_{c,m}}\propto\ket{\Phi_{\vect j,m}^{\delta_\rS,\delta_\rA}}$ that individually approach the measurement coherence rank bound of Proposition \ref{pan}.
    \item\label{atp} Finally, we show that the ``IO Kraus operators'' (in quotes because the resulting map may fail to be trace\n non-increasing) constructed using these fragments asymptotically output an approximation to $\Psi_{M_n}$ near-deterministically. In the case that the map is a sub-channel, we show that it can be completed to a channel with an IO residual subchannel incurring a measurement coherence overhead asymptotically vanishing in relation to the rest.
\end{enumerate}
Appendix \ref{psc2} provides a detailed survey of the above construction and an attempt to thereby prove Conjecture \ref{pas}. The overall difficulties here are much more involved than in the case of Conjecture \ref{pan}. This suggests that Conjecture \ref{pas} may stand a bleaker chance of holding up. We now close this section and turn to the implications of our conjectures, were they to hold.

\section{Ramifications of our conjectures}\label{secram}
If true, conjectures \ref{pan} and \ref{pas} would pin down the coherent measurement cost of asymptotically maximal distillation, up to subextensive terms. We shall now explore some implications that would then follow. We will first discuss the operational implications for maximal distillation, and then sketch some possibilities for an asymptotic tradeoff between the coherent measurement budget and the distilled yield.

All statements in the remainder of this section will be premised on conjectures \ref{pan} and \ref{pas}.

\subsection{Maximal distillation may be a net loss}\label{dontdist}
Recall our observation from Section \ref{cohasp} that dynamical coherence can be used to endue both incoherent states and incoherent measurements with (respectively, static and mensural) coherence. Indeed, although the latter forms of coherence are well-defined operational primitives in the formalism, in practice we always derive them from dynamical coherence\m we never find ourselves in possession of the operational capability to prepare or measure coherently but not to implement coherent dynamics. Thus, it is justified to consider all aspects of coherence involved in a task as ultimately dynamical. In this spirit, we now weigh the coherence used in implementing coherent measurements (whose quantification has been our main preoccupation) against that eventually distilled in the form of standard resource states.

According to Conjecture \ref{pan}, the coherent measurements used in any maximal distillation protocol (on average, up to subextensive terms) would be equivalent to a number of Hadamard measurements no smaller than $n\ell(\rho)$, where $n$ is the number of copies of the input $\rho$ used. Meanwhile, the distilled yield is equivalent to $nC_r(\rho)$ cobits. Could the coherent measurement cost exceed the yield, i.e.\ $\ell(\rho)>C_r(\rho)$?

Indeed, this does sometimes happen. While no general method is known for efficiently computing the coherence of formation, it does have a closed-form expression for states of a single qubit \cite{YZCM15}:
\begin{equation}
    C_f(\rho)=h\left(\frac{1+\sqrt{1-4\abs{\bra0\rho\ket1}^2}}{2}\right).
\end{equation}
Using this formula, we computed the excess coherence cost $\ell(\rho)-C_r(\rho)$ for a representative sample of qubit states (Fig.\ \ref{fig:ellCr}). The results suggest that a nonzero measure of states incur an excess cost for maximal distillation, and moreover, that this excess cost can be an arbitrarily large multiple of the distilled yield!

Coherence distillation is sometimes an instrumental or incidental outcome in a larger task, e.g.\ entanglement distillation using incoherent local operations \cite{CH16}. But if the task is coherence distillation itself, then our results call into question its operational utility in situations where $\ell(\rho)\ge C_r(\rho)$: if we are able to implement $n\ell(\rho)$ Hadamard gates, we should rather use them to simply prepare as many fresh cobits than squander them in distilling only $nC_r(\rho)$.

\begin{figure*}
    \centering
    \includegraphics[width=.5\textwidth]{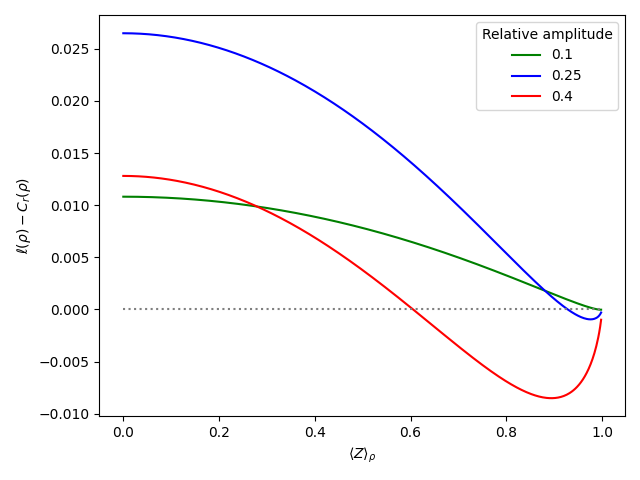}\includegraphics[width=.5\textwidth]{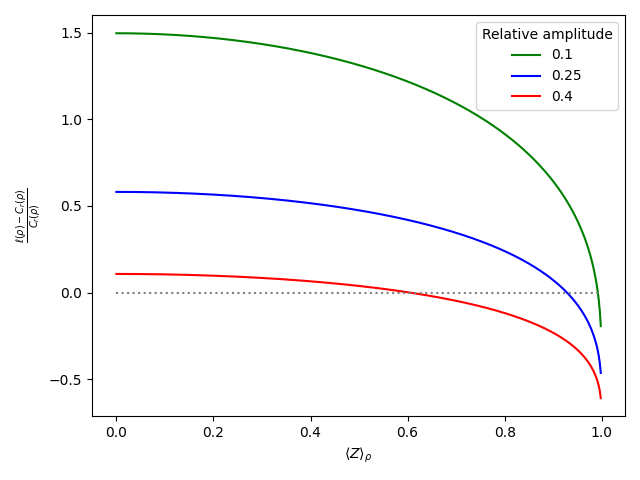}
    \caption{(Left) Excess coherent measurement cost, i.e.\ the difference between the asymptotic coherent measurement cost $\ell(\rho)$ and the distillable coherence $C_r(\rho)$, for qubit states with different values of the relative off-diagonal amplitude $\abs{\bra0\rho\ket1}/\sqrt{\bra0\rho\ket0\bra1\rho\ket1}$: the cost generically exceeds the yield for relative amplitude $\lesssim0.4$. (Right) The ratio of the excess coherent measurement cost $\ell(\rho)-C_r(\rho)$ to the distillable coherence $C_r(\rho)$: the ratio attains arbitrarily large values for low enough relative amplitude.}\label{fig:ellCr}
\end{figure*}

\subsection{Asymptotic cost\n yield tradeoff}\label{sectrade}
An obvious question that arises is what the coherent measurement cost of asymptotically\n\emph{non-}maximal distillation is. In this section we will make some conjectures on this question, restricting our consideration to IO. Since the coherence rank $M_n\approx\exp_2\left[n\,C_r(\rho)\right]$ of the maximal distillate and that of the conjectured requisite measurement, $L_n\approx\exp_2\left[n\,\ell(\rho)\right]$, both scale exponentially in the number of copies, one might naively expect that the best one can do with a measurement rank scaling $\tilde L_n\approx\left(L_n\right)^t$ with $0\le t<1$ is to distill maximally from a fraction $t$ of the input's copies, thereby achieving $\tilde M_n\approx\left(M_n\right)^t$, corresponding to the distillation rate $t\,C_r(\rho)$.

However, this would only be a loose lower bound on the achievable rate, as can already be seen by putting Conjecture \ref{pas} together with past results on SIO distillation \cite{lami_2019-1,lami_2019}. Recall that SIO are just IO that don't use coherent measurements, i.e.\ the case $t=0$. We know that the maximal SIO-distillable rate is nonzero for certain inputs whose maximal IO rate is strictly larger\m thus, in these cases, although $\log_2L_n\in\Omega(n)$ for maximal IO distillation (assuming Conjecture \ref{pan}), $\tilde L_n=1\in O\left[\left(L_n\right)^{t=0}\right]$ is nevertheless able to achieve $\Omega(n)\ni\log_2\tilde M_n\notin O\left[tn\,C_r(\rho)\right]$.

Can we go further and find the exact asymptotic tradeoff between the cost and the yield?
\begin{definition}[Feasible cost\n yield pair]
    $(l,r)\in\bbR^2$ is a \emph{feasible cost\n yield pair} for $\rho$ if there exists a sequence $\cE_n$ of IO channels that respectively distill from $\rho^{\otimes n}$ at the asymptotic rate $r$ and coherent measurement cost $l$, defined as
    \begin{equation}
        l=\lim_{n\to\infty}\frac{\log_2L_n}{n}
    \end{equation}
    with $L_n$ the cost of $\cE_n$ both on (logarithmic) average and on a per-measurement basis (for almost all measurements involved)\m the two being identical under Conjecture \ref{pas}.
\end{definition}
For example, $(l,r)=\left(\ell[\rho],C_r[\rho]\right)$ is a feasible pair for $\rho$ according to Conjecture \ref{pas}. Taking Theorem \ref{tfinu} and conjectures \ref{pan} and \ref{pas} as clues, we now make the informed guess that some quantity like the ones defined in Definition \ref{defcp} would quantify the necessary coherent measurement cost of non-maximal asymptotic distillation. If so, the freedom we have in optimizing a distillation strategy would consist of choosing an appropriate density operator $\tau_\perp$ orthogonal to $\varrho_n$, such as to minimize the mixture's $C_f$. Intuitively, for $\tau_\perp$ to help reduce the final $C_f$, its incoherent alphabet must overlap with $\varrho_n$'s as much as possible.

One possibility suggested by this intuition is to pick $\tau_\perp$ so that the mixture becomes $\sigma^{\otimes n}$ for some state $\sigma$ satisfying $\Delta(\sigma)=\Delta(\rho)$. In particular, define the set
\begin{equation}\label{drho}
    \cD(\rho):=\left\{\sigma:\;\Delta(\sigma)=\Delta(\rho),\;\rho\stackrel{\mathrm{SIO}}\longmapsto\sigma\right\}.
\end{equation}
As per Conjecture \ref{pas}, $(l,r)=\left(\ell[\sigma],C_r[\sigma]\right)$ for any $\sigma\in\cD(\rho)$ would also be a feasible pair for $\rho$, since each copy of $\rho$ could first be mapped to one of $\sigma$ without requiring coherent measurements. A special case of this is when $\sigma$ is Lami's $\bar\rho$ \cite{lami_2019-1,lami_2019}, the state consisting only of the pure diagonal blocks of $\rho$\m we then get the SIO-feasible pair $\left(\ell\left[\bar\rho\right],C_r\left[\bar\rho\right]\right)=\left(0,Q\left[\rho\right]\right)$, where $Q[\rho]$ is Lami's \emph{quintessential coherence}.

We can achieve any convex combination of a set of feasible pairs by performing the respective protocols achieving them on appropriate fractions of the total number of input copies. Finally, the feasibility of $(l,r)$ trivially implies that of any $\left(l,r_1<r\right)$. Thus, Conjecture \ref{pas} implies that any pair in the following set is feasible for $\rho$:
\begin{equation}\label{crho}
    \cC(\rho):=\mathrm{cvx}\left\{\left(\ell[\sigma],r\right):r\in\left[0,C_r(\sigma)\right],\sigma\in\cD(\rho)\right\}.
\end{equation}
The symmetry and typicality properties of the asymptotic limit suggest that there might be no nontrivial feasible pairs besides these. We discussed above why we believe it would not help to expand $\cD(\rho)$ to include states with a diagonal part different from the input's. As for the possibility of collective SIO pre-processing on many copies of $\rho$, Lami's results on SIO distillation hint against the prospect of this producing new feasible pairs. There would of course be IO strategies that deviate from maximal distillation in ways other than an SIO pre-processing. But based on our preliminary studies of distillation in the dilation picture (Appendix \ref{secdec}), we are inclined to believe that these variations do not expand the feasible set either. Hence, we have the following:
\begin{conjm}
    The $\cC(\rho)$ defined through \eqref{drho} and \eqref{crho} is the set of all feasible cost\n yield pairs for $\rho$.
\end{conjm}
In the above line of reasoning we considered using the coherence budget $l$ only in implementing an IO acting on the given input. But allowing the coherent action to be used possibly also in preparing fresh coherent states (as discussed in Section \ref{dontdist}) is arguably more operationally meaningful. In that case, $(l,l)$ should also be considered a feasible pair for any $l\ge0$ and any $\rho$. If we wish to include these in our reckoning, we could modify the definition \eqref{crho} by replacing $C_r(\sigma)$ with $\max\left\{\ell(\sigma),C_r(\sigma)\right\}$.

In general, the problem of computing the largest achievable rate $r$ for a given coherent measurement budget $l$, or inversely, that of computing the least $l$ achieving a desired rate $r$, is likely to be very hard (even if our conjectures hold)\m after all, we have no known efficient way to compute $\ell(\rho)$ itself, saying nothing of the difficulty of finding the set $\cD(\rho)$. But we expect there to be a nontrivial cost\n yield tradeoff landscape, finding which would be both fundamentally and operationally important.

\section{Conclusion}
We laid out a framework for quantifying the cost of coherent measurements involved in distilling coherent states from arbitrary inputs using coherence\n non-creating channels (MIO). The framework uses a construction we call the \emph{target effect}\m a measurement effect that quantifies the probability with which a given channel maps an input to a desired target state. This object also doubles up as a quantifier of the requisite coherent measurement action in implementing the channel. We derived conditions on the target effect\m and thereby, lower bounds on the coherent measurement cost\m for exact (maximal and non-maximal) and approximate maximal coherence distillation from finite-sized inputs.

Based on our result on the approximate case, we conjectured a scaling law for the coherent measurement cost of asymptotic distillation at the maximal rate\m namely, that the necessary and sufficient cost (in an equivalent number of qubit Hadamard gates) is extensive in the number of input copies, with a rate given by the input's irretrievable coherence $\ell(\rho)=C_f(\rho)-C_r(\rho)$. We went through detailed and rigorous proof sketches for both the necessity and the sufficiency of this cost, and discussed the difficulties we encountered in completing the proofs. As a byproduct, we showed a connection between our problem and the generalized quantum Stein's lemma. By virtue of this connection, the coherent measurement cost of maximal asymptotic distillation can be bounded by the coherence of hypothesis-testing effects that saturate the error exponent scaling bound in the lemma\m a potential tactic for future attempts at our conjectures.

We then discussed some implications of our conjectures. First, we noted that our conjectured coherent measurement cost exceeds the very distilled yield in a nonzero measure of qubit instances of maximal distillation. Thus, maximal coherence distillation as a standalone operational task may sometimes amount to a net attrition of coherence-resource and should therefore be superseded by coherent state preparation. We then made some speculations on an asymptotic tradeoff between the coherent measurement budget and distilled yield.

A question of possible interest for future work is whether the appearance of the irretrievable coherence\m a signature of irreversibility\m is accidental, or if, rather, there is a deeper connection between resource-theoretic irreversibility and auxiliary (or otherwise hidden) costs of distillation. The fact that our results all apply to MIO, which do not exhibit irreversibility, supports the accident hypothesis, but there may yet be some subtleties that we have missed.

Our results and conjectures have foundational significance in understanding the resource theory of coherence, but also operational implications in applications that use the resource. For example, combining our results with those of Ref.\ \cite{HFW21} directly implies that our coherent measurement cost bounds apply also to the task of private incoherent randomness extraction from the same input state; our Conjecture \ref{pan}, if true, would further imply that, given an operational coherence budget, more private randomness can sometimes be extracted by ignoring the input and instead simply preparing fresh coherent states (to be incoherently measured for generating randomness). Further investigation on our conjectures may shed light on information-theoretic aspects of coherence manipulation, e.g.\ a possible asymptotic equipartition property of convex-roof extensions of entropic coherence quantifiers. Developing methods and heuristics to actually compute such quantifiers would be a challenging project in itself. Other natural avenues for inquiry would be in situations involving the interplay of coherence with other resources, e.g.\ multipartite settings with local incoherent operations \cite{CH16}.

Naively, one might expect non-SIO IOs to be implementable as SIO augmented by the consumption of coherent states (a basic treatment of whose asymptotic distilling power was done by Lami \cite{lami_2019}). However, nontrivial coherent state\n augmented SIO may not even be IO\m at any rate, the effective Kraus operators induced by an SIO Kraus operator representation of the orignal SIO channel fail to be IO. This necessitates other methods, such as ours, for studying the coherent measurement cost of tasks.

Besides the target effect\n based approach detailed in this paper, we also attempted some alternative approaches towards quantifying the coherent measurement cost of distillation. In Appendix \ref{appsdp}, we present some preliminary results towards a semidefinite programming (SDP)\n based approach. In Appendix \ref{secmuk}, we describe an approach wherein we studied the behaviour under constrained IO of certain SIO monotones constructed by Lami \cite{lami_2019}, in the process defining a generalization of the measures. A better understanding of how Lami's monotones, and generalizations thereof, behave under tensor products (especially in the asymptotic limit) seems essential for studying the coherent measurement cost either directly through constrained IO (as we do in Appendix \ref{secmuk}) or by relating it with the dynamical coherence \cite{saxena2020dynamical,takahashi2022creating} of unitary dilations.

Lami relates these coherence measures to the information-theoretic properties of a classical variable labelling the pure diagonal blocks of the input state\m i.e.,\ classical information encoded in a certain way in coherent quantum states. Generalizations such as ours may be operationally related to corresponding generalized encoding tasks. This suggests as-yet unexplored connections between Shannon theory and the resource theory of coherence.

In Appendix \ref{secdec}, we present our preliminary investigations on what we call decoupling schemes: certain linear-algebraic structures that emerge when distillation is framed in the dilation picture. These suggest yet other fruitful lines of inquiry, including possible connections with established notions of decoupling \cite{Buscemi09,DBWR14,FR15,AJ19}.

\section*{Acknowledgements}
We thank Anurag Anshu, Kishor Bharti, Dagmar Bru\ss, Ian George, Ludovico Lami, Alessandro Luongo, Iman Marvian, Ryuji Takagi, Thomas Theurer and Benjamin Yadin for helpful discussions. Special thanks are due to Eric Chitambar, Mile Gu and Yunlong Xiao for particularly detailed and engaging discussions. This project was originally conceived during a collaboration with Francesco Buscemi and Bartosz Regula. We express our sincerest gratitude for their help through intense discussions during the early phase and occasional consultations thereafter.
\begin{widetext}

\end{widetext}

\appendix

\section{Useful properties of the fidelity}\label{appfid}
In this paper we adhere to the \emph{Uhlmann\n Jozsa} (i.e.,\ square root\n free) definition of the fidelity, $F(\sigma,\tau):=\left(\tr\sqrt{\sqrt{\sigma}\tau\sqrt{\sigma}}\right)^2$, which we will apply on arbitrary positive-semidefinite arguments (not only density operators). We will make use of the following properties, referring to this section when we do so:
\begin{enumerate}
    \item One or both arguments pure: $F(\psi,\tau)=\bra\psi\tau\ket\psi$.
    \item Multiplication by nonnegative scalar factors: $F(t\sigma,\tau)=tF(\sigma,\tau)$ for $t\ge0$.
    \item Joint concavity of the square-root fidelity: for $0\le p\le1$,
    \begin{align}
        &\sqrt{F\left[p\sigma_1+(1-p)\sigma_2,p\tau_1+(1-p)\tau_2\right]}\nonumber\\
        \ge&p\sqrt{F\left(\sigma_1,\tau_1\right)}+(1-p)\sqrt{F\left(\sigma_2,\tau_2\right)}.
    \end{align}
    \item\label{fidproj} Fidelity of an operator with a normalized projection thereof: for some space $\cV$, if $\sigma^\cV:=\eins_\cV\sigma\eins_\cV$ and $\sigma^{|\cV}:=\sigma^\cV/\tr\sigma^\cV$, then
    \begin{align}
        F\left(\sigma,\sigma^{|\cV}\right)&=F\left(\sigma^{\cV},\sigma^{|\cV}\right)\nonumber\\
        &=\left(\tr\sigma^\cV\right)F\left(\sigma^{|\cV},\sigma^{|\cV}\right)\nonumber\\
        &=\tr\sigma^\cV.
    \end{align}
    \item\label{fidcon} Fidelity of a density operator with a convex subcomponent thereof: if $\tau=p\tau_1+(1-p)\tau_2$ is a density operator with $0\le p\le1$ and $\tau_{1/2}$ also density operators, then by joint concavity,
    \begin{align}
        F\left(\tau,\vphantom{\tau_1}\right.&\left.\tau_1\right)\nonumber\\
        =&F\left[p\tau_1+(1-p)\tau_2,p\tau_1+(1-p)\tau_1\right]\nonumber\\
        \ge&\left[p\sqrt{F\left(\tau_1,\tau_1\right)}+(1-p)\sqrt{F\left(\tau_2,\tau_1\right)}\right]^2\nonumber\\
        \ge&\left[p\cdot1+(1-p)\cdot0\right]^2=p^2.
    \end{align}
\end{enumerate}

\section{Useful results from the literature}\label{appcont}
In our work we will make use of the following results from \cite{WY16}, namely the asymptotic continuity of the relative entropy of coherence $C_r$ and the coherence of formation $C_f$.
\begin{alem}[{\cite[Suppl.\ Material, Lemma 12]{WY16}}]\label{crcon}
    For two states $\tau_1^\rA$ and $\tau_2^\rA$ with $\nrm{\tau_1-\tau_2}_1\le\delta$,
    \begin{equation}
        \abs{C_r\left(\tau_1\right)-C_r\left(\tau_2\right)}\le\delta\log_2A+2h\left(\frac\delta2\right),
    \end{equation}
    where $h(t)=-t\log_2t-(1-t)\log_2(1-t)$ is the binary entropy function.
\end{alem}
\begin{alem}[{\cite[Suppl.\ Material, Lemma 15]{WY16}}]\label{cfcon}
    For two states $\tau_1^\rA$ and $\tau_2^\rA$ with $B\left(\tau_1,\tau_2\right)\le\delta$,
    \begin{equation}
        \abs{C_f\left(\tau_1\right)-C_f\left(\tau_2\right)}\le\delta\log_2A+(1+\delta)h\left(\frac\delta{1+\delta}\right).
    \end{equation}
\end{alem}

\section{Asymptotic typicality}\label{apptyp}
Here we give a brief self-contained review of the necessary background on asymptotic typicality for i.i.d.\ classical and quantum samples. For a detailed treatment, see \cite{wilde2013quantum}.

\subsection{Classical sources and typical sets}
We will denote classical variables with uppercase letters and specific values they take with lowercase. Suppose a classical i.i.d.\ stochastic source outputs a variable $X$ taking values in a \emph{finite} alphabet $\cX$ and distributed according to $\vect p\equiv\left[p(x)\right]_{x\in\cX}$. Consider a length-$n$ sample $\vect X\equiv X_1X_2\dots X_n$ from the source; it takes values of length-$n$ sequences $\vect x\equiv x_1x_2\dots x_n\in\cX^n$. By the source's i.i.d.\ property, the sample distribution is $\vect p_n\left(\vect X\right)=\vect p^{\otimes n}$. Recall that $H(\vect p)\equiv H(X)_{\vect p}=-\sum_{x\in\cX}p(x)\log_2p(x)$ denotes the Shannon entropy of the source.
\begin{defs}[Typical sequences and sets]\label{deftseq}
    For any real $\delta>0$, a $\delta$\n weakly-typical (or entropy-typical) sequence $\vect x$ is one that satisfies
    \begin{equation}
        2^{-n\left[H(\vect p)+\delta\right]}\le p_n\left(\vect x\right)\le2^{-n\left[H(\vect p)-\delta\right]}.
    \end{equation}
    For any $x\in\cX$, let $f_x\left(\vect x\right)$ denote the frequency of occurrences of $x$ in $\vect x$, i.e.\ $f_x\left(\vect x\right)=\abs{\left\{j:\,x_j=x\right\}}/n$. A $\delta$\n strongly-typical (or letter-typical) sequence $\vect x$ is one that satisfies
    \begin{equation}
        \sum_{x\in\cX}\abs{f_x\left(\vect x\right)-p(x)}\le\delta.
    \end{equation}
    The set $\bar\cT_n^\delta$ of all $\delta$\n weakly-typical sequences is called the $\delta$\n weakly-typical set, and the set $\cT_n^\delta$ of all $\delta$\n strongly-typical sequences the $\delta$\n strongly-typical set.
\end{defs}
\begin{alem}[Asymp.\ equipartition property [AEP{]}]\label{laep}
    For any $\delta>0$ and large enough $n$, the weakly-typical set $\bar\cT_n^\delta$ has the following properties:
    \begin{enumerate}
        \item $\sum_{\vect x\in\bar\cT_n^\delta}p_n(\vect x)\ge1-\delta$;
        \item $(1-\delta)2^{n\left[H(\vect p)-\delta\right]}\le\abs{\bar\cT_n^\delta}\le2^{n\left[H(\vect p)+\delta\right]}$.
    \end{enumerate}
    Furthermore, there exists a continuous real function $\eta(\delta)$, such that $\eta(\delta)\stackrel{\delta\to0}\longrightarrow0$ and for large enough $n$, the strongly-typical set $\cT_n^\delta$ has the following properties:
    \begin{enumerate}
        \item $2^{-n\left[H(\vect p)+\eta(\delta)\right]}\le p_n\left(\vect x\right)\le2^{-n\left[H(\vect p)-\eta(\delta)\right]}$ for all $\vect x\in\cT_n^\delta$. In other words, $\cT_n^\delta\subset\bar\cT_n^{\eta(\delta)}$;
        \item $\sum_{\vect x\in\cT_n^\delta}p_n(\vect x)\ge1-\delta$;
        \item $(1-\delta)2^{n\left[H(\vect p)-\eta(\delta)\right]}\le\abs{\cT_n^\delta}\le2^{n\left[H(\vect p)+\eta(\delta)\right]}$.
    \end{enumerate}
\end{alem}
Intuitively, we can understand AEP as follows: for a large enough sample from an i.i.d.\ source, the weakly-\ or strongly-typical set for any $\delta>0$ supports all but $\delta$ of the probability. The cardinality of this set scales exponentially, roughly as $\exp_2\left[nH(\vect p)\right]$; note that this is an exponentially small fraction of the cardinality $\abs\cX^n$ of all $n$-length sequences! Finally, the distribution of sequences within this set is nearly uniform.

\subsection{Quantum sources and typical subspaces}
As in the classical case, there is also a quantum AEP. Suppose a quantum source outputs i.i.d.\ copies of an elementary \emph{finite-dimensional} system $\rA$, each prepared in the state $\rho^{\rA}$. Considering a length-$n$ sample $\varrho_n\equiv\rho^{\otimes n}$, the quantum counterparts to the classical sequences $\vect x$ are directions in the Hilbert space of $\rA_n\equiv\rA^{\otimes n}$, and the quantum counterparts to sequence probabilities are the directional densities induced by $\varrho_n$. In particular, the eigenvectors of $\varrho_n$, and their associated eigenvalues, embody an informationally-complete description of its structure. If $\rho=\sum_{x\in[A]}r_x\psi_x$ is an eigendecomposition, then $\ket{\Psi_{\vect x}}:=\bigotimes_{k=1}^n\ket{\psi_{x_k}}$, for $\vect x\equiv\left(x_k\right)_{k=1}^n\in[A]^n$, constitute a complete basis of eigenvectors of $\varrho_n$, with associated eigenvalues $R_{\vect x}:=\prod_kr_{x_k}$. Based on these, we can define quantum counterparts of typical sequences and sets:
\begin{defs}[Typical eigenvectors and subspaces]
    For any real $\delta>0$, a $\delta$\n weakly-\ (respectively, strongly-) typical eigenvector $\ket{\Psi_{\vect x}}$ is one whose associated label sequence $\vect x$ is $\delta$\n weakly-\ (resp.\ strongly-) typical under the distribution $\vect r^{\otimes n}$. The $\delta$\n weakly-\ (resp.\ strongly-) typical subspace $\bar\cV_n^\delta$ (resp.\ $\cV_n^\delta$) is the span of the $\delta$\n weakly-\ (resp.\ strongly-) typical eigenvectors. Note that these subspaces do not depend on the choice of eigenbasis $\left\{\ket{\psi_x}\right\}_x$.
\end{defs}
\begin{alem}[Quantum AEP]\label{lqaep}
    For any $\delta>0$ and large enough $n$, the weakly-typical subspace $\bar\cV_n^\delta$ has the following properties:
    \begin{enumerate}
        \item $\tr\left(\eins_{\bar\cV_n^\delta}\varrho_n\right)\ge1-\delta$;
        \item $(1-\delta)2^{n\left[S(\rho)-\delta\right]}\le\dim\bar\cV_n^\delta\le2^{n\left[S(\rho)+\delta\right]}$.
    \end{enumerate}
    Again, there exists a continuous real-valued function $\eta(\delta)$, such that $\eta(\delta)\stackrel{\delta\to0}\longrightarrow0$ and for large enough $n$, the strongly-typical subspace $\cV_n^\delta$ has the following properties:
    \begin{enumerate}
        \item $2^{-n\left[S(\rho)+\eta(\delta)\right]}\le\eins_{\cV_n^\delta}\varrho_n\eins_{\cV_n^\delta}\le2^{-n\left[S(\rho)-\eta(\delta)\right]}$;
        \item $\tr\left(\eins_{\cV_n^\delta}\varrho_n\right)\ge1-\delta$;
        \item $(1-\delta)2^{n\left[S(\rho)-\eta(\delta)\right]}\le\dim\cV_n^\delta\le2^{n\left[S(\rho)+\eta(\delta)\right]}$.
    \end{enumerate}
\end{alem}

\section{Proof sketch for Conjecture \ref{pas}}\label{psc2}
We now present a detailed attempt at proving Conjecture \ref{pas} using the construction summarized in Section \ref{how2}. Note that the $\epsilon$'s and $\delta$'s below are not the same as those in the proof sketch for Conjecture \ref{pan}.
\begin{proof}[Proof sketch for Conjecture \ref{pas}]
    We will now try to construct a sequence of viable IO channels based on the one from Winter and Yang's maximal distillation protocol \cite{WY16}. 
    Each of Winter\n Yang's channels\footnote{Where there is scope for confusion, we use tildes to denote objects associated with Winter\n Yang's construction.} $\tilde\cE_n$ has the following special property: it can be decomposed into IO Kraus operators $\tilde K_s$ that take the form
    \begin{equation}
        \tilde K_s=\sum_m\kbra{m}{\tilde w_{s,m}}=\sqrt{M_n}\sum_m\kbra{m}{\tilde w_{s}}\eins_{\cI_m},
    \end{equation}
    where $\left\{\cI_m\right\}_m$ is a partitioning of $\cA_n\equiv \cA^n$ into $M_n=\exp_2\left(n\left[C_r(\rho)-\tilde\epsilon_n\right]\right)$ disjoint subalphabets, and $\eins_{\cI_m}$ denotes the projector onto $\mathrm{span}\left\{\ket{\vect a}:\;\vect a\in \cI_m\right\}$. Importantly, the fragments $\ket{\tilde w_{s,m}}$ for a given $m$ all lie within the subspace corresponding to $\cI_m$, independent of $s$. Consider the target effect $\tilde T_n=\sum_s\proj{\tilde w_s}$. The particular Kraus operators in Winter\n Yang's construction happen to involve $\ket{w_c}\equiv\ket{\tilde w_s}$ that form a basis of $\varrho_n$'s eigenvectors. But thanks to the above-noted $s$-independent block structure, \emph{any} convex decomposition of $\tilde T_n$ into arbitrary pure components $\ket{w_c}$ can as well be used to construct a set of IO Kraus operators (implementing the same $\tilde\cE_n$), since $\ket{w_{c,m}}:=\sqrt{M_n}\eins_{\cI_m}\ket{w_c}\in\mathrm{span}\left\{\ket{\vect a}:\;\vect a\in \cI_m\right\}$. More generally, we can even use some other $T_n$ sharing $\tilde T_n$'s essential properties and construct IO Kraus operators from arbitrary pure decompositions thereof:
    \begin{obs}\label{consmax}
    Suppose a sequence $T_n$ of effects satisfies
    \begin{enumerate}
        \item\label{MTM} $M_n\eins_{\cI_m}T_n\eins_{\cI_m}\le\eins_{\cI_m}$ for every value of $m$ indexing a partitioning $\left\{\cI_m\right\}_m$ of $\cA_n$ into $M_n=\exp_2\left(n\left[C_r(\rho)-o(1)\right]\right)$ disjoint subalphabets;
        \item\label{Tnrhon} $\tr\left(T_n\varrho_n\right)\to1$ as $n\to\infty$.
    \end{enumerate}
    Such a sequence can be used to construct a sequence of IO channels distilling asymptotically maximally from $\varrho_n$.
    \end{obs}
    \begin{proof}
    If $T_n=\sum_c\proj{w_c}$ is a decomposition into pure components, we define the Kraus operators
    \begin{equation}\label{kcfromwc}
        K_c:=\sum_m\kbra{m}{w_{c,m}}:=\sqrt{M_n}\sum_m\kbra{m}{w_{c}}\eins_{\cI_m}
    \end{equation}
    and the maps $\cF_n(\cdot):=\sum_cK_c(\cdot)K_c^\dagger$. Note that
    \begin{equation}
        \sum_cK_c^\dagger K_c=M_n\sum_m\eins_{\cI_m}T_n\eins_{\cI_m}=:\vect T_n.
    \end{equation}
    Therefore, as long as condition \ref{MTM} is satisfied, we can use the construction of \eqref{kcfromwc} on $T_n^{\setminus}:=M_n^{-1}\left(\eins^{\rA_n}-\vect T_n\right)$ to obtain IO subchannels $\cG_n$ such that $\cE_n:=\cF_n+\cG_n$ are channels. With these conditions met, the resulting $\cE_n$ are maximally-distilling channels if condition \ref{Tnrhon} is met, since the $T_n$ are (by construction) target witnesses of the subchannels $\cF_n$ with respect to the maximally-scaling targets $\Psi_{M_n}$.
    \end{proof}
    \renewcommand{\qedsymbol}{\framebox{\tiny?}}
    This concludes our formalization of point \ref{stepWY} in Section \ref{howto}. In our proof sketch for Conjecture \ref{pan}, we tried to show that the target effect\n based density operators $\tau_n$ satisfy $F\left(\tau_n,\varrho_n\right)\to1$ asymptotically. One way of trying to construct bound-attaining distillation channels is to construct a sequence $T_n$ satisfying the conditions of Observation \ref{consmax} starting from $\tau_n=\varrho_n$ itself. Since a necessity for condition \ref{Tnrhon} is $\tr T_n\ge2^{n\left[S(\rho)-o(1)\right]}$, we can aim to satisfy this weaker condition through a $\varrho_n$-based construction.
    
    As noted in point \ref{stepcf}, we can choose any optimal decomposition of $\varrho_n$ attaining $C_f\left(\varrho_n\right)$ [possibly after some pruning to uphold condition \ref{MTM}] to saturate the measurement coherence bound on average\m but we can try to do better.

    As summarized in point \ref{flc}, we will flatten relative to $c$ by decomposing $\varrho_n$ in a special way. Let $\rho=\sum_jq_j\phi_j$ be an optimal decomposition attaining $C_f(\rho)$; such a decomposition exists, by virtue of the finitude of $A$. Since $\varrho_n=\left(\sum_jq_j\phi_j\right)^{\otimes n}$, it can of course be decomposed convexly into pure states of the form $\ket{\Phi_{c\equiv\vect j}}\equiv\bigotimes_{k=1}^n\ket{\phi_{j_k}}$, where $\vect j\equiv\left(j_k\right)_k$. If we then take the strongly\n $\delta_\rJ$-typical set $\cT_n^{\delta_\rJ}$ under the distribution $\vect Q_n\equiv\vect q^{\otimes n}$ (see Definition \ref{deftseq} and Lemma \ref{laep}), the letter frequencies in any $\vect j\in\cT_n^{\delta_\rJ}$ satisfy $\abs{f_j\left(\vect j\right)-q_j}\le\delta_\rJ$. By the additivity of the relative entropy of coherence under tensor products,
    \begin{align}
        C_r\left(\Phi_{\vect j}\right)&=\sum_j\left[f_j\left(\vect j\right)n\right]C_r\left(\phi_j\right)\nonumber\\
        &=n\sum_j\left[q_j+O\left(\delta_\rJ\right)\right]C_r\left(\phi_j\right)\nonumber\\
        &=n\left[C_f(\rho)+O\left(\delta_\rJ\right)\right]
    \end{align}
    for any of these $\vect j$. The part of $\varrho_n$ composed thereof is $\varrho_n^{\delta_\rJ}\equiv\sum_{\vect j\in\cT_n^{\delta_\rJ}}Q_n\left(\vect j\right)\Phi_{\vect j}$. By the AEP, it has weight
    \begin{equation}
        \tr\varrho_n^{\delta_\rJ}=Q_n^{\delta_\rJ}\ge1-\delta_\rJ,
    \end{equation}
    where $Q_n^{\delta_\rJ}:=\sum_{\vect j\in\cT_n^{\delta_\rJ}}Q_n\left(\vect j\right)$. As before, letting $\varrho_n^{|\delta_\rJ}$ denote the normalized version thereof, the behaviour of the fidelity under convex mixtures (point \ref{fidcon} of Appendix \ref{appfid}) yields
    \begin{equation}\label{rnrnj}
        F\left(\varrho_n,\varrho_n^{|\delta_\rJ}\right)\ge\left(Q_n^{\delta_\rJ}\right)^2\ge1-2\delta_\rJ.
    \end{equation}
    This concludes point \ref{flc}. We will now subject the remaining $\ket{\Phi_{\vect j}}$ to a minor modification to ensure that our final construction can achieve $\tr T_n\approx\exp_2\left[nS(\rho)\right]$. Let $\cV^{\delta_\rS}$ be $\varrho_n$'s $\delta_\rS$\n strongly-typical subspace. Define $\ket{\Phi_{\vect j}^{\delta_\rS}}:=\eins_{\cV^{\delta_\rS}}\ket{\Phi_{\vect j}}$ and
    \begin{equation}
        \varrho_n^{\delta_\rJ,\delta_\rS}:=\sum_{\vect j\in\cT_n^{\delta_\rJ}}Q_n\left(\vect j\right)\Phi_{\vect j}^{\delta_\rS}.
    \end{equation}
    By the quantum AEP (Lemma \ref{lqaep}),
    \begin{equation}
        2^{-n\left[S(\rho)+\eta\left(\delta_\rS\right)\right]}\eins\le\varrho_n^{\delta_\rS}\le2^{-n\left[S(\rho)-\eta\left(\delta_\rS\right)\right]}\eins,
    \end{equation}
    where $\eta(\delta)\stackrel{\delta\to0}\longrightarrow0$. Noting that $\varrho_n^{\delta_\rS}=\varrho_n^{\delta_\rJ,\delta_\rS}+\varrho_n^{\setminus\delta_\rJ,\delta_\rS}$,
    \begin{equation}\label{eigbound}
        \varrho_n^{\delta_\rJ,\delta_\rS}\le\varrho_n^{\delta_\rS}\le2^{-n\left[S(\rho)-\eta\left(\delta_\rS\right)\right]}.
    \end{equation}
    Since the $\vect j$'s we have retained are all strongly-typical, the further restriction to a strongly-typical subspace has a negligible effect on the distribution of incoherent label strings $\vect a$ conditioned on $\vect j$. To see this, let $\rho=\sum_ir_i\psi_i$ be an eigendecomposition, and suppose $\ket{\phi_j}=\sum_{i}\chi_{ji}\ket{\psi_i}$. Then, $\sum_jq_j\abs{\chi_{ji}}^2=r_i$, and therefore every $\vect i$ sequence that is strongly-typical \emph{conditional} on $\vect j$ is also part of the \emph{unconditional} strongly-typical set that determines the typical subspace. As such, the distribution of $\vect a$ conditioned on $\vect j$ is unaffected except for a $\delta_\rS$ loss of total measure. We will therefore consider the distribution of $\vect a$ in $\ket{\Phi_{\vect j}}$. 
    
    Each $\ket{\Phi_{\vect j}}$ can be made arbitrarily close to a uniform superposition in the incoherent $\ket{\vect a}$ basis. Let $\ket{\phi_j}=\sum\limits_a\sqrt{\xi_{a|j}}e^{i\varphi_{j,a}}\ket a$ for $\vecg\xi_{|j}\equiv\left(\xi_{a|j}\right)_a$ a distribution and $\varphi_{j,a}\in\bbR$; note that $H\left(\vecg\xi_{|j}\right)=C_r\left(\phi_j\right)$. Then,
    \begin{equation}
        \ket{\phi_j}^{\otimes n_j}=\left(\sum\limits_a\sqrt{\xi_{a|j}}e^{i\varphi_{j,a}}\ket a\right)^{\otimes n_j}.
    \end{equation}
    We now apply (weak) asymptotic typicality on the sequences $\vect a_j\equiv\left(a_k\right)_{k=1}^{n_j}$ (where the subscript $j$ in $\vect a_j$ signifies that the latter takes values in $\cA_{n_j}$ and not $\cA_n$) under the distribution $\vecg\xi_{|j}^{\otimes n_j}$. The number of $\delta_\rA$\n weakly-typical sequences is $\le2^{n_j\left[C_r\left(\phi_j\right)+\delta_\rA\right]}$, and they collectively garner an amplitude $\ge\sqrt{1-\delta_\rA}$; this holds alike for all $j$. Note that the strongly-typical $\ket{\Phi_{\vect j}}$ described above can be obtained by applying subsystem permutations (which are incoherent unitaries) on $\bigotimes_j\ket{\phi_{j}}^{\otimes n_j\equiv\left[q_j+O\left(\delta_\rJ\right)\right]n}$. Therefore, if $\ket{\Phi_{\vect j}^{\delta_\rA}}$ denote $\ket{\Phi_{\vect j}}$ projected onto the span of $\vect a\in \cA_n$ simultaneously $\delta_\rA$-typicalized under all the $\vecg\xi_{|j}^{\otimes n_j}$,
    \begin{align}\label{rcpj}
        r_C\left(\Phi_{\vect j}^{|\delta_\rA}\right)&\le2^{n\left[C_f\left(\rho\right)+\delta_\rA+O\left(\delta_\rJ\right)\right]};\\
        \abs{\bket{\Phi_{\vect j}^{\delta_\rA}}{\Phi_{\vect j}}}&\ge\sqrt{\left(1-\delta_\rA\right)^J}=1-O\left(\delta_\rA\right),
    \end{align}
    where $\Phi_{\vect j}^{|\delta_\rA}$ denotes the normalized $\Phi_{\vect j}^{\delta_\rA}$ and $J$ is the number of components in the optimal decomposition $\rho=\sum_jq_j\phi_j$. Applying the same arguments to the subspace-typicalized $\ket{\Phi_{\vect j}^{\delta_\rS,\delta_\rA}}$,
    \begin{align}\label{rcpjd}
        r_C\left(\Phi_{\vect j}^{|\delta_\rS,\delta_\rA}\right)&\le2^{n\left[C_f\left(\rho\right)+\delta_\rA+O\left(\delta_\rJ\right)+O\left(\delta_\rS\right)\right]};\\
        \abs{\bket{\Phi_{\vect j}^{\delta_\rS,\delta_\rA}}{\Phi_{\vect j}}}&\ge1-O\left(\delta_\rS\right)-O\left(\delta_\rA\right).
    \end{align}
    Applying Lemma \ref{crcon} (the asymptotic continuity of the relative entropy of coherence \cite{WY16}), we also have
    \begin{align}\label{crpj}
        C_r&\left(\Phi_{\vect j}^{|\delta_\rS,\delta_\rA}\right)\nonumber\\
        &=n\left[C_f\left(\rho\right)+O\left(\delta_\rJ\right)+O\left(\delta_\rS\right)+O\left(\delta_\rA\right)\right].
    \end{align}
    Now define $\varrho_n^{\delta_\rJ,\delta_\rS,\delta_\rA}:=\sum_{\vect j\in\cT_n^{\delta_\rJ}}Q_n\left(\vect j\right)\Phi_{\vect j}^{\delta_\rS,\delta_\rA}$ and let $\varrho_n^{|\delta_\rJ,\delta_\rS,\delta_\rA}$ denote its normalized version. The joint concavity of the square-root fidelity implies
    \begin{align}
        F\left(\varrho_n^{|\delta_\rJ},\varrho_n^{|\delta_\rJ,\delta_\rS,\delta_\rA}\right)&\nonumber\\
        \ge&\left(\sum_{\vect j\in\cT_n^{\delta_\rJ}}\frac{Q_n\left(\vect j\right)}{Q_n^{\delta_\rJ}}\abs{\bket{\Phi_{\vect j}^{\delta_\rS,\delta_\rA}}{\Phi_{\vect j}}}\right)^2\nonumber\\
        \ge&1-O\left(\delta_\rS\right)-O\left(\delta_\rA\right),
    \end{align}
    whence
    \begin{equation}\label{rrjsa}
        F\left(\varrho_n,\varrho_n^{|\delta_\rJ,\delta_\rS,\delta_\rA}\right)\ge1-\left(2\delta_\rJ\boxplus\left[O\left(\delta_\rS\right)+O\left(\delta_\rA\right)\right]\right).
    \end{equation}
    For convenience, let $Q_n^{\delta_\rS,\delta_\rA}\left(\vect j\right)$ denote the appropriate normalized measure such that
    \begin{equation}
        \varrho_n^{|\delta_\rJ,\delta_\rS,\delta_\rA}=\sum_{\vect j\in\cT_n^{\delta_\rJ}}Q_n^{\delta_\rS,\delta_\rA}\left(\vect j\right)\Phi_{\vect j}^{|\delta_\rS,\delta_\rA}.
    \end{equation}
    This $\varrho_n^{|\delta_\rJ,\delta_\rS,\delta_\rA}$ can be decomposed into the pure components $\Phi_{\vect j}^{|\delta_\rS,\delta_\rA}$, each of whose coherence rank $\le\exp_2\left(n\left[C_f\left(\rho\right)+\delta_\rA+O\left(\delta_\rJ\right)+O\left(\delta_\rS\right)\right]\right)$, thus accomplishing point \ref{flm1} (point \ref{typs} is achieved by the $\cV^{\delta_\rS}$ projection).

    Next, as anticipated in point \ref{flm2}, we will show that the logarithm of the remaining coherence rank within each $\cI_m$ is tightly concentrated around the value $n\ell(\rho)$. To this end, for any given $\vect j\in\cT_n^{\delta_\rJ}$, define
    \begin{equation}
        \cM_n^{\vect j,\delta_\rR}:=\left\{m:\;r_C\left(\Phi_{\vect j,m}^{|\delta_\rS,\delta_\rA}\right)\le2^{n\left[\ell(\rho)+\tilde\epsilon_n+\delta_\rR\right]}\right\},
    \end{equation}
    where $\Phi_{\vect j,m}^{|\delta_\rS,\delta_\rA}:=\eins_{\cI_m}\Phi_{\vect j,m}^{|\delta_\rS,\delta_\rA}\eins_{\cI_m}/\tr\left(\eins_{\cI_m}\Phi_{\vect j,m}^{|\delta_\rS,\delta_\rA}\right)$ and $\delta_\rR>0$. If $\mu_n^{\vect j,\delta_\rR}:=\abs{\cM_n^{\vect j,\delta_\rR}}/M_n$,
    \begin{align}
        2&{}^{n\left[C_f\left(\rho\right)+\delta_\rA+O\left(\delta_\rJ\right)+O\left(\delta_\rS\right)\right]}\nonumber\\
        &\ge r_C\left(\Phi_{\vect j}^{|\delta_\rS,\delta_\rA}\right)\nonumber\\
        &=\sum_mr_C\left(\Phi_{\vect j,m}^{|\delta_\rS,\delta_\rA}\right)\nonumber\\
        &\ge\left(1-\mu_n^{\vect j,\delta_\rR}\right)M_n2^{n\left[\ell(\rho)+\tilde\epsilon_n+\delta_\rR\right]}.
    \end{align}
    Thus, if we choose $\delta_\rR\gg$ the $\delta_\rA+O\left(\delta_\rJ\right)+O\left(\delta_\rS\right)$ in the exponent on the left,
    \begin{equation}\label{ranktail}
        1-\mu_n^{\vect j,\delta_\rR}\lesssim\exp_2\left(-n\delta_\rR\right),
    \end{equation}
    where we used the fact that $M_n2^{n\left[\ell(\rho)+\tilde\epsilon_n\right]}=2^{nC_f(\rho)}$. We have thus shown that the fraction of $m$ values whose logarithmic coherence rank exceeds $n\ell(\rho)$ decays exponentially. However, this still does not rule out the possibility that such a small fraction of $m$'s actually contribute most of the norm (and of the coherence) of $\Phi_{\vect j}^{|\delta_\rS,\delta_\rA}$\m for, the actual measure over $m$ is not the uniform distribution (which we used above) but the one induced by the ranks $r_C\left(\Phi_{\vect j,m}^{|\delta_\rS,\delta_\rA}\right)$ themselves. Intuitively, one expects that a measure concentration around a small fraction of $m$'s is inconsistent with the fact that this operator approximates $\varrho_n$, which in turn is known to be maximallly-distillable under this block structure; we will now formalize this intuition.

    Here we will exploit another special property of the Winter\n Yang construction: each $\tilde\cE_n$ can be dilated unitarily within $\cH^{\rA_n}$, without appending any auxiliary system:
    \begin{equation}
        \tilde\cE_n^{\rA_n\to\rM_n}\left(\varrho_n\right)=\tr^{\rS_n}\left[\cU^{\rA_n\to\rS_n\rM_n}\left(\varrho_n\right)\right],
    \end{equation}
    where $\rS_n$ and $\rM_n$ are suitable subsystems of $\rA_n$ and $\cU_n(\cdot)\equiv U_n(\cdot)U_n^\dagger$ for some unitary $U_n$. To be sure, the unitary is preceded by a type measurement and a ``good vs.\ bad $m$'' measurement, but both of these have asymptotically-deterministic outcomes. Since the measurements are block-incoherent, the post-measurement unitary can be expanded to a fully-unitary protocol that subsumes the measurements, by controlling on the incoherent basis and defining an arbitrary unitary action (e.g.,\ the identity) on the outlying outcomes.
    
    Since $\tilde\cE_n$ are valid maximal distillation channels, there is an asymptotically-vanishing sequence $\tilde\epsilon_n^{(1)}$ such that
    \begin{align}\label{feur}
        F\left[\tilde\cE_n\left(\varrho_n\right),\Psi_{M_n}\right]&\ge1-\tilde\epsilon_n^{(1)}\nonumber\\
        \Rightarrow\bra{\Psi_{M_n}}\tilde\cE_n\left(\varrho_n\right)\ket{\Psi_{M_n}}&\ge1-\tilde\epsilon_n^{(1)}\nonumber\\
        \Rightarrow\tr\left[\left(\eins^{\rS_n}\otimes{\Psi_{M_n}}^{\rM_n}\right)\cU_n\left(\varrho_n\right)\right]&\ge1-\tilde\epsilon_n^{(1)}.
    \end{align}
    Let
    \begin{align}
        \cV':=\left(\eins^{\rS_n}\otimes{\Psi_{M_n}}^{\rM_n}\right)&\left[\mathrm{supp}\;\cU_n\left(\varrho_n\right)\right]\nonumber\\
        &=:\cW^{\rS_n}\otimes\ket{\Psi_{M_n}}^{\rM_n};
    \end{align}
    note that $\cV'$ and $\cW$ are vector spaces by construction. For any $\ket{w}\in\cW$,
    \begin{equation}\label{evenm}
        \ket v:=U_n^\dagger\left(\ket{w}\otimes\ket{\Psi_{M_n}}\right)={M_n}^{-1/2}\sum_m\ket{v_m}
    \end{equation}
    with $\ket{v_m}:=U_n^\dagger\left(\ket{w}\otimes\ket{m}\right)$, whereby $\bkt{v_m}=\bkt v$ for all $m$. Moreover, since $U_n$ induces the block structure discussed in point \ref{stepWY}, $\ket{v_m}\in\mathrm{span}\left\{\ket{\vect a}:\;\vect a\in\cI_m\right\}$. The span of all such $\ket v$'s is $\cV:=U_n^\dagger\cV'$. By unitarity, \eqref{feur} implies
    \begin{equation}
        \tr\left(\eins_{\cV}\varrho_n\right)\ge1-\tilde\epsilon_n^{(1)}.
    \end{equation}
    Combining this with \eqref{rrjsa}, we have
    \begin{align}
        \tr\left(\eins_{\cV}\varrho_n^{|\delta_\rJ,\delta_\rS,\delta_\rA}\right)&\ge1-\epsilon_n^{(2)}\Rightarrow\nonumber\\
        \sum_{\vect j\in\cT_n^{\delta_\rJ}}Q_n^{\delta_\rS,\delta_\rA}\left(\vect j\right)\bkt{\Phi_{\vect j}^{|\delta_\rS,\delta_\rA;\cV}}&\ge1-\epsilon_n^{(2)},
    \end{align}
    where $\ket{\Phi_{\vect j}^{|\delta_\rS,\delta_\rA;\cV}}:=\eins_{\cV}\ket{\Phi_{\vect j}^{|\delta_\rS,\delta_\rA}}$ (unnormalized). For any $\epsilon\in\left(\epsilon_n^{(2)},1\right]$, let
    \begin{align}
        \cT_n^{\delta_\rJ,\epsilon}&:=\nonumber\\
        &\left\{\vect j\in\cT_n^{\delta_\rJ}:\bkt{\Phi_{\vect j}^{|\delta_\rS,\delta_\rA;\cV}}\ge1-\epsilon\right\}
    \end{align}
    and $Q_n^{\delta_\rJ,\epsilon}:=\sum_{\vect j\in\cT_n^{\delta_\rJ,\epsilon}}Q_n^{\delta_\rS,\delta_\rA}\left(\vect j\right)$. Then,
    \begin{align}
        Q_n^{\delta_\rJ,\epsilon}\cdot1+\left(1-Q_n^{\delta_\rJ,\epsilon}\right)\left(1-\epsilon\right)&\ge1-\epsilon_n^{(2)}\nonumber\\
        \Rightarrow Q_n^{\delta_\rJ,\epsilon}&\ge1-\frac{\epsilon_n^{(2)}}\epsilon.
    \end{align}
    In particular, choosing $\epsilon\equiv\epsilon_n^{(3)}:=\sqrt{\epsilon_n^{(2)}}$,
    \begin{equation}\label{qnde}
        Q_n^{\delta_\rJ,\epsilon^{(3)}}\ge1-\epsilon_n^{(3)}.
    \end{equation}
    Consequently,
    \begin{widetext}
        \begin{equation}\label{trvrje}
            \tr\left[\sum_{\vect j\in\cT_n^{\delta_\rJ,\epsilon_n^{(3)}}}Q_n^{\delta_\rS,\delta_\rA}\left(\vect j\right)\Phi_{\vect j}^{|\delta_\rS,\delta_\rA;\cV}\right]\ge Q_n^{\delta_\rJ,\epsilon_n^{(3)}}\min_{\vect j\in\cT_n^{\delta_\rJ,\epsilon_n^{(3)}}}\bkt{\Phi_{\vect j}^{|\delta_\rS,\delta_\rA;\cV}}\ge\left[1-\epsilon_n^{(3)}\right]^2\ge1-2\epsilon_n^{(3)}.
        \end{equation}
    \end{widetext}
    Henceforth we shall only consider $\vect j\in\cT_n^{\delta_\rJ,\epsilon_n^{(3)}}$. Define
    \begin{equation}\label{rje2}
        \varrho_n^{|\delta_\rJ,\delta_\rS,\delta_\rA,\epsilon_n^{(3)}}:=\frac1{Q_n^{\delta_\rJ,\epsilon}}\sum_{\vect j\in\cT_n^{\delta_\rJ,\epsilon_n^{(3)}}}Q_n^{\delta_\rS,\delta_\rA}\left(\vect j\right)\Phi_{\vect j}^{|\delta_\rS,\delta_\rA}.
    \end{equation}
    Note that we have not projected the vectors onto $\cV$ in this definition: we have merely restricted the values of $\vect j$ further. Due to \eqref{qnde},
    \begin{equation}\label{rnjrnje}
        F\left(\varrho_n^{|\delta_\rJ,\delta_\rS,\delta_\rA},\varrho_n^{|\delta_\rJ,\delta_\rS,\delta_\rA,\epsilon_n^{(3)}}\right)\ge1-2\epsilon_n^{(3)},
    \end{equation}
    which implies via \eqref{rrjsa} that
    \begin{equation}\label{rnrnje}
        F\left(\varrho_n,\varrho_n^{|\delta_\rJ,\delta_\rS,\delta_\rA,\epsilon_n^{(3)}}\right)\ge1-\epsilon_n^{(4)}.
    \end{equation}
    Now define $\ket{\Phi_{\vect j}^{|\delta_\rS,\delta_\rA,\cV}}$ by normalizing $\ket{\Phi_{\vect j}^{|\delta_\rS,\delta_\rA;\cV}}$. By construction, for each remaining $\vect j$,
    \begin{align}\label{withjm}
        \abs{\bket{\Phi_{\vect j}^{|\delta_\rS,\delta_\rA}}{\Phi_{\vect j}^{|\delta_\rS,\delta_\rA,\cV}}}^2&\ge1-\epsilon_n^{(3)}\nonumber\\
        \Rightarrow\sum_m\abs{\bra{\Phi_{\vect j}^{|\delta_\rS,\delta_\rA}}\eins_{\cI_m}\ket{\Phi_{\vect j}^{|\delta_\rS,\delta_\rA,\cV}}}^2&\ge1-\epsilon_n^{(3)}.
    \end{align}
    Since each $\ket{\Phi_{\vect j}^{|\delta_\rS,\delta_\rA,\cV}}$ is a unit vector in $\cV$, its norm within each $\cI_m$ is exactly $M_n^{-1}$ [see \eqref{evenm}]:
    \begin{equation}
        \bra{\Phi_{\vect j}^{|\delta_\rS,\delta_\rA,\cV}}\eins_{\cI_m}\ket{\Phi_{\vect j}^{|\delta_\rS,\delta_\rA,\cV}}=M_n^{-1}.
    \end{equation}
    If $\alpha_m:=\sqrt{\bra{\Phi_{\vect j}^{|\delta_\rS,\delta_\rA}}\eins_{\cI_m}\ket{\Phi_{\vect j}^{|\delta_\rS,\delta_\rA}}}$ (for brevity, we suppress its dependency on $\vect j$ and other parameters), the Cauchy\n Schwartz inequality implies
    \begin{equation}
        \abs{\bra{\Phi_{\vect j}^{|\delta_\rS,\delta_\rA}}\eins_{\cI_m}\ket{\Phi_{\vect j}^{|\delta_\rS,\delta_\rA,\cV}}}^2\le\frac{\abs{\alpha_m}^2}{M_n}.
    \end{equation}
    Thus, defining $\ket\alpha:=\sum_m\alpha_m\ket m$, \eqref{withjm} begets
    \begin{equation}\label{alphapsi}
        \abs{\bket{\alpha}{\Psi_{M_n}}}^2\ge1-\epsilon_n^{(3)}.
    \end{equation}
    Again applying Lemma \ref{crcon},
    \begin{align}
        C_r\left(\proj\alpha\right)&=C_r\left(\Psi_{M_n}\right)+n\epsilon_n^{(5)}\nonumber\\
        &=\log_2M_n+n\epsilon_n^{(5)}\nonumber\\
        &=n\left[C_r(\rho)+\epsilon_n^{(6)}\right].
    \end{align}
    Define $\ket{\Phi_{\vect j,m}^{|\delta_\rS,\delta_\rA}}$ as the normalized $\eins_{\cI_m}\ket{\Phi_{\vect j}^{|\delta_\rS,\delta_\rA}}$, and through these, the vector
    \begin{equation}
        \ket{\bar\Phi_{\vect j}^{|\delta_\rS,\delta_\rA}}:=M_n^{-1/2}\sum_m\ket{\Phi_{\vect j,m}^{|\delta_\rS,\delta_\rA}}.
    \end{equation}
    Although this vector is not necessarily in $\cV$, it has even $m$ amplitudes and, moreover, satisfies $\bket{\bar\Phi_{\vect j}^{|\delta_\rS,\delta_\rA}}{\Phi_{\vect j}^{|\delta_\rS,\delta_\rA}}=\bket{\alpha}{\Psi_{M_n}}$. Thus, applying Lemma \ref{crcon} again and using \eqref{crpj},
    \begin{equation}
        C_r\left(\bar\Phi_{\vect j}^{|\delta_\rS,\delta_\rA}\right)=n\left[C_f\left(\rho\right)+\epsilon_n^{(7)}\right].
    \end{equation}
    Note that
    \begin{align}
        C_r\left(\bar\Phi_{\vect j}^{|\delta_\rS,\delta_\rA}\right)&=C_r\left(\Psi_{M_n}\right)+\frac{\sum\limits_mC_r\left(\Phi_{\vect j,m}^{|\delta_\rS,\delta_\rA}\right)}{M_n}\nonumber\\
        &=n\left[C_r(\rho)-\tilde\epsilon_n\right]+\frac{\sum\limits_mC_r\left(\Phi_{\vect j,m}^{|\delta_\rS,\delta_\rA}\right)}{M_n}.
    \end{align}
    Thus,
    \begin{equation}\label{uniavgm}
        M_n^{-1}\sum_mC_r\left(\Phi_{\vect j,m}^{|\delta_\rS,\delta_\rA}\right)=n\left[\ell(\rho)+\epsilon_n^{(8)}\right].
    \end{equation}
    Now recall \eqref{ranktail}, where we showed that for $\delta_\rR\gg\delta_\rA+O\left(\delta_\rJ\right)+O\left(\delta_\rS\right)$, the fraction of $m$ values (under the uniform measure $M_n^{-1}$) with $r_C\left(\Phi_{\vect j,m}^{|\delta_\rS,\delta_\rA}\right)>\exp_2\left(n\left[\ell(\rho)+\tilde\epsilon_n+\delta_\rR\right]\right)$ is no larger than $\exp_2\left(-n\delta_\rR\right)$. Define
    \begin{equation}
        \bar\cM_n^{\vect j,\delta_\rR}:=\left\{m:C_r\left(\Phi_{\vect j,m}^{|\delta_\rS,\delta_\rA}\right)\le n\left[\ell(\rho)+\tilde\epsilon_n+\delta_\rR\right]\right\}
    \end{equation}
    Since $C_r\le\log_2r_C$,
    \begin{equation}
        1-\frac{\abs{\bar\cM_n^{\vect j,\delta_\rR}}}{M_n}\le\exp_2\left(-n\delta_\rR\right).
    \end{equation}
    Thus, \eqref{uniavgm} implies
    \begin{equation}
        M_n^{-1}\sum\limits_{m\in\bar\cM_n^{\vect j,\delta_\rR}}C_r\left(\Phi_{\vect j,m}^{|\delta_\rS,\delta_\rA}\right)\ge n\left[\ell(\rho)+\epsilon_n^{(9)}\right].
    \end{equation}
    As in the context of \eqref{trvrje}, this is again a situation where the average value of a function is close to an extremal value. Choosing $\delta_\rR\approx\epsilon_n^{(10)}:=\sqrt{\epsilon_n^{(9)}}$ and making the other $\delta$'s small enough, we can apply similar arguments to get
    \begin{equation}\label{mprune}
        \frac{\abs{\cM_n^{\vect j,\pm\epsilon_n^{(10)}}}}{M_n}\ge1-\epsilon_n^{(10)},
    \end{equation}
    where
    \begin{align}
        \cM&{}_n^{\vect j,\pm\epsilon_n^{(10)}}:=\nonumber\\
        &\left\{m:\abs{\frac{C_r\left(\Phi_{\vect j,m}^{|\delta_\rS,\delta_\rA}\right)}n-\ell(\rho)}\le\epsilon_n^{(10)}\right\}.
    \end{align}
    Now defining
    \begin{equation}
        \ket{\bar\Phi_{\vect j}^{\delta_\rS,\delta_\rA,\epsilon_n^{(10)}}}:=M_n^{-1/2}\sum_{m\in\cM_n^{\vect j,\pm\epsilon_n^{(10)}}}\ket{\Phi_{\vect j,m}^{|\delta_\rS,\delta_\rA}},
    \end{equation}
    it follows from \eqref{mprune} that
    \begin{equation}
        \abs{\bket{\bar\Phi_{\vect j}^{\delta_\rS,\delta_\rA,\epsilon_n^{(10)}}}{\bar\Phi_{\vect j}^{|\delta_\rS,\delta_\rA}}}\ge1-\epsilon_n^{(10)}.
    \end{equation}
    Similarly defining $\ket{\Phi_{\vect j}^{\delta_\rS,\delta_\rA,\epsilon_n^{(10)}}}$ and then applying
    \begin{equation}
        \abs{\bket{\Phi_{\vect j}^{|\delta_\rS,\delta_\rA}}{\bar\Phi_{\vect j}^{|\delta_\rS,\delta_\rA}}}=\abs{\bket{\alpha}{\Psi_{M_n}}}\ge\sqrt{1-\epsilon_n^{(3)}}
    \end{equation}
    twice, we have
    \begin{align}
        \abs{\bket{\Phi_{\vect j}^{\delta_\rS,\delta_\rA,\epsilon_n^{(10)}}}{\Phi_{\vect j}^{\delta_\rS,\delta_\rA}}}&\approx\abs{\bket{\Phi_{\vect j}^{\delta_\rS,\delta_\rA,\epsilon_n^{(10)}}}{\bar\Phi_{\vect j}^{\delta_\rS,\delta_\rA}}}\nonumber\\
        &=\abs{\bket{\Phi_{\vect j}^{|\delta_\rS,\delta_\rA}}{\bar\Phi_{\vect j}^{\delta_\rS,\delta_\rA,\epsilon_n^{(10)}}}}\nonumber\\
        &\approx\abs{\bket{\bar\Phi_{\vect j}^{|\delta_\rS,\delta_\rA}}{\bar\Phi_{\vect j}^{\delta_\rS,\delta_\rA,\epsilon_n^{(10)}}}}.
    \end{align}
    Thus,
    \begin{equation}\label{fincomp}
        \abs{\bket{\Phi_{\vect j}^{\delta_\rS,\delta_\rA,\epsilon_n^{(8)}}}{\Phi_{\vect j}^{|\delta_\rS,\delta_\rA}}}\ge1-\epsilon_n^{(11)}.
    \end{equation}
    We now define
    \begin{equation}
        \tau_n:=\mathrm{norm}\sum_{\vect j\in\cT_n^{\delta_\rJ,\epsilon_n^{(3)}}}Q_n^{\delta_\rS,\delta_\rA}\left(\vect j\right)\Phi_{\vect j}^{\delta_\rS,\delta_\rA,\epsilon_n^{(10)}},
    \end{equation}
    where ``norm'' denotes normalization. The joint concavity of the square-root fidelity, together with \eqref{rnrnje} and \eqref{fincomp}, yields
    \begin{equation}
        F\left(\tau_n,\varrho_n\right)\ge1-\epsilon_n^{(12)}.
    \end{equation}
    Furthermore, every modification since \eqref{eigbound} has been among the following types:
    \begin{enumerate}
        \item Unnormalized projection of the entire operator;
        \item Unnormalized projection of a pure component in a convex decomposition;
        \item Renormalization of the entire operator by a factor close to 1.
    \end{enumerate}
    Therefore,
    \begin{equation}
        \tau_n\le\left(1+\epsilon_n^{(13)}\right)2^{-n\left[S(\rho)-\eta\left(\delta_\rS\right)\right]}.
    \end{equation}
    Defining $\bar S_n:=\left(\nrm{\tau_n}_\infty\right)^{-1}$,
    \begin{equation}
        T_n:=\bar S_n\tau_n
    \end{equation}
    satisfies $\tr\,T_n\approx\left(1-\epsilon_n^{(13)}\right)2^{n\left[S(\rho)-\eta\left(\delta_\rS\right)\right]}$ and $T_n\le\eins$. But this $T_n$ may not satisfy condition \ref{MTM}, since in general
    \begin{equation}
        M_n\eins_{\cI_m}T_n\eins_{\cI_m}\nleq\eins_{\cI_m}.
    \end{equation}
    This adds to our growing list of obstacles against completing our proof. A possible remedy for this particular challenge might be to appeal to the symmetric and random nature of the choice of partitioning (see \cite[Supplemental Material, Lemma 16]{WY16}; \cite[Prop.~2.4]{DW05}) to prune out a small fraction of offending $m$ values. Alternately, in the step \eqref{rje2} where (as we then explicated) we only restricted $\vect j$ to $\cT_n^{\delta_\rJ,\epsilon_n^{(3)}}$, we could additionally project each pure component $\ket{\Phi_{\vect j}^{|\delta_\rS,\delta_\rA}}$ onto $\cV$, which by its special structure automatically satisfies
    \begin{equation}
        M_n\eins_{\cI_m}\eins_\cV\eins_{\cI_m}\le\eins_{\cI_m}.
    \end{equation}
    If we therefore ensure that the overall $T_n$ (at this projection step) satisfies $T_n\lesssim\eins_\cV$, the subsequent modifications would still maintain $M_n\eins_{\cI_m}T_n\eins_{\cI_m}\lesssim\eins_{\cI_m}$, even though the eventual $T_n$ may not be supported on $\cV$. This would enable us to uphold condition \ref{MTM} by applying a normalization factor close to 1 (as in the other steps).

    However, projecting $\ket{\Phi_{\vect j}^{|\delta_\rS,\delta_\rA}}$ onto $\cV$ brings another problem: while we were able to approximate each $\ket{\Phi_{\vect j}^{|\delta_\rS}}$ with a near-uniform superposition $\ket{\Phi_{\vect j}^{|\delta_\rS,\delta_\rA}}$ by appealing to the former's tensor-product structure, no such structure is assured for $\ket{\Phi_{\vect j}^{|\delta_\rS,\delta_\rA,\cV}}$. Though we are able to show that this projected vector is close to the unprojected one, it is not close enough to assure the retention of the near-uniformity of the superposition. Here again, we speculate that the symmetry and randomness in the choice of partitioning might help argue for the existence of near-uniform approximations to the $\cV$-projected vectors.

    Finally, if everything worked out, the residual $T_n^\setminus$ would have trace exponentially smaller than that of $T_n$. Thus, the additional coherent measurement cost from it, as well as its contribution to the average (measured in logarithmic rank), would be negligible.
\end{proof}

\section{Potential SDP-based approach}\label{appsdp}
Here we discuss a possible semidefinite programming (SDP)\n based approach to the problem of estimating the coherent measurement cost of coherence distillation. We start with some definitions and results from Ringbauer \textit{et al.}\ \cite{RBC+18}. Define the set
\begin{align}
    \cC_k:&=\left\{\text{density operators }\sigma:\;r_C(\sigma)\le k\right\}\nonumber\\
    &=\mathrm{cvx}\left\{\psi:\;\bkt\psi=1,\,r_C(\psi)\le k\right\}.
\end{align}
The scope of the density operators is to be understood as determined by the underlying system Hilbert space, which we leave implicit.
\begin{defs}[Robustness of multilevel coherence]
    For a density operator $\rho$, its \emph{generalized robustness of $(k+1)$-coherence} is defined as
    \begin{equation}
        R_{\cC_k}(\rho):=\inf_{\sigma}\left\{s\ge0:\;\frac{\rho+s\sigma}{1+s}\in\cC_k\right\},
    \end{equation}
    where the infimum is over all density operators $\sigma$.
\end{defs}
For a system $\rA$, denote
\begin{equation}
\cP_{k}:=\left\{\cI\subseteq\cA:\;\left|\cI\right|=k\right\}.
\end{equation}
\begin{alem}[{\cite[Supplemental Material G]{RBC+18}}]
    The robustness of multilevel coherence can be cast as the following SDP:
    \begin{equation}
        \begin{array}{rl}
            R_{{\cC}_k}(\rho)=\min  &\tr\left(\sum_{\cI\in\mathcal{P}_k}\sigma_{\cI}\right)-1\\
            \textup{s.t.}           &\sum_{\cI\in\mathcal{P}_k}\sigma_{\cI}\geq\rho\\
                                    &\left. \begin{array}{l}
                                                \sigma_{\cI}\geq0\\
                                                \eins_{\cI}\sigma_{\cI}\eins_\cI=\sigma_\cI
                                            \end{array}\right\}\forall\cI\in\mathcal{P}_k.
        \end{array}
    \end{equation}
\end{alem}
We now note that this definition can be extended to arbitrary positive-semidefinite operators (beyond density operators), by defining
\begin{equation}
    \begin{array}{rl}
        R_{{\cC}_k}(T):=\min    &\tr\left(\sum_{\cI\in\mathcal{P}_k}S_{\cI}-T\right)\\
        \textup{s.t.}           &\sum_{\cI\in\mathcal{P}_k}S_{\cI}\ge T\\
                                &\left. \begin{array}{l}
                                            S_{\cI}\ge0\\
                                            \eins_{\cI}S_{\cI}\eins_\cI=S_\cI
                                        \end{array}\right\}\forall\cI\in\mathcal{P}_k,
    \end{array}
\end{equation}
also an SDP. The resulting ``robustness'' measure is itself not normalized; but it has the convenient property of being (non)zero iff the corresponding value evaluated on the normalized density operator $\tau\equiv T/\tr T$ is (non)zero. Now, note that the coherence rank can be defined through the robustness measures as
\begin{equation}
    r_C(T):=\min\left\{k\in\bbZ_+:\;R_{{\cC}_k}(T)=0\right\}.
\end{equation}
This definition has the desirable property that $r_C(T)=r_C(xT)\;\forall x>0$. Thereby, it lends itself well to use in quantifying the coherent measurement cost of a channel $\cE$ in terms of the requisite measurement coherence rank $r_C\left(T_\cE\right)$. This quantity can be computed directly through (iterations over instances of) the latter SDP, without the tricky business of having to normalize $T_\cE$. Thus, using the properties of the target effect construction, we have the following:
\begin{obs}
    Given an input state $\rho^\rA$ and some $M,k\in[A]+1$ and $\epsilon\in\bbR_+$, distilling $\Psi_M$ with fidelity $\ge1-\epsilon$ by MIO requires measurements of coherence rank $>k$ if $R_k^\epsilon\left(\rho,A,M\right)>0$, where
    \begin{equation}
    \begin{array}{rl}
        R_k^\epsilon\left(\rho,A,M\right)&:=\\
        \min    &\tr\left(\sum_{\cI\in\mathcal{P}_k}S_{\cI}-T\right)\\
        \textup{s.t.}           &0\le T\le\eins\\
        &\tr\left(T\proj a\right)=M^{-1}\quad\forall a\in[A]\\
        &\tr(\rho T)\ge1-\epsilon\\
                                &\sum_{\cI\in\mathcal{P}_k}S_{\cI}\ge T\\
        &\left. \begin{array}{l}
                                            S_{\cI}\ge0\\
                                            \eins_{\cI}S_{\cI}\eins_\cI=S_\cI
                                        \end{array}\right\}\quad\forall\cI\in\mathcal{P}_k.
    \end{array}
    \end{equation}
    Therefore, the requisite coherence rank for the task is no smaller than
    \begin{align}
        r_C^\epsilon\left(\rho,A,M\right)&:=\nonumber\\
        &\min\left\{k\in\bbZ_+:R_k^\epsilon\left(\rho,A,M\right)=0\right\}.
    \end{align}
\end{obs}
Given any specific tuple of $(A,M,\rho,\epsilon,k)$, the computation of $R_k^\epsilon\left(\rho,A,M\right)$ is an SDP. But it may not be efficient in terms of these parameters' sizes: in particular, $\abs{\cP_k}=\left(\begin{array}{c}
     A\\
     k 
\end{array}\right)$ grows roughly exponentially with $A$. Computing $r_C^\epsilon\left(\rho,A,M\right)$ would further entail iterating over many $k$ values, potentially up to $k\approx A/M$.

Nevertheless, we hope that the above SDP formulation of $R_k^\epsilon\left(\rho,A,M\right)$ affords insights into the behaviour of $r_C^\epsilon\left(\rho,A,M\right)$, e.g.\ for the case of maximal asymptotic distillation. We leave this for future work as a potential approach towards settling our conjectures.

\section{Lami's SIO monotones under rank-constrained IO}\label{secmuk}
In order to tease out SIO-distillable coherence, Lami \cite{lami_2019} introduced a family of functions $\mu_k$ that are monotones under SIO but not under general IO. For any $\rho^\rA$, let
\begin{equation}\label{rrho}
    R^\rho:=\left[\Delta(\rho)\right]^{-1/2}\rho\left[\Delta(\rho)\right]^{-1/2},
\end{equation}
with the inverse defined on the support of $\Delta(\rho)$. Then,
\begin{equation}
    \mu_k(\rho):=\max_{\cI\subseteq\cA:\,\abs{\cI}=k}\log_2\nrm{\eins_\cI R^\rho\eins_\cI}_\infty.
\end{equation}
One way of approaching the problem of quantifying the requisite coherent measurement cost of coherence distillation would be to study the behaviour of these SIO monotones under IO with constrained coherent measurement action.
\begin{defs}[$L$-ary IO]\label{lary}
    An IO $\cE$ is \emph{$L$-ary} if it can be decomposed in terms of IO Kraus operators each of whose rows has $\le L$ nonzero entries.
\end{defs}
For convenience, also define
\begin{widetext}
\begin{equation}
\cP_{k,L}:=\left\{\sI\equiv\left\{\cI_m\subseteq\cA\right\}_{m\in[k]}:\;\left|\cI_m\right|\le L\;\&\;\left|\cI_m\cap \cI_{m'}\right|\propto\delta_{mm'}\;\forall m,m'\in[k]\right\}.
\end{equation}
\end{widetext}
At first glance, it may seem that $\cP_{k,1}\cong\cP_{k}$ in the notation of Appendix \ref{appsdp} (where we use ``$\cong$'' to signify the identification of $\sI\equiv\left\{\{j_m\}\right\}_{m\in[k]}$ with $\cI\equiv\{j_m\}_{m\in[k]}$). But actually, since $\abs{\cI_m}$ is allowed to be $<L$ in the definition of $\cP_{k,L}$, we have, instead, $\cP_{k,1}\cong\bigcup\limits_{l=0}^k\cP_{l}$.
\begin{remm}
	A generic $L$-ary IO Kraus operator with an $M$-dimensional output has the form
	\begin{equation}\label{genlio}
	K=\sum_m\kbra{m}{v_m},
	\end{equation}
	where $\ket{v_m}\in\mathrm{span}\left\{\ket i\right\}_{i\in \cI_m}$ for some $\sI\in\cP_{M,L}$.
\end{remm}
We will now perform an $L$-ary construction analogous to the definition \eqref{rrho}. For any given $\sI\in\cP_{M,L}$, define
\begin{equation}
\Delta_\sI(\cdot):=\sum_m\eins_{\cI_m}(\cdot)\eins_{\cI_m},
\end{equation}
and therewith, for any $\rho$,
\begin{equation}
R^\rho_\sI:=\left[\Delta_\sI(\rho)\right]^{-1/2}\rho\left[\Delta_\sI(\rho)\right]^{-1/2}.
\end{equation}
It is worth noting that
\begin{equation}\label{rid}
\eins_{\cI_m}R^\rho_\sI\eins_{\cI_m}=\eins_{\rho_{\cI_m}\equiv\eins_{\cI_m}\rho\eins_{\cI_m}}.
\end{equation}
Now define
\begin{equation}\label{maxlio}
\mu_{M,L}(\rho):=\max_{\sI\in\cP_{M,L}}\log_2\left\|R^\rho_\sI\right\|_\infty.
\end{equation}
\begin{obs}
	Given an input $\rho$,
	\begin{enumerate}
		\item For an individual $L$-ary IO Kraus operator $K$ that partitions its input according to some $\sI\in\cP_{M,L}$, let $\sigma:=K\rho K^\dagger$. Then,  $\left\|R^\sigma\right\|_\infty\le\left\|R^\rho_\sI\right\|_\infty$. By convexity, $\mu_M\left[\cE(\rho)\right]\le\mu_{M,L}(\rho)$ for any $L$-ary IO $\cE$.
		\item For each $\sI\in\cP_{M,L}$, there exists an $L$-ary IO Kraus operator $K$ partitioning by $\sI$ such that, for $\sigma:=K\rho K^\dagger$, $\left\|R^\sigma\right\|_\infty=\left\|R^\rho_\sI\right\|_\infty$.
	\end{enumerate}
	Consequently,
	\begin{equation}
	\max_{L\textnormal{-ary IO }\cE}\mu_M\left[\cE(\rho)\right]=\mu_{M,L}(\rho).
	\end{equation}
\end{obs}
\begin{proof}
	Let
	\begin{equation}
	K=\sum_m\ket m\bra{v_m},
	\end{equation}
	where $\ket{v_m}\in\mathrm{span}\left\{\ket a\right\}_{a\in \cI_m}$ for some $\sI\in\cP_{M,L}$. For every $m$, define $\ket{u_m}:=\sqrt{\rho_{\cI_m}}\ket{v_m}\in\mathrm{supp}\left(\rho_{\cI_m}\right)$. Then, for $\sigma:=K\rho K^\dagger$,
	\begin{align}
		\bra{m_1}\sigma\ket{m_2} & =\bra{v_{m_1}}\rho\ket{v_{m_2}}\nonumber                                         \\
		                         & =\bra{v_{m_1}}\sqrt{\rho_{\cI_{m_1}}}R^\rho_\sI\sqrt{\rho_{\cI_{m_2}}}\ket{v_{m_2}}\nonumber \\
		                         & =\bra{u_{m_1}}R^\rho_\sI\ket{u_{m_2}}                                            
	\end{align}
	for all $m_1$, $m_2$. Therefore,
	\begin{align}\label{sigrho}
		R^\sigma & =\sum_{m_1,m_2}\frac{\proj{m_1}\sigma\proj{m_2}}{\sqrt{\bra{m_1}\sigma\ket{m_1}\bra{m_2}\sigma\ket{m_2}}}\nonumber                                                     \\
		         & =\sum_{m_1,m_2}\frac{\ket{m_1}\bra{u_{m_1}}R^\rho_\sI\ket{u_{m_2}}\bra{m_2}}{\sqrt{\bra{u_{m_1}}R^\rho_\sI\ket{u_{m_1}}\bra{u_{m_2}}R^\rho_\sI\ket{u_{m_2}}}}\nonumber \\
		         & =\sum_{m_1,m_2}\frac{\ket{m_1}\bra{u_{m_1}}R^\rho_\sI\ket{u_{m_2}}\bra{m_2}}{\sqrt{\bra{u_{m_1}}\left.u_{m_1}\right\rangle\bra{u_{m_2}}\left.u_{m_2}\right\rangle}},   
	\end{align}
	the last line following from \eqref{rid}. Now, for any $\ket\phi\in\mathrm{supp}\left(R^\sigma\right)$, define
	\begin{equation}
	\ket\psi:=\sum_m\frac{\ket{u_m}\bra m\phi\rangle}{\sqrt{\bra{u_m}\left.u_m\right\rangle}}.
	\end{equation}
	$\bra\psi\psi\rangle=\bra\phi\phi\rangle$ holds by construction, while $\bra\psi R^\rho_\sI\ket\psi=\bra\phi R^\sigma\ket\phi$ follows from \eqref{sigrho}. Thus, $\left\|R^\sigma\right\|_\infty\le\left\|R^\rho_\sI\right\|_\infty$.
	
	Conversely, given any $\sI\in\cP_{M,L}$ and $\ket\psi:=\sum\limits_m\ket{u_m}$
	with $\ket{u_m}\in\mathrm{supp}\left(\rho_{\cI_m}\right)$, define $\ket{v_m}:=\rho_{\cI_m}^{-1/2}\ket{u_m}$ for each $m$. Then, $K:=\sum_m\ket m\bra{v_m}$ is an $L$-ary IO Kraus operator such that, for $\sigma:=K\rho K^\dagger$, \eqref{sigrho} holds. Defining
	\begin{equation}
	\ket\phi:=\sum_m\sqrt{\bra{u_m}\left.u_m\right\rangle}\ket m,
	\end{equation}
	we obtain $\bra\phi\phi\rangle=\bra\psi\psi\rangle$ and $\bra\phi R^\sigma\ket\phi=\bra\psi R^\rho_\sI\ket\psi$. Thus, $\left\|R^\sigma\right\|_\infty\le\left\|R^\rho_\sI\right\|_\infty$ can be saturated by choosing $\ket\psi$ appropriately.
	
	Finally, if $\sI$ is a partitioning that attains the maximum on the right side of \eqref{maxlio}, and $K$ a Kraus operator constructed as above to saturate $\left\|R^\sigma\right\|_\infty\le\left\|R^\rho_\sI\right\|_\infty$, we can complete this to an $L$-ary IO channel $\cE$ with other Kraus operators whose output spaces don't overlap with that of $K$. This $\cE$ then attains the maximum.
\end{proof}
This bound is tight when the maximization is over \emph{all} $L$-ary IO $\cE$. But what is relevant in the context of distillation is the $\mu_M$ for $M$ equalling the dimension of the channel's output. Clearly, this is a much more constrained quantity and would generically not attain the bound we have found. Specifically, the argument used in the last paragraph of the proof would fail when the output space is constrained to be $M$-dimensional. If we are to use this bound to pin down the distillation rate under $L$-ary IO, we need to understand the behaviour of $\mu_{M,L}$ under tensor product copies (i.e.,\ an analog of \cite[Proposition 16]{lami_2019}). By virtue of asymptotic typicality, we might get by without having to incorporate the constraint mentioned above.

\section{Decoupling schemes}\label{secdec}
Here we will study a type of decoupling task, whereof coherence distillation is a special case. In particular, we will look at \emph{exact deterministic decoupling of a pure output}. The input in this task is some state $\rho$, and the goal is to apply a channel $\cE$ such that $\cE(\rho)=\proj\alpha$ for some specified pure state $\ket\alpha:=\sum_m\alpha_m\ket m$. On the face of it, it may appear superfluous to view this as decoupling\m for example, $\cE$ could just ignore the input and prepare the required output. But the utility of the decoupling perspective will become apparent presently.

Recall that exactly producing a pure output $\proj\alpha^\rM$ from $\rho^\rA$ entails that the channel $\cE$ map the entire space $\cL\left(\cV\equiv\mathrm{supp}\rho\right)$ to (scalar multiples of) $\proj\alpha$. Let $S:=\dim\cV$ and $\left\{\ket{v_s}\right\}_{s\in\cS}$ be a $\cV$ basis. Then, $\cE$ must have a dilation $V$ with the action
\begin{equation}
V^{\rA\to\rC\rM}\ket{v_s}^\rA=\ket{u_s}^\rC\ket{\alpha}^\rM
\end{equation}
with orthonormal $\left\{\ket{u_s}\right\}_{s\in\cS}$. If $V$ can be completed to a unitary within $\cH^\rA$, we then have some well-defined $\ket{v_{s,m}}^\rA:=V^\dagger\left(\ket{u_s}^\rC\ket m^\rM\right)\in\cH^\rA$. Again by unitarity, the collection $\sV_m:=\left\{\ket{v_{s,m}}\right\}_{s\in\cS}$ for each $m\in\cM$ must be orthonormal; in fact, the entire collection $\sV_\cM:=\bigcup\limits_{m\in\cM}\sV_m$ must be orthonormal.

The conditions above are already necessary for the existence of a unitary $V^{\rA\to\rC\rM}$ that accomplishes the task; the requirement that it be an IO dilation will bring in further constraints. But we will ignore these for the time being and try to learn what we can about unitary decoupling in general.
\begin{defs}[$\alpha$-decoupling scheme]\label{ddec2}
    For a dimension-$S$ subspace $\cV$ of a Hilbert space $\cH^\rA\cong\bbC^A$ and a unit vector $\alpha\equiv\left(\alpha_m\right)_{m\in\cM\equiv[M]}\in\bbC^M$, a collection $\sV_\cM\equiv\left\{\ket{v_{s,m}}^\rA\in\cH^\rA\right\}_{m\in\cM,s\in\cS}=\bigcup_m\left(\sV_m\equiv\left\{\ket{v_{s,m}}^\rA\right\}_s\right)$ of orthonormal vectors is an \emph{$\alpha$-decoupling scheme for $\cV$ in $\cH^\rA$} if $\sV\equiv\left\{\ket{v_s}=\sum_m\alpha_m\ket{v_{s,m}}\right\}_s$ is a basis for $\cV$.
\end{defs}
\begin{obs}\label{odec}
    For an $\alpha$-decoupling scheme $\sV_\cM$ in $\cH^\rA$ for $\cV$, let $\ket\alpha:=\sum_m\alpha_m\ket m$ and $\cV_\cM:=\mathrm{span}\sV_\cM$. For some system $\rC$ of dimensionality $C\ge A/M$, define a unitary $V^{\rA\to\rC\rM}$ whose action on $\cV_\cM$ is given by
    \begin{equation}\label{Vdec}
        V^{\rA\to\rC\rM}_{\cV_\cM}=\sum_{s\in\cS,m\in\cM}\ket s^\rC\ket m^\rM\bra{v_{s,m}}^\rA.
    \end{equation}
    Then, for any $\ket v\equiv\sum_s\xi_s\ket{v_s}\in\cV$,
    \begin{equation}
        V\ket v=\left(\sum_s\xi_s\ket s\right)^\rC\otimes\ket\alpha^\rM.
    \end{equation}
\end{obs}
We omit a proof for the above observation, since it is straightforward. Nevertheless, it is instructive in explicating the operational motivation for our definition of the term ``$\alpha$-decoupling scheme''. It also suggests that the decoupling perspective on the problem has potential nontrivial utility when the space $\cH^\rA$ within which the channel can be dilated is explicitly considered.

So far we restricted ourselves to cases where it can be done through unitary action within $\cH^\rA$. Such a unitary gave us a decoupling scheme consisting of orthogonal vectors $\ket{v_{s,m}}^\rA$. But in general, we need to consider decoupling isometries that may not be completable to a unitary within $\cH^\rA$. The associated analogs to decoupling schemes could then be oblique and overcomplete. To understand these cases, let us once again examine the basic condition on an $\alpha$-decoupling isometry $V$ on a space $\cV$:
\begin{equation}
V^{\rA\to\rC\rM}\ket{v_s}^\rA=\ket{u_s}^\rC\ket{\alpha}^\rM,
\end{equation}
where $\ket{v_s}$ is a basis on $\cV$. Now let us expand $\rA$ to some $\bar\rA$, of dimensionality $\bar A:=CM$, such that $V^{\rA\to\rC\rM}$ can be completed to the unitary $V^{\bar\rA\to\rC\rM}$. By assumption, $V^\dagger\ket{u_s}^\rC\ket{\alpha}^\rM=\ket{v_s}^\rA=:\ket{\bar v_s}^{\bar\rA}$. Indeed, if we complete $\left\{\ket{u_s}\right\}_{s\in\cS}$ to an $\cH^\rC$ basis $\left\{\ket{u_c}\right\}_{c\in\cC}$ and thereby define $\ket{\bar v_c}^{\bar\rA}:=V^\dagger\ket{u_c}^\rC\ket{\alpha}^\rM$ for all $c\in\cC$, then by construction, $\bar\cV:=\mathrm{span}\left[\bar\sV\equiv\left\{\ket{\bar v_c}^{\bar\rA}\right\}_{c\in\cC}\right]$ is an $\alpha$-decouplable subspace of $\cH^{\bar\rA}$.

Recalling our ultimate aim to estimate the coherent measurement cost of an IO whose dilation is some $V$ as above, let us inspect a set of Kraus operators that would implement the combined action of the embedding of $\rA$ in $\bar\rA$, then the unitary $V$, and finally a $\rC$\n partial trace; since we have kept $\rC$ and $\ket{u_c}^\rC$ generic, we lose no generality in decomposing the partial trace in the canonical basis $\bra c^\rC$, as in Remark \ref{iodil}. The Kraus operators, then, are
\begin{equation}\label{genio2}
    K_c^{\rA\to\rM}=\bra{c}^\rC V^{\bar\rA\to\rC\rM}\eins^\rA=:\sum_{m\in\cM}\ket m^\rM\bra{w_{c,m}}^\rA.
\end{equation}
Here $\ket{w_{c,m}}^\rA=\eins^\rA\ket{\bar w_{c,m}}^{\bar\rA}$, the latter defined through
\begin{align}
\sum_{c,m}\ket{c}^\rC\ket{m}^\rM\bra{\bar w_{c,m}}^{\bar\rA}&=\sum_{c,m}\ket{u_c}^\rC\ket{m}^\rM\bra{\bar v_{c,m}}^{\bar\rA}\nonumber\\
&(=V).
\end{align}
Notice that $\bar\sW_\cM\equiv\left\{\ket{\bar w_{c,m}}^{\bar\rA}\right\}_{c,m}$ is, just like $\bar\sV_\cM$, an $\alpha$-decoupling scheme for $\bar\cV$ in $\cH^{\bar\rA}$; the schemes are associated, respectively, with the $\bar\cV$ bases $\bar\sW\equiv\left\{\ket{\bar w_{c}}^{\bar\rA}:=\sum_m\alpha_m\ket{\bar w_{c,m}}^{\bar\rA}\right\}_{c}$ and $\bar\sV$. We summarize the above observations in the following\dots well, observation.
\begin{obs}\label{odil}
    If $\cE^{\rA\to\rM}$ is a channel (IO or otherwise) that deterministically maps a subspace $\cV\in\cH^\rA$ to $\proj{\alpha}^\rM$, \emph{any} Kraus operator decomposition of $\cE$ must involve operators of the form
    \begin{equation}
        K_c=\sum_{m\in\cM}\ket m^\rM\bra{w_{c,m}}^\rA
    \end{equation}
    with $\ket{w_{c,m}}^\rA=\eins^\rA\ket{\bar w_{c,m}}^{\bar\rA}$ projections of an $\alpha$-decoupling scheme $\bar\sW_\cM\equiv\left\{\ket{\bar w_{c,m}}^{\bar\rA}\right\}_{c,m}$ for some $\bar\cV\supset\cV$ in some $\cH^{\bar\rA}\supset\cH^\rA$.
\end{obs}
Thus, the problem of finding the least-coherent IO channel(s) executing a given decoupling task involves optimizing over all possible decoupling schemes within arbitrarily large $\bar\rA$ and $\bar\cV$.

Since $\bar\sW$ (introduced before Observation \ref{odil}) is a $\bar\cV$-basis and $\cV\subset\bar\cV$,
    \begin{align}\label{wwbar}
        \eins_{\bar\cV}=\sum_c\proj{\bar w_c}&\ge\eins_\cV\nonumber\\
        \Rightarrow T_\cE=\eins^{\rA}\left(\sum_c\proj{\bar w_c}\right)\eins^{\rA}&\ge\eins^{\rA}\eins_\cV\eins^{\rA}=\eins_\cV,
    \end{align}
where we used the fact that $\cV\subset\cH^\rA$. Meanwhile, since $\eins_{\bar\cV}\le\eins^{\bar\rA}$, $T_\cE=\eins^{\rA}\eins_{\bar\cV}\eins^{\rA}\le\eins^{\rA}\eins^{\bar\rA}\eins^{\rA}=\eins^{\rA}$. Thus, using the concept of decoupling schemes we have arrived at the result of Lemma \ref{TWisV} via a different route. More research into these structures may allow us to harness their properties to attack channel-related problems in the dilation picture.


\begin{thebibliography}{10}

\bibitem{SAP17}
Alexander Streltsov, Gerardo Adesso, and Martin~B Plenio.
\newblock ``Colloquium: Quantum coherence as a resource''.
\newblock \href{https://dx.doi.org/10.1103/RevModPhys.89.041003}{Reviews of
  Modern Physics {\bf 89}, 041003}~(2017).

\bibitem{aaberg2004subspace}
Johan {\AA}berg.
\newblock ``Subspace preservation, subspace locality, and gluing of completely
  positive maps''.
\newblock \href{https://dx.doi.org/10.1016/j.aop.2004.01.006}{Annals of Physics
  {\bf 313}, 326--367}~(2004).

\bibitem{aberg2006quantifying}
Johan Aberg.
\newblock ``Quantifying superposition''~(2006).

\bibitem{baumgratz2014quantifying}
Tillmann Baumgratz, Marcus Cramer, and Martin~B Plenio.
\newblock ``Quantifying coherence''.
\newblock \href{https://dx.doi.org/10.1103/PhysRevLett.113.140401}{Physical
  review letters {\bf 113}, 140401}~(2014).

\bibitem{WY16}
Andreas Winter and Dong Yang.
\newblock ``Operational resource theory of coherence''.
\newblock \href{https://dx.doi.org/10.1103/PhysRevLett.116.120404}{Physical
  review letters {\bf 116}, 120404}~(2016).

\bibitem{YMG+16}
Benjamin Yadin, Jiajun Ma, Davide Girolami, Mile Gu, and Vlatko Vedral.
\newblock ``Quantum processes which do not use coherence''.
\newblock \href{https://dx.doi.org/10.1103/PhysRevX.6.041028}{Physical Review X
  {\bf 6}, 041028}~(2016).

\bibitem{lami_2019-1}
Ludovico Lami, Bartosz Regula, and Gerardo Adesso.
\newblock ``Generic {{Bound Coherence}} under {{Strictly Incoherent
  Operations}}''.
\newblock \href{https://dx.doi.org/10.1103/PhysRevLett.122.150402}{Phys. Rev.
  Lett. {\bf 122}, 150402}~(2019).

\bibitem{lami_2019}
Ludovico Lami.
\newblock ``Completing the grand tour of asymptotic quantum coherence
  manipulation''.
\newblock \href{https://dx.doi.org/10.1109/TIT.2019.2904308}{IEEE Transactions
  on Information Theory {\bf 66}, 2165--2183}~(2019).

\bibitem{hayashi2024generalized}
Masahito Hayashi and Hayata Yamasaki.
\newblock ``Generalized quantum stein's lemma and second law of quantum
  resource theories''~(2024).

\bibitem{lami2024solution}
Ludovico Lami.
\newblock ``A solution of the generalised quantum stein's lemma''.
\newblock \href{https://dx.doi.org/10.1109/TIT.2025.3543610}{IEEE Transactions
  on Information Theory}~(2025).

\bibitem{MS16}
Iman Marvian and Robert~W Spekkens.
\newblock ``How to quantify coherence: Distinguishing speakable and unspeakable
  notions''.
\newblock \href{https://dx.doi.org/10.1103/PhysRevA.94.052324}{Physical Review
  A {\bf 94}, 052324}~(2016).

\bibitem{BKB21}
Felix Bischof, Hermann Kampermann, and Dagmar Bru{\ss}.
\newblock ``Quantifying coherence with respect to general quantum
  measurements''.
\newblock \href{https://dx.doi.org/10.1103/PhysRevA.103.032429}{Physical Review
  A {\bf 103}, 032429}~(2021).

\bibitem{TKEP17}
Thomas Theurer, Nathan Killoran, Dario Egloff, and Martin~B Plenio.
\newblock ``Resource theory of superposition''.
\newblock \href{https://dx.doi.org/10.1103/PhysRevLett.119.230401}{Physical
  review letters {\bf 119}, 230401}~(2017).

\bibitem{aaberg2014catalytic}
Johan {\AA}berg.
\newblock ``Catalytic coherence''.
\newblock \href{https://dx.doi.org/10.1103/PhysRevLett.113.150402}{Physical
  review letters {\bf 113}, 150402}~(2014).

\bibitem{kim2022relation}
Ho-Joon Kim and Soojoon Lee.
\newblock ``Relation between quantum coherence and quantum entanglement in
  quantum measurements''.
\newblock \href{https://dx.doi.org/10.1103/PhysRevA.106.022401}{Physical Review
  A {\bf 106}, 022401}~(2022).

\bibitem{kim2024maneuvering}
Ho-Joon Kim and Soojoon Lee.
\newblock ``Maneuvering measurement-coherence into
  measurement-entanglement''~(2024).

\bibitem{chitambar2019quantum}
Eric Chitambar and Gilad Gour.
\newblock ``Quantum resource theories''.
\newblock \href{https://dx.doi.org/10.1103/RevModPhys.91.025001}{Reviews of
  modern physics {\bf 91}, 025001}~(2019).

\bibitem{HFW21}
Masahito Hayashi, Kun Fang, and Kun Wang.
\newblock ``Finite block length analysis on quantum coherence distillation and
  incoherent randomness extraction''.
\newblock \href{https://dx.doi.org/10.1109/TIT.2021.3068357}{IEEE Transactions
  on Information Theory {\bf 67}, 3926--3944}~(2021).

\bibitem{korzekwa2019avoiding}
Kamil Korzekwa, Christopher~T Chubb, and Marco Tomamichel.
\newblock ``Avoiding irreversibility: Engineering resonant conversions of
  quantum resources''.
\newblock \href{https://dx.doi.org/10.1103/PhysRevLett.122.110403}{Physical
  Review Letters {\bf 122}, 110403}~(2019).

\bibitem{BG15}
Fernando~GSL Brandao and Gilad Gour.
\newblock ``Reversible framework for quantum resource theories''.
\newblock \href{https://dx.doi.org/10.1103/PhysRevLett.115.070503}{Physical
  review letters {\bf 115}, 070503}~(2015).

\bibitem{umegaki1959conditional}
Hisaharu Umegaki.
\newblock ``Conditional expectation in an operator algebra, iii''.
\newblock In Kodai Mathematical Seminar Reports.
\newblock \href{https://dx.doi.org/10.2996/kmj/1138844157}{Volume~11, pages
  51--64}.
\newblock Department of Mathematics, Tokyo Institute of Technology~(1959).

\bibitem{LR23}
Ludovico Lami and Bartosz Regula.
\newblock ``No second law of entanglement manipulation after all''.
\newblock \href{https://dx.doi.org/10.1038/s41567-022-01862-7}{Nature Physics
  {\bf 19}, 184--189}~(2023).

\bibitem{LRS23}
Ludovico Lami, Bartosz Regula, and Alexander Streltsov.
\newblock ``No-go theorem for entanglement distillation using catalysis''.
\newblock \href{https://dx.doi.org/10.1103/PhysRevA.109.L050401}{Physical
  Review A {\bf 109}, L050401}~(2024).

\bibitem{YZCM15}
Xiao Yuan, Hongyi Zhou, Zhu Cao, and Xiongfeng Ma.
\newblock ``Intrinsic randomness as a measure of quantum coherence''.
\newblock \href{https://dx.doi.org/10.1103/PhysRevA.92.022124}{Physical Review
  A {\bf 92}, 022124}~(2015).

\bibitem{YZGM19}
Xiao Yuan, Qi~Zhao, Davide Girolami, and Xiongfeng Ma.
\newblock ``Quantum coherence and intrinsic randomness''.
\newblock \href{https://dx.doi.org/10.1002/adv.201900053}{Advanced Quantum
  Technologies {\bf 2}, 1900053}~(2019).

\bibitem{wilde2013quantum}
Mark~M Wilde.
\newblock ``Quantum information theory''.
\newblock \href{https://dx.doi.org/10.1017/CBO9781139525343}{Cambridge
  University Press}. ~(2013).

\bibitem{chitambar2018dephasing}
Eric Chitambar.
\newblock ``Dephasing-covariant operations enable asymptotic reversibility of
  quantum resources''.
\newblock \href{https://dx.doi.org/10.1103/PhysRevA.97.050301}{Physical Review
  A {\bf 97}, 050301}~(2018).

\bibitem{brandao2010generalization}
Fernando~GSL Brandao and Martin~B Plenio.
\newblock ``A generalization of quantum stein's lemma''.
\newblock \href{https://dx.doi.org/10.1007/s00220-009-0946-7}{Communications in
  Mathematical Physics {\bf 295}, 791--828}~(2010).

\bibitem{berta2023gap}
Mario Berta, Fernando~GSL Brand{\~a}o, Gilad Gour, Ludovico Lami, Martin~B
  Plenio, Bartosz Regula, and Marco Tomamichel.
\newblock ``On a gap in the proof of the generalised quantum stein's lemma and
  its consequences for the reversibility of quantum resources''.
\newblock \href{https://dx.doi.org/10.22331/q-2023-5-10-1103}{Quantum {\bf 7},
  1103}~(2023).

\bibitem{tomamichel2015quantum}
Marco Tomamichel.
\newblock ``Quantum information processing with finite resources: mathematical
  foundations''.
\newblock \href{https://dx.doi.org/10.1007/978-3-319-18812-2}{Volume~5}.
\newblock Springer. ~(2015).

\bibitem{tomamichel2009fully}
Marco Tomamichel, Roger Colbeck, and Renato Renner.
\newblock ``A fully quantum asymptotic equipartition property''.
\newblock \href{https://dx.doi.org/10.1109/TIT.2009.2038501}{IEEE Transactions
  on information theory {\bf 55}, 5840--5847}~(2009).

\bibitem{zhao2018one}
Qi~Zhao, Yunchao Liu, Xiao Yuan, Eric Chitambar, and Xiongfeng Ma.
\newblock ``One-shot coherence dilution''.
\newblock \href{https://dx.doi.org/10.1103/PhysRevLett.120.070403}{Physical
  review letters {\bf 120}, 070403}~(2018).

\bibitem{uhlmann2010roofs}
Armin Uhlmann.
\newblock ``Roofs and convexity''.
\newblock \href{https://dx.doi.org/10.3390/e12071799}{Entropy {\bf 12},
  1799--1832}~(2010).

\bibitem{regula2017convex}
Bartosz Regula.
\newblock ``Convex geometry of quantum resource quantification''.
\newblock \href{https://dx.doi.org/10.1088/1751-8121/51/4/045303}{Journal of
  Physics A: Mathematical and Theoretical {\bf 51}, 045303}~(2017).

\bibitem{CH16}
Eric Chitambar and Min-Hsiu Hsieh.
\newblock ``Relating the resource theories of entanglement and quantum
  coherence''.
\newblock \href{https://dx.doi.org/10.1103/PhysRevLett.117.020402}{Physical
  review letters {\bf 117}, 020402}~(2016).

\bibitem{saxena2020dynamical}
Gaurav Saxena, Eric Chitambar, and Gilad Gour.
\newblock ``Dynamical resource theory of quantum coherence''.
\newblock \href{https://dx.doi.org/10.1103/PhysRevResearch.2.023298}{Physical
  Review Research {\bf 2}, 023298}~(2020).

\bibitem{takahashi2022creating}
Masaya Takahashi, Swapan Rana, and Alexander Streltsov.
\newblock ``Creating and destroying coherence with quantum channels''.
\newblock \href{https://dx.doi.org/10.1103/PhysRevA.105.L060401}{Physical
  Review A {\bf 105}, L060401}~(2022).

\bibitem{Buscemi09}
Francesco Buscemi.
\newblock ``Private quantum decoupling and secure disposal of information''.
\newblock \href{https://dx.doi.org/10.1088/1367-2630/11/12/123002}{New Journal
  of Physics {\bf 11}, 123002}~(2009).

\bibitem{DBWR14}
Fr{\'e}d{\'e}ric Dupuis, Mario Berta, J{\"u}rg Wullschleger, and Renato Renner.
\newblock ``One-shot decoupling''.
\newblock \href{https://dx.doi.org/10.1007/s00220-014-2038-2}{Communications in
  Mathematical Physics {\bf 328}, 251--284}~(2014).

\bibitem{FR15}
Omar Fawzi and Renato Renner.
\newblock ``Quantum conditional mutual information and approximate markov
  chains''.
\newblock \href{https://dx.doi.org/10.1007/s00220-015-2328-x}{Communications in
  Mathematical Physics {\bf 340}, 575--611}~(2015).

\bibitem{AJ19}
Anurag Anshu and Rahul Jain.
\newblock ``Efficient methods for one-shot quantum communication''.
\newblock \href{https://dx.doi.org/10.1038/s41534-022-00608-1}{npj Quantum
  Information {\bf 8}, 97}~(2022).

\bibitem{DW05}
Igor Devetak and Andreas Winter.
\newblock ``Distillation of secret key and entanglement from quantum states''.
\newblock \href{https://dx.doi.org/10.1098/rspa.2005.1552}{Proceedings of the
  Royal Society A: Mathematical, Physical and Engineering Sciences {\bf 461},
  207--235}~(2005).

\bibitem{RBC+18}
Martin Ringbauer, Thomas~R Bromley, Marco Cianciaruso, Ludovico Lami, WY~Sarah
  Lau, Gerardo Adesso, Andrew~G White, Alessandro Fedrizzi, and Marco Piani.
\newblock ``Certification and quantification of multilevel quantum coherence''.
\newblock \href{https://dx.doi.org/10.1103/PhysRevX.8.041007}{Physical Review X
  {\bf 8}, 041007}~(2018).

\end{thebibliography}
\end{document}